\documentclass[onecolumn,draftclsnofoot]{IEEEtran}
\usepackage{indentfirst}
\usepackage[dvips]{graphicx}
\usepackage{amsfonts}
\usepackage{multirow}
\usepackage{amsmath,amsthm,amssymb}
\usepackage{color,changepage}
\usepackage{algorithm}
\usepackage{algorithmic}
\usepackage{bbm}
\usepackage{verbatim}
\usepackage{makecell}
\usepackage{subfigure}
\usepackage{mathrsfs}
\usepackage{arydshln}
\usepackage{cases}
\usepackage{threeparttable}
\usepackage{extarrows}
\usepackage{array}
\usepackage{bm}
\usepackage{pifont}
\usepackage[T1]{fontenc}
\usepackage{threeparttable}
\usepackage{booktabs}
\usepackage{xr}
\usepackage{hyperref}
\usepackage{cite}
\usepackage{diagbox}
\usepackage{tabularx}
\allowdisplaybreaks[4]

\definecolor{newcolor}{rgb}{0.5,0,1}

\newcommand{\rightt}{\ding{52}}
\newcommand{\cmarkk}{\ding{56}}

\newcolumntype{I}{!{\vrule width 2pt}}

\newtheorem{Theorem}{Theorem}

\newtheorem{Definition}{Definition}

\newtheorem{Lemma}{Lemma}

\theoremstyle{remark}
\newtheorem{Remark}{Remark}
\newtheorem{Example}{Example}

\DeclareMathAlphabet{\mathpzc}{OT1}{pzc}{m}{it}

\definecolor{newcolor}{rgb}{0.5,0,1}

\def\BibTeX{{\rm B\kern-.05em{\sc i\kern-.025em b}\kern-.08em
    T\kern-.1667em\lower.7ex\hbox{E}\kern-.125emX}}

\let\emptyset\varnothing

\begin{document}

\title{A General Coding Framework for Adaptive Private Information Retrieval}
\author{Jinbao Zhu and Xiaohu Tang
\thanks{J. Zhu and X. Tang are 
with the Information Coding and Transmission (ICT) Key Laboratory of Sichuan Province,
Southwest Jiaotong University, Chengdu 611756, China
(email: jinbaozhu@swjtu.edu.cn, xhutang@swjtu.edu.cn).}}
\maketitle

\begin{abstract}
The problem of $T$-colluding private information retrieval (PIR) enables the user to retrieve one out of $M$ files from a distributed storage system with $N$ servers without revealing anything about the index of the desired file to any group of up to $T$ colluding servers. In the considered storage system, the $M$ files are stored across the $N$ distributed servers in an $X$-secure $K$-coded manner such that any group of up to $X$ colluding servers learns nothing about the files; the storage overhead at each server is reduced by a factor of $\frac{1}{K}$ compared to the total size of the files; and the files can be reconstructed from any $K+X$ servers.
However, in practical scenarios, when the user retrieves the desired file from the distributed system, some servers may respond to the user very slowly or not respond at all. These servers are referred to as \emph{stragglers}, and particularly their identities and numbers are unknown in advance and may change over time. 
This paper considers the adaptive PIR problem that can be capable of tolerating the presence of a varying number of stragglers. 
We propose a general coding method for designing adaptive PIR schemes by introducing the concept of a \emph{feasible PIR coding framework}. We demonstrate that any \emph{feasible PIR coding framework} over a finite field $\mathbb{F}_q$ with size $q$ can be used to construct an adaptive PIR scheme that achieves a retrieval rate of $1-\frac{K+X+T-1}{N-S}$ simultaneously for all numbers of stragglers $0\leq S\leq N-(K+X+T)$ over the same finite field.
Additionally, we provide an implementation of the \emph{feasible PIR coding framework}, ensuring that the adaptive PIR scheme operates over any finite field $\mathbb{F}_q$ with size $q\geq N+\max\{K, N-(K+X+T-1)\}$.
\end{abstract}

\section{Introduction}\label{Introduction}
In the current technological era, the rapid advancement of emerging technologies has created immense potential 
while also introducing new ethical and moral challenges. Among these, safeguarding data privacy and security remains a pressing concern.
Motivated by the demand to protect the privacy of downloaded information from public servers, private information retrieval (PIR) was introduced by Chor \emph{et al.} in 1995 \cite{ChorFOCS,chor1998private} and has garnered widespread attention over the past few decades. 

In the standard information-theoretic setting, PIR allows a user to retrieve a desired file from a set of $M$ files stored across $N$ distributed servers, while revealing nothing to any individual server about which file is being retrieved. To this end, the user sends private queries to each server, which generates corresponding responses based on the available data. The user then reconstructs the desired file from the collected responses. As in the information-theoretic formulation, the size of files is allowed to be arbitrary large, and the queries are typically designed without scaling with the file size \cite{sun2017capacity,sun2017optimal,tian2019capacity}. Therefore, the upload cost for sending queries is negligible compared to the download cost from the servers. The efficiency of a PIR scheme is measured in terms of the retrieval rate, defined as the ratio of the file size to the download cost from the servers. In other words, the retrieval rate means the number of bits that the user can privately retrieve per bit of downloaded data, and thus it should be maximized. 
The maximum achievable retrieval rate was first characterized as $\big(1+\frac{1}{N}+\ldots+\frac{1}{N^{M-1}}\big)^{-1}$ by Sun and Jafar in \cite{sun2017capacity} for the PIR problem with replicated storage. Since then, numerous intriguing variants and extensions of the fundamental PIR problem have been investigated to characterize their information-theoretic limits. 
We refer readers to \cite{ulukus2022private,vithana2023privateReview,d2024guided} for a more 
comprehensive survey. 

In practical scenarios, due to reasons such as server failures and the dynamic network environment, there are some servers that may respond to the user's queries very slowly or not respond at all, causing significant delays for the user in retrieving the desired file.
We refer to these servers as \emph{stragglers} \cite{dean2013tail,yadwadkar2016multi
}, particularly whose identities and numbers are unknown in advance and may change over time. While the problem of PIR in the presence of stragglers has been extensively studied in 
\cite{sun2017capacityColluding,tajeddine2019private,jia2020x,zhu2022multi,jia2022xFSL}, they all share a restrictive assumption that the user needs to know the exact number of stragglers in advance. 
Indeed, a PIR scheme can be designed for the maximal number of stragglers, but it does not achieve high retrieval rate when there are less stragglers. 
This motivates the study of the adaptive PIR problem, which aims to tolerate the presence of stragglers while simultaneously achieving a high retrieval rate based on the actual number of stragglers, even if their identities and numbers are unknown in advance and may change over time.
In the adaptive PIR problem, the user sends private queries to each server. Each server keeps computing and sending responses to the user until it receives sufficient responses from a varying number of the fastest servers to decode the desired file.

The adaptive PIR problem was initially explored by Tajeddine and Rouayheb in \cite{tajeddine2017robust} for coded storage systems, where the files are stored across the distributed servers according to $(N,K)$ MDS codes over a finite field $\mathbb{F}_q$ with size $q$. The adaptive PIR scheme presented in \cite{tajeddine2017robust} achieves the asymptotically optimal retrieval rate of $1-\frac{K}{N-S}$ (as the number of files $M$ approaches infinity) over the same finite field in which the MDS storage codes are defined, for all numbers of stragglers $0\leq S\leq N-K-1$, under the assumption that the identities of stragglers remain unchanged over time. 
Soon afterwards, the $T$-colluding adaptive PIR problem with replicated storage was considered by Bitar and Rouayheb in \cite{bitar2018staircase}, where any group of up to $T$ colluding servers may collude to deduce the index of desired file. 
They proposed an adaptive PIR scheme that achieves the asymptotically optimal retrieval rate of $1-\frac{T}{N-S}$ over the finite field $\mathbb{F}_q$ with size $q>N$, for all numbers of stragglers $0\leq S\leq N-T-1$, by utilizing the staircase secret sharing codes \cite{bitar2017staircase}. Furthermore, Jia and Jafar investigated the $T$-colluding PIR problem with $X$-secure $K$-coded storage in \cite{jia2020x}, where the $M$ files are stored across the $N$ servers such that any group of up to $X$ colluding servers learn nothing about the files; the storage overhead at each server is reduced by a factor of $\frac{1}{K}$ compared to the total size of the files; and the files can be reconstructed from any $K+X$ servers. 
They proposed a PIR scheme based on the idea of cross-subspace alignment \cite{jia2019cross} under the assumption that the number of stragglers is known in advance.
This scheme achieves a retrieval rate of $1-\frac{K+X+T-1}{N-S}$ over the finite field $\mathbb{F}_q$ with size $q\geq 2N-(K+X+T-1)$, for a known number of stragglers $S$ with $0\leq S\leq N-(K+X+T)$. 
Recently, the problem of federated submodel learning (FSL) has garnered attention \cite{jia2022xFSL,vithana2023private}. The FSL problem involves two distinct phases: the reading phase and the writing phase, with the reading phase being equivalently viewed as the PIR problem.
The problem of $T$-colluding FSL with $X$-secure $K$-coded storage was considered by Jia and Jafar in \cite{jia2022xFSL}. 
In the proposed FSL scheme \cite{jia2022xFSL}, the reading phase introduces a PIR scheme that achieves a retrieval rate of $1-\frac{K+X+T-1}{N-S}$ over the finite field $\mathbb{F}_q$ with size $q\geq N+\max\{K, N-(K+X+T-1)\}$ simultaneously for all numbers of stragglers $0\leq S\leq N-(K+X+T)$. However, it implicitly requires the user to know the actual number of stragglers before retrieving the desired file.

In this paper, we revisit the $T$-colluding PIR problem from an $X$-secure $K$-coded storage system. 
As far as we know, the current best retrieval rate for this PIR problem is $1-\frac{K+X+T-1}{N-S}$ in the presence of a known number of stragglers $S$ with $0\leq S\leq N-(K+X+T)$. This retrieval rate is also conjectured to be asymptotically optimal by Jia and Jafar \cite{jia2020x,jia2022xFSL}.
Given the more practical scenario where the identities and numbers of stragglers are unknown in advance and may change over time, we are interested in investigating the following question.

\noindent\textbf{\emph{Would it be possible to design an efficient adaptive PIR scheme that achieves a retrieval rate of $\mathbf{1-\frac{K+X+T-1}{N-S}}$ simultaneously for all numbers of stragglers $\mathbf{0\leq S\leq N-(K+X+T)}$ over a small finite field $\mathbf{\mathbb{F}_q}$?}}

We answer the question in the affirmative, and propose a general coding framework to design this efficient adaptive PIR scheme. 
The main contributions of this paper include the following two aspects.
\begin{enumerate}
\item We construct a general adaptive PIR scheme by introducing the concept of a \emph{feasible PIR coding framework}, demonstrating that any \emph{feasible PIR coding framework} over a finite field $\mathbb{F}_q$ of size $q$ can be utilized to construct an adaptive PIR scheme that achieves a retrieval rate of $1-\frac{K+X+T-1}{N-S}$ simultaneously for all numbers of stragglers $0\leq S\leq N-(K+X+T)$ over the same finite field $\mathbb{F}_q$.
\item We implement the \emph{feasible PIR coding framework} over the finite field $\mathbb{F}_q$ for any prime power $q$ of size $q\geq N+\max\{K,N-(K+X+T-1)\}$.
This framework is associated with an adaptive PIR scheme that achieves a retrieval rate of $1-\frac{K+X+T-1}{N-S}$ simultaneously for all numbers of stragglers $0\leq S\leq N-(K+X+T)$ over the finite field $\mathbb{F}_q$ of size $q\geq N+\max\{K,N-(K+X+T-1)\}$.
\end{enumerate}

In general, the adaptive PIR problem considered in this paper generalizes previously known related problems. It includes the adaptive PIR problem with coded storage \cite{tajeddine2017robust} and the $T$-colluding adaptive PIR problem \cite{bitar2018staircase} as special cases by setting $X=0,T=1$ and $X=0,K=1$, respectively. Additionally, it extends the $T$-colluding PIR problem with an $X$-secure $K$-coded storage system \cite{jia2020x,jia2022xFSL} to the more practical adaptive scenario. Table \ref{Model:comparision} provides a detailed comparison with these related PIR schemes.

\begin{table*}[htbp!]
\centering
\caption{Comparison of the adaptive PIR scheme proposed in this paper with previously related PIR schemes} \label{Model:comparision}
\begin{threeparttable}
  \begin{tabular}{|@{}c@{}|@{}c@{}|@{}c@{}|@{}c@{}|@{}c@{}|@{}c@{}|@{}c@{}|@{}c@{}|}
  \hline
\multirow{2}{*}{\diagbox[width=12.22em]{PIR Schemes}{PIR Settings}}   & \multirow{2}{*}{\shortstack{MDS Coded \\ Storage over $\mathbb{F}_q$}} & \multirow{2}{*}{\shortstack{$T$-Colluding\\ Privacy}} & \multirow{2}{*}{\shortstack{$X$-Secure\\ Storage}} & \multicolumn{2}{c|}{Straggler Scenarios} & \multirow{2}{*}{Retrieval Rate}  & \multirow{2}{*}{Finite Field Size} \\ \Xcline{5-6}{0.5pt}
& & & & Variable Identities$^{\dag}$  & Variable Number$^{\ddag}$ & & \\\hline
  Tajeddine-Rouayheb PIR \cite{tajeddine2017robust} & \rightt & \cmarkk & \cmarkk & \cmarkk & \rightt & $1-\frac{K}{N-S}$ & $q$ \\ \hline
  Bitar-Rouayheb PIR \cite{bitar2018staircase} & \cmarkk & \rightt & \cmarkk & \rightt & \rightt & $1-\frac{T}{N-S}$ & $N+1$ \\ \hline
  Jia-Jafar PIR \cite{jia2020x} & \rightt & \rightt & \rightt & \rightt & \cmarkk & $1-\frac{K+X+T-1}{N-S}$ & $N+\lambda^\S$ \\ \hline
 Jia-Jafar FSL \cite{jia2022xFSL} & \rightt & \rightt & \rightt & \rightt & \cmarkk & $1-\frac{K+X+T-1}{N-S}$ & $N+\max\{K, \lambda\}$ \\ \hline
  Our Adaptive PIR & \rightt & \rightt & \rightt & \rightt & \rightt & $1-\frac{K+X+T-1}{N-S}$ & $N+\max\{K, \lambda\}$ \\ \hline
\end{tabular}
\begin{tablenotes}
\footnotesize
\item[$^{\dag}$]Stragglers with variable identities mean that the identities of stragglers can change over time during the execution of the PIR scheme.
\item[$^{\ddag}$]A variable number of stragglers means that the number of stragglers can change over time during the execution of the PIR scheme.
\item[$^\S$]$\lambda=N-(K+X+T-1)$.
\end{tablenotes}
\end{threeparttable}
\end{table*}

The works most relevant to ours are \cite{jia2020x} and \cite{jia2022xFSL}, which consider the same $T$-colluding PIR problem with an $X$-secure $K$-coded storage system in the presence of $S$ stragglers for $0\leq S\leq N-(K+X+T)$. To enhance communication efficiency, reference \cite{jia2020x} partitions each file into $(N-S-(K+X+T-1)) \times K$ subfiles and proposes a PIR scheme that achieves a retrieval rate of $1-\frac{K+X+T-1}{N-S}$ over a finite field with size $q\geq 2N-(K+X+T-1)$, under the assumption of a known number of stragglers $S$ with $0\leq S\leq N-(K+X+T)$.
Compared to \cite{jia2020x}, both the PIR scheme in \cite{jia2022xFSL} and ours achieve the same retrieval rate of $1-\frac{K+X+T-1}{N-S}$ over a finite field of size $q\geq N+\max\{K,N-(K+X+T-1)\}$, simultaneously for all numbers of stragglers $0\leq S\leq N-(K+X+T)$. 
The most important difference is that reference \cite{jia2022xFSL} implicitly assumes that the actual number of stragglers is known before retrieving the desired file, whereas our scheme imposes no such restriction. 
More specifically, to simultaneously tolerate all numbers of stragglers $0\leq S\leq N-(K+X+T)$, reference \cite{jia2022xFSL} constructs the PIR scheme by further partitioning each file into $\mathrm{lcm}(1,2,\ldots,N-(K+X+T-1))\times K$ subfiles and then repeating the $S$-straggler scheme from \cite{jia2020x} $\frac{\mathrm{lcm}(1,2,\ldots,N-(K+X+T-1))}{N-S-(K+X+T-1)}$ times. The feasibility of this PIR scheme fundamentally relies on the prior knowledge of the actual number of stragglers $S$.
However, when the number of stragglers is unknown in advance, the core challenge lies in how to fully exploit the responses from all non-straggler servers to improve communication efficiency.
To address this, we design the adaptive PIR scheme by introducing a query array consisting of $N-(K+X+T-1)$ subarrays. The $0$-th subarray is used to design private queries for retrieving the desired file in the presence of $S=0$ straggler. For any given $h=1,2,\ldots,N-(K+X+T)$, the $h$-th subarray is used to design private queries to compensate for the missing responses required to decode the first $h$ subarray in the presence of $S=h$ stragglers. This ensures that all the responses associated with the $0$-th, $1$-th,$\ldots$, and $S$-th subarrays can be exploited to decode the desired file associated with the $0$-th subarray, even if the actual number of stragglers $S$ is unknown simultaneously for all $0\leq S\leq N-(K+X+T)$. In particular, by carefully designing the elements of the query array, our adaptive PIR scheme operates over the same finite field size as in \cite{jia2022xFSL}, demonstrating that providing an adaptive guarantee does not incur any increase in the finite field size.


The remainder of this paper is organized as follows. In Section \ref{Problem:formulation}, we formally formulate the adaptive PIR problem.  
In Section \ref{Section:PIR:framework}, we describe the \emph{feasible PIR coding framework} and provide its implementation.
In Section \ref{scheme:APIR}, we establish a connection between the \emph{feasible PIR coding framework} and the adaptive PIR scheme by showing that any \emph{feasible PIR coding framework} can be utilized to construct an adaptive PIR scheme.
Finally, this paper is concluded in Section \ref{conclusion}.

\subsubsection*{Notation}Throughout this paper, instead of the more conventional indexing of starting at $1$, all indices start from $0$, including both matrix and array indices.
Let boldface capital letters represent matrices or arrays, e.g., $\mathbf{W}$ and $\mathbf{U}$. 
Let cursive capital letters denote sets, such as $\mathcal{S}$, while $|\mathcal{S}|$ denotes the cardinality of the set $\mathcal{S}$.
For any non-negative integers $m,n$ such that $m\leq n$, $[m,n)$ and $[m:n]$ denote the sets $\{m,m+1,\ldots,n-1\}$ and $\{m,m+1,\ldots,n\}$, respectively.
Define $\mathcal{A}_{\mathcal{S}}$ as $\{\mathcal{A}_{s_0},\mathcal{A}_{s_1},\ldots,\mathcal{A}_{s_{n-1}}\}$ for any index set $\mathcal{S}=\{s_0,s_1,\ldots,s_{n-1}\}$ of size $n$.
For a matrix $\mathbf{W}$ of dimensions $m\times n$, the notation $\mathbf{W}(\mathcal{S},:)$ represents the submatrix of $\mathbf{W}$ formed by retaining only the rows corresponding to the elements in the set $\mathcal{S}$ for any subset $\mathcal{S}\subseteq[0:m)$.
For an array $\mathbf{U}$ of dimensions $m\times n$, the notation $u_{i,j}$ represents 
the element in the $i$-th row and $j$-th column of the array $\mathbf{U}$ for any given $i\in[0:m),j\in[0:n)$.
Denote by $(\cdot)_n$ the modulo $n$ operation.

\section{Problem Formulation}\label{Problem:formulation}
Consider a dataset containing $M$ independent files $\mathbf{W}^{(0)},\mathbf{W}^{(1)},\ldots,\mathbf{W}^{(M-1)}$, each of which consists of $L$ symbols drawn independently and uniformly from a finite field $\mathbb{F}_q$ for a prime power $q$, i.e.,
\begin{IEEEeqnarray*}{rCl}
H(\mathbf{W}^{(0)},\ldots,\mathbf{W}^{(M-1)})&=&\sum_{m=0}^{M-1}H(\mathbf{W}^{(m)}), \label{model:file inden} \\
H(\mathbf{W}^{(0)})&=&\ldots=H(\mathbf{W}^{(M-1)})=L, \label{infor indenpe}
\end{IEEEeqnarray*}
where the entropy function $H(\cdot)$ is measured with logarithm $q$.

The dataset is secretly and reliably stored in a distributed storage system with $N$ servers, while minimizing the storage overhead on each server as much as possible.
Denote the random variables stored at server $n$ by $\mathcal{Y}_n$ for all $n\in[0:N)$. Formally, the storage system needs to satisfy the following constraints.\footnote{The storage schemes that satisfy the constraints in \eqref{security}-\eqref{reconstruction} belong to a class of widely studied non-perfect secret sharing schemes known as \emph{ramp} schemes \cite{blakley1985security,jackson1996combinatorial}. It has been proven in \cite{capocelli1993size,csirmaz1997size} that any perfect secret sharing scheme \cite{shamir1979share,blakley1979safeguarding} requires the storage overhead at each server to be at least as large as the total size of the dataset. The primary motivation for defining ramp schemes is to reduce the storage overhead. A ramp scheme is considered \emph{optimal} if it meets the storage constraint in \eqref{storage:head} \cite{jackson1996combinatorial}. In general, we focus on the distributed storage system where the dataset is stored in an optimal ramp manner.
}
\begin{itemize}
    \item\textbf{Secrecy:} Any $X$ colluding servers learn nothing about the dataset. This ensures that the dataset is secure even if an eavesdropper listens to up to $X$ servers.
    \begin{IEEEeqnarray}{rCl}\label{security}
    I(\mathbf{W}^{(0)},\ldots,\mathbf{W}^{(M-1)};\mathcal{Y}_{\mathcal{X}})=0, \quad\forall\,\mathcal{X}\subseteq[0:N),|\mathcal{X}|\leq X. \IEEEeqnarraynumspace
    \end{IEEEeqnarray}
    \item \textbf{Reliability:} The dataset can be reconstructed by connecting to at least $K+X$ out of the $N$ servers. This provides data reliability against up to $N-(K+X)$ server failures.
    \begin{IEEEeqnarray}{c}\label{reconstruction}
    H(\mathbf{W}^{(0)},\ldots,\mathbf{W}^{(M-1)}|\mathcal{Y}_{\mathcal{K}})=0,\quad\forall\,\mathcal{K}\subseteq[0:N),|\mathcal{K}|\geq K+X.
    \end{IEEEeqnarray}
    \item \textbf{Storage Overhead:}  The storage overhead at each server is constrained to ${H(\mathbf{W}^{(0)},\ldots,\mathbf{W}^{(M-1)})}/{K}$, which is reduced by a factor of ${1}/{K}$ compared to the total size of the dataset. 
    \begin{IEEEeqnarray}{rCl}\label{storage:head}
    H(\mathcal{Y}_n)=\frac{H(\mathbf{W}^{(0)},\ldots,\mathbf{W}^{(M-1)})}{K}=\frac{ML}{K}, \quad\forall\,n\in[0:N).
    \end{IEEEeqnarray}
\end{itemize}
Clearly, the storage system includes $(N,K)$ MDS coded storage and $(N,X)$ perfect threshold secret share storage as special cases, obtained by setting $X=0$ and $K=1$, respectively.

In the private information retrieval (PIR) problem, the user privately selects an index $\theta$ from the range $0$ to $M-1$, and wishes to retrieve the file $\mathbf{W}^{(\theta)}$ from the distributed system without revealing any information about the index to any group of up to $T$ colluding servers. For this purpose, the user sends private queries to each server, and the sever accordingly generates responses based on the information it stores.
The user waits for the responses from the servers until it can decode the desired file. 
However, in practical scenarios, there are some straggler servers that may respond to the user very slowly or not respond at all, and particularly their identities and numbers are unknown in advance and may change over time. 

It would be preferable for a PIR scheme to \emph{adaptively} tolerate the presence of a varying number of stragglers, meaning that the user does not need to know the presence of stragglers a priori and  only waits for a sufficient number of responses from a varying number of fastest servers in order to recover the desired file. We formally present the definition of adaptive PIR schemes as follows.
\begin{Definition}[Adaptive Private Information Retrieval]
A PIR scheme is said to be adaptive if it can
tolerate the presence of a varying number of stragglers such that the user can privately retrieve the desired file $\mathbf{W}^{(\theta)}$ even in the presence of $S$ stragglers simultaneously for all $0\leq S\leq S_{\max}$, where $S_{\max}$ represents the maximum number of tolerated stragglers for the adaptive PIR scheme. More specifically, an adaptive PIR scheme comprises the following three phases.
\begin{enumerate}
\item \emph{Query Phase:} The user sends the $F$ queries $\mathcal{Q}^{(\theta)}_{n,0},\mathcal{Q}^{(\theta)}_{n,1},\ldots,\mathcal{Q}^{(\theta)}_{n,F-1}$ to each server $n$ for all $n\in[0:N)$. Indeed, all the queries are generated by the user without any prior knowledge of the data content stored at the servers, i.e.,
\begin{IEEEeqnarray}{c}\notag
    I(\mathcal{Y}_{[0:N)};\mathcal{Q}^{(\theta)}_{[0:N),[0:F)})=0.
\end{IEEEeqnarray}
\item \emph{Answer Phase:} Upon receiving the $F$ queries $\mathcal{Q}^{(\theta)}_{n,0},\mathcal{Q}^{(\theta)}_{n,1},\ldots,\mathcal{Q}^{(\theta)}_{n,F-1}$, the server $n$ sequentially generates the $F$ responses $\mathcal{A}_{n,0}^{(\theta)},\mathcal{A}_{n,1}^{(\theta)},\ldots,\mathcal{A}_{n,F-1}^{(\theta)}$ in the given order and sends $\mathcal{A}_{n,\ell}^{(\theta)}$ back to the user once its computation is completed for any $\ell\in[0:F)$. These responses are a determined function of the queries and the data stored at server $n$, i.e.,
\begin{IEEEeqnarray}{c}\notag
H(\mathcal{A}_{n,[0:F)}^{(\theta)}|\mathcal{Q}_{n,[0:F)}^{(\theta)},\mathcal{Y}_n)=0,\quad \forall\,n\in[0:N).
\end{IEEEeqnarray}
\item \emph{Decoding Phase:} In the presence of $S$ stragglers for all $0\leq S\leq S_{\max}$, the user recovers the file $\mathbf{W}^{(\theta)}$ after waiting for the first $F_S$ responses from each of the fastest $N-S$ servers. It is natural to require that  $1\leq F_0\leq F_1\leq\ldots\leq F_{S_{\max}}\leq F$ to enhance the communication efficiency of the adaptive PIR scheme.
\end{enumerate}
\end{Definition}
In general, the user sends $F$ queries to each server. Each server continues to calculate and sequentially sends $F$ responses to the user. The user can decode the desired file as long as each of the fastest $N-S$ servers successfully returns the first $F_S$ responses for all $0\leq S\leq S_{\max}$. 
This enables the PIR scheme to adaptively tolerate the presence of $S$ stragglers simultaneously for all $0\leq S\leq S_{\max}$.


The adaptive PIR scheme must satisfy the following two constraints.
\begin{itemize}
    \item \textbf{Correctness:} The user must be able to correctly decode $\mathbf{W}^{(\theta)}$ 
    after collecting the first $F_S$ responses from each of the fastest $N-S$ servers for all $0\leq S\leq S_{\max}$, i.e.,
\begin{IEEEeqnarray}{c}\notag
H(\mathbf{W}^{(\theta)}|\mathcal{A}_{[0:N)\backslash\mathcal{S},[0:F_S)}^{(\theta)},\mathcal{Q}_{[0:N),[0:F)}^{(\theta)})=0,\quad\forall\,\mathcal{S}\subseteq[0:N),|\mathcal{S}|=S.
\end{IEEEeqnarray}
    \item \textbf{Privacy:} The strategies for retrieving any two files $\mathbf{W}^{(\theta)}$ and $\mathbf{W}^{(\theta')}$ must be indistinguishable from the perspective of any up to $T$ colluding servers, i.e.,
\begin{IEEEeqnarray}{c}
(\mathcal{Q}_{\mathcal{T},[0:F)}^{(\theta)},\mathcal{A}_{\mathcal{T},[0:F)}^{(\theta)},\mathcal{Y}_{\mathcal{T}})\sim (\mathcal{Q}_{\mathcal{T},[0:F)}^{(\theta')},\mathcal{A}_{\mathcal{T},[0:F)}^{(\theta')},\mathcal{Y}_{\mathcal{T}}),\quad\forall\, \theta,\theta'\in[0:M), \forall\, \mathcal{T}\subseteq[0:N), |\mathcal{T}|\leq T,\notag
\end{IEEEeqnarray}
where $X\sim Y$ indicates that the random variables $X$ and $Y$ have identical distributions. Equivalently, all possible  information available to any $T$ colluding servers is independent of the index $\theta$ of the desired file, i.e.,
\begin{IEEEeqnarray}{c}\label{Infor:priva cons}
I(\theta;\mathcal{Q}_{\mathcal{T},[0:F)}^{(\theta)},\mathcal{A}_{\mathcal{T},[0:F)}^{(\theta)},\mathcal{Y}_{\mathcal{T}})=0, \quad \forall\, \mathcal{T}\subseteq[0:N), |\mathcal{T}|\leq T.
\end{IEEEeqnarray}
\end{itemize}

The performance of an adaptive PIR scheme is evaluated by the following two metrics.
\begin{enumerate}
  \item[1.] The retrieval rate $R_S$, which is the number of desired bits that the user can privately retrieve per downloaded bit in the presence of $S$ stragglers for all $0\leq S\leq S_{\max}$, defined as\footnote{To simplify the definition of retrieval rate, we assume that straggler servers will not provide any information, while at the same time the remaining non-straggler servers have identical response capabilities (i.e., if one non-straggler server sends a certain number of responses back to the user, then another non-straggler server will also send the same number of responses).
  } 
  \begin{IEEEeqnarray}{c}\label{def:PPC rate}
  R_S\triangleq\frac{H(\mathbf{W}^{(\theta)})}{\max\limits _{\mathcal{S}\subseteq[0:N),|\mathcal{S}|=S}\sum\limits_{n\in[0:N)\backslash\mathcal{S}}\sum\limits_{\ell\in[0:F_S)}H(\mathcal{A}_{n,\ell}^{(\theta)})}=\frac{L}{D_S},
  \end{IEEEeqnarray}
  where $D_S=\max\limits_{\mathcal{S}\subseteq[0:N),|\mathcal{S}|=S}\sum\limits_{n\in[0:N)\backslash\mathcal{S}}\sum\limits_{\ell\in[0:F_S)}H(\mathcal{A}_{n,\ell}^{(\theta)})$ is the total download cost from the  servers $[0:N)\backslash\mathcal{S}$ in the presence of $S$ stragglers, maximized over $\mathcal{S}\subseteq[0:N)$ with $|\mathcal{S}|=S$. 
  \item[2.] The finite field size $q$, which ensures the achievability of the 
  storage system and the adaptive PIR scheme.
\end{enumerate}

The objective of this paper is to design adaptive PIR schemes that maximize the retrieval rate $R_S$ simultaneously for all $0\leq S\leq S_{\max}$ while minimizing the finite field size $q$.

\section{Feasible PIR Coding Framework and Its Implementation}\label{Section:PIR:framework}
In this section, we present a \emph{feasible PIR coding framework} and provide its implementation. 
The feasible PIR coding framework serves as a building block for constructing the adaptive PIR scheme. To establish this connection, both the PIR coding framework and the adaptive PIR adopt the same notation, with the only difference being that the notation in the PIR framework is marked with a tilde ($\sim$). This notation convention highlights that, for an encoding function using the same notation in both the adaptive PIR and the feasible PIR framework, the encoding function in the adaptive PIR can be constructed in a manner analogous to its counterpart in the feasible PIR framework. 
For clarity, the main notations used in the feasible PIR coding framework and the adaptive PIR scheme are listed in Table \ref{tab:parameters}, along with accompanying explanations that further illustrate how the feasible PIR framework can be applied to construct the adaptive PIR scheme.

\begin{table*}[htbp]
\extrarowheight=4pt
\centering
\caption{Notations and Explanations in the Feasible PIR Coding Framework and the Adaptive PIR}
\begin{tabular}{|@{\;}c@{\;}|@{\;}c@{\;}|l@{\;}||@{\;}c@{\;}|l@{\;}|}
\hline  
\multirowcell{2}[0pt][c]{\textbf{PIR Phase}} & \multicolumn{2}{c||}{\textbf{Feasible PIR Coding Framework}} &  \multicolumn{2}{c|}{\textbf{Adaptive PIR}}   \\ \Xcline{2-5}{0.5pt} 
 & \textbf{Notation} & \textbf{Explanation} & \textbf{Notation} & \textbf{Explanation} \\ \hline 
\multirowcell{7}[0pt][c]{Storage} & \multirowcell{2}[0pt][c]{$\widetilde{\mathbf{W}}^{(m)}$} & \multirowcell{2}[0pt][l]{the $m$-th file of dimensions $\lambda\times K$ \\  for $m\in[0:M)$} & \multirowcell{2}[0pt][c]{$\mathbf{W}^{(m)}$} & \multirowcell{2}[0pt][l]{the $m$-th file of dimensions $P\times K$ for $m\in[0:M)$} \\   
&  & & & \\ \Xcline{2-5}{0.5pt}
&  \multirowcell{3}[0pt][c]{$\tilde{f}_i^{(m)}(x)$} & \multirowcell{3}[0pt][l]{storage function used to encode \\ the $i$-th row of  the $m$-th file for \\ $m\in[0:M)$ and $i\in[0:\lambda)$}
  & \multirowcell{3}[0pt][c]{$f_i^{(m)}(x)$} & \multirowcell{3}[0pt][l]{storage function used to encode the $i$-th row of \\ the $m$-th file for $m\in[0:M)$ and $i\in[0:P)$, \\ designed in a manner similar to $\tilde{f}_{(i)_{\lambda}}^{(m)}(x)$}  \\
&  & & & \\  
& & & & \\ \Xcline{2-5}{0.5pt}
& \multirowcell{2}[0pt][c]{$\widetilde{\mathcal{Y}}_n$} & \multirowcell{2}[0pt][l]{encoded data stored at server $n$ \\ for $n\in[0:N)$} & \multirowcell{2}[0pt][c]{$\mathcal{Y}_n$} & \multirowcell{2}[0pt][l]{encoded data stored at server $n$ for $n\in[0:N)$, \\ designed in a manner similar to $\widetilde{\mathcal{Y}}_n$} \\ 
&  & & & \\ \hline
\multirowcell{12}[0pt][c]{Query} &  \multirowcell{6}[0pt][c]{$\mathcal{R}$} & \multirowcell{6}[0pt][l]{row index of the desired partial file \\ for any nonempty subset $\mathcal{R}\!\subseteq\![0\!:\!\lambda)$}   & $\mathbf{U}$ & \multirowcell{1}[0pt][l]{query array consisting of $\lambda$ subarrays} \\
  \Xcline{4-5}{0.5pt}
 & & & $\mathbf{U}^{h}$ & \multirowcell{1}[0pt][l]{the $h$-th query subarray of dimensions $\lambda\times\Gamma^h$} \\ \Xcline{4-5}{0.5pt}
 & & & \multirowcell{2}[0pt][c]{$\mathcal{R}_j^{h}$} & \multirowcell{2}[0pt][l]{the set consisting of the integer elements in the $j$-th \\  column of the subarray $\mathbf{U}^{h}$ for $h\!\in\![0\!:\!\lambda),j\!\in\![0\!:\!\Gamma^h)$} \\
 & & & & \\ \Xcline{4-5}{0.5pt}
 & & & \multirowcell{2}[0pt][c]{$\widetilde{\mathcal{R}}^{h}_j$} & \multirowcell{2}[0pt][l]{the set consisting of the integer elements \\ in $\mathcal{R}^{h}_j$ modulo $\lambda$} \\
 & & & & \\ \Xcline{2-5}{0.5pt} 
 & \multirowcell{3}[0pt][c]{$\widetilde{q}_{i,k}^{(m,\mathcal{R})}(x)$}  & \multirowcell{3}[0pt][l]{query function used to retrieve \\ the $k$-th column of the partial \\ file $\widetilde{\mathbf{W}}^{(\theta)}(\mathcal{R},:)$ for $k\!\in\![0\!:\!K)$} & \multirowcell{3}[0pt][c]{$q_{i,k}^{(m,\mathcal{R}_j^h)}(x)$} & \multirowcell{3}[0pt][l]{query function used to retrieve the $k$-th column \\ of the partial file $\mathbf{W}^{(\theta)}(\mathcal{R}_j^h,:)$, designed in a \\ encoding manner similar to $\widetilde{q}_{(i)_{\lambda},k}^{(m,\widetilde{\mathcal{R}}_j^h)}(x)$} \\
 & & & & \\ & & & & \\   \Xcline{2-5}{0.5pt}
 & \multirowcell{3}[0pt][c]{$\widetilde{\mathcal{Q}}_n^{(\theta,\mathcal{R})}$} & \multirowcell{3}[0pt][l]{query sent to server $n$ for retrieving \\ the partial file $\widetilde{\mathbf{W}}^{(\theta)}(\mathcal{R},:)$} & \multirowcell{3}[0pt][c]{$\mathcal{Q}_n^{(\theta,\mathcal{R}_j^h)}$} & \multirowcell{3}[0pt][l]{query sent to server $n$ for retrieving the partial \\ file ${\mathbf{W}}^{(\theta)}(\mathcal{R}_j^h,:)$, designed in a encoding  \\  manner similar to $\widetilde{\mathcal{Q}}_n^{(\theta,\widetilde{\mathcal{R}}_j^h)}$} \\
 & & & & \\ & & & & \\   \hline
\multirowcell{2}[0pt][c]{Answer} & \multirowcell{2}[0pt][c]{$\widetilde{\mathcal{A}}_n^{(\theta,\mathcal{R})}$} & \multirowcell{2}[0pt][l]{response to the query $\widetilde{\mathcal{Q}}_n^{(\theta,\mathcal{R})}$} & \multirowcell{2}[0pt][c]{$\mathcal{A}_n^{(\theta,\mathcal{R}_j^h)}$} & \multirowcell{2}[0pt][l]{response to the query $\mathcal{Q}_n^{(\theta,\mathcal{R}_j^h)}$, designed in a encoding \\ manner similar to $\widetilde{\mathcal{A}}_n^{(\theta,\widetilde{\mathcal{R}}_j^h)}$} \\ 
 & & & & \\ \hline  
\multirowcell{2}[0pt][c]{Decoding} & \multirowcell{2}[0pt][c]{$\widetilde{\mathbf{W}}^{(\theta)}(\mathcal{R},:)$}  & \multirowcell{2}[0pt][l]{the partial file decoded from \\ the received responses} & \multirowcell{2}[0pt][c]{$\mathbf{W}^{(\theta)}(\mathcal{R}_j^h,:)$}  & \multirowcell{2}[0pt][l]{the partial file decoded from the received responses, \\ following a decoding method similar to $\widetilde{\mathbf{W}}^{(\theta)}(\widetilde{\mathcal{R}}_j^h,:)$} \\ 
 & & & & \\ \hline
  \end{tabular}
  \label{tab:parameters}
\end{table*}

\subsection{Feasible PIR Coding Framework}\label{Subsection:PIR:framework}
The PIR coding framework focuses on solving the following PIR problem. Assume that a dataset contains $M$ files, each of length $\widetilde{L}$, denoted as $\widetilde{\mathbf{W}}^{(0)},\widetilde{\mathbf{W}}^{(1)},\ldots,\widetilde{\mathbf{W}}^{(M-1)}$. Let $\widetilde{L}=\lambda K$, i.e., each file $\widetilde{\mathbf{W}}^{(m)}$ consists of $\lambda K$ symbols for any $m\in[0:M)$, where $\lambda$ is defined as
\begin{IEEEeqnarray}{c}\label{define:lambda}
\lambda\triangleq N-(K+X+T-1)
\end{IEEEeqnarray}
with $N>K+X+T-1$. Without loss of generality, we rearrange each file into a $\lambda\times K$ matrix, given by
\begin{IEEEeqnarray}{rCl}\label{Franmework:file:symbols}
\widetilde{\mathbf{W}}^{(m)} =\left[
  \begin{array}{@{\;}cccc@{\;}}
\widetilde{w}^{(m)}_{0,0} &  \widetilde{w}^{(m)}_{0,1} & \cdots & \widetilde{w}^{(m)}_{0,K-1}  \\
\widetilde{w}^{(m)}_{1,0} &  \widetilde{w}^{(m)}_{1,1} & \cdots & \widetilde{w}^{(m)}_{1,K-1}  \\
 \vdots &  \vdots & \vdots & \vdots \\
\widetilde{w}^{(m)}_{\lambda-1,0} &  \widetilde{w}^{(m)}_{\lambda-1,1} & \cdots &\widetilde{w}^{(m)}_{\lambda-1,K-1}\\
\end{array}
\right],\IEEEeqnarraynumspace \forall\,m\in[0:M).
\end{IEEEeqnarray}
The $M$ files are stored in a secure distributed system under the constraints similar to those of \eqref{security}, \eqref{reconstruction}, and \eqref{storage:head}. 

Unlike traditional PIR problems that focus on privately retrieving an \emph{entire file}, the PIR coding framework requires a more general retrieval scenario in which the user wishes to retrieve the \emph{partial file} $\widetilde{\mathbf{W}}^{(\theta)}(\mathcal{R},:)$ from the distributed storage system for any subset $\mathcal{R}\subseteq[0:\lambda)$ of size $r$ satisfying $1\leq r\leq\lambda$. This retrieval must be resilient to the presence of $\lambda-r$ stragglers, while keeping the index $\theta$ of the desired file private from any $T$ colluding servers. 
To address this PIR problem, we provide the following coding framework.



\subsubsection*{Secure Distributed Storage} To provide secure distributed storage, we encode each row of data in each file using random noise as a mask. Specifically, for any $m\in[0:M)$ and $i\in[0:\lambda)$, let $\widetilde{z}_{i,K}^{(m)},\ldots,\widetilde{z}_{i,K+X-1}^{(m)}$ be $X$ random noises distributed independently and uniformly on the finite field $\mathbb{F}_q$, 
and then define a storage function $\tilde{f}_i^{(m)}(x)$ as
\begin{IEEEeqnarray}{c}\label{PIRFramework:storage:polynomial}
\tilde{f}_i^{(m)}(x)=\sum\limits_{k=0}^{K-1}\widetilde{w}_{i,k}^{(m)}\cdot b_{i,k}(x)+\sum\limits_{k=K}^{K+X-1}\widetilde{z}_{i,k}^{(m)}\cdot b_{i,k}(x),\quad\forall\,m\in[0:M),i\in[0:\lambda),
\end{IEEEeqnarray}
where $\widetilde{w}_{i,0}^{(m)},\ldots,\widetilde{w}_{i,K-1}^{(m)}$ are the $K$ elements in the $i$-th row of the file $\widetilde{\mathbf{W}}^{(m)}$ by \eqref{Franmework:file:symbols}, and $b_{i,0}(x),\ldots,b_{i,K+X-1}(x)$ are deterministic functions of $x$ on $\mathbb{F}_q$. In particular, the functions $b_{i,0}(x),\ldots,b_{i,K+X-1}(x)$ are designed independently of the file index $m$, which means that all files are stored in the distributed system using the same encoding method.

Then, the evaluations of the storage functions $\{\tilde{f}_i^{(m)}(x):m\!\in\![0:M),i\!\in\![0:\lambda)\}$ at the point $x=\alpha_n$ are distributedly stored at the server $n$, given by
\begin{IEEEeqnarray}{c}\label{PIRFramework:storage:evaluation}
\widetilde{\mathcal{Y}}_n=\left\{\tilde{f}_i^{(m)}(\alpha_n):m\in[0:M),i\in[0:\lambda)\right\},\quad\forall\,n\in[0:N),
\end{IEEEeqnarray}
where $\alpha_0,\alpha_1,\ldots,\alpha_{N-1}$ are $N$ elements from $\mathbb{F}_q$, serving as the encoding parameters for this PIR framework.

\subsubsection*{Query Phase} To privately retrieve the desired symbols in $\widetilde{\mathbf{W}}^{(\theta)}(\mathcal{R},:)$ for any nonempty subset $\mathcal{R}\subseteq[0:\lambda)$, the user generates $T$ random noises $\widetilde{z}_{i,k,0}^{(m,\mathcal{R})},\ldots,\widetilde{z}_{i,k,T-1}^{(m,\mathcal{R})}$ and then constructs the private query function $\widetilde{q}_{i,k}^{(m,\mathcal{R})}(x)$ for any $m\in[0:M),i\in\mathcal{R}$, and $k\in[0:K)$, given by
\begin{IEEEeqnarray}{c}\label{Framework:query:polynomial}
\widetilde{q}_{i,k}^{(m,\mathcal{R})}(x)=\sum\limits_{t=0}^{T-1}\widetilde{z}_{i,k,t}^{(m,\mathcal{R})}\cdot v_{i,k,t}^{(\mathcal{R})}(x)
+\left\{
\begin{array}{@{}l@{\;\;}l}
v_{i,k,T}^{(\mathcal{R})}(x),
&\mathrm{if}\,\, m=\theta \\
0, & \mathrm{otherwise}
\end{array} \right.,
\end{IEEEeqnarray}
where $v_{i,k,0}^{(\mathcal{R})}(x),\ldots,v_{i,k,T}^{(\mathcal{R})}(x)$ are deterministic functions of $x$ on $\mathbb{F}_q$, designed independently of the file index $m$ since all files are stored using the same encoding method.

The query sent to the server $n$ is obtained by evaluating these query functions at $x=\alpha_n$, i.e.,
\begin{IEEEeqnarray}{c}\label{Framework:query:design}
\widetilde{\mathcal{Q}}_n^{(\theta,\mathcal{R})}=\left\{\widetilde{q}_{i,k}^{(m,\mathcal{R})}(\alpha_n):m\in[0:M),i\in\mathcal{R},k\in[0:K)\right\},\quad\forall\,n\in[0:N).
\end{IEEEeqnarray}

\subsubsection*{Answer Phase} The server $n$ will generate a response composed of $K$ sub-responses:
\begin{IEEEeqnarray}{c}\label{Framework:answers:1}
\widetilde{\mathcal{A}}_n^{(\theta,\mathcal{R})}=\left\{\widetilde{A}_{n,k}^{(\theta,\mathcal{R})}:k\in[0:K)\right\}, \quad\forall\,n\in[0:N),
\end{IEEEeqnarray}
where the sub-response $\widetilde{A}_{n,k}^{(\theta,\mathcal{R})}$ is determined based on the received queries $\{\widetilde{q}_{i,k}^{(m,\mathcal{R})}(\alpha_n)\}_{m\in[0:M),i\in\mathcal{R}}$ and the stored encoding data $\{\tilde{f}_{i}^{(m)}(\alpha_n)\}_{m\in[0:M),i\in\mathcal{R}}$, given by
\begin{IEEEeqnarray}{c}\label{Framework:answers:2}
\widetilde{A}_{n,k}^{(\theta,\mathcal{R})}=\sum\limits_{m\in[0:M)}\sum\limits_{i\in\mathcal{R}}\widetilde{q}_{i,k}^{(m,\mathcal{R})}(\alpha_n)\cdot \tilde{f}_{i}^{(m)}(\alpha_n),\quad\forall\,k\in[0:K).
\end{IEEEeqnarray}

\subsubsection*{Decoding Phase}
The user can decode the partial file $\widetilde{\mathbf{W}}^{(\theta)}(\mathcal{R},:)$ from the received responses, as described in the decoding conditions $\widetilde{\mathrm{F}}$2 and $\widetilde{\mathrm{F}}$3 of the \emph{feasible PIR coding framework} below.

Obviously, the PIR coding framework is determined by the encoding parameters  $\alpha_0,\alpha_1,\ldots,\alpha_{N-1}$ along with the encoding functions $\{b_{i,0}(x),\ldots,b_{i,K+X-1}(x)\}_{i\in[0:\lambda)}$ and $\{v_{i,k,0}^{(\mathcal{R})}(x),\ldots,v_{i,k,T}^{(\mathcal{R})}(x)\}_{i\in\mathcal{R},k\in[0:K)}$ for any nonempty subset $\mathcal{R}\subseteq[0:\lambda)$. 
We provide the formal definition of the \emph{feasible PIR coding framework} as follows.

\begin{Definition}[Feasible PIR Coding Framework]\label{Definition:feasible:PIR:framwork}
The described PIR coding framework is considered feasible over a finite field $\mathbb{F}_q$ of size $q$ if there exist a group of encoding parameters $\alpha_0,\alpha_1,\ldots,\alpha_{N-1}$ along with the encoding functions $\{b_{i,0}(x),\ldots,b_{i,K+X-1}(x)\}_{i\in[0:\lambda)}$ and $\{v_{i,k,0}^{(\mathcal{R})}(x),\ldots,v_{i,k,T}^{(\mathcal{R})}(x)\}_{i\in\mathcal{R},k\in[0:K)}$ over the finite field $\mathbb{F}_q$, such that the following four conditions are satisfied for any subset $\mathcal{R}\subseteq[0:\lambda)$ of size $r$ with $1\leq r\leq \lambda$:\footnote{The storage overhead constraint is satisfied because the storage overhead for each server is reduced by a factor of $K$ relative to the total size of all the files by \eqref{PIRFramework:storage:polynomial}-\eqref{PIRFramework:storage:evaluation}.}
\begin{enumerate}
\item[$\widetilde{\mathrm{F}}$0:] 
For any $i\in[0:\lambda)$ and any subset $\mathcal{K}=\{n_0,n_1,\ldots,n_{K+X-1}\}\subseteq[0:N)$ of size $K+X$, the following matrix $\widetilde{\mathbf{B}}_{i,\mathcal{K}}$ of dimensions $(K+X)\times (K+X)$ is non-singular over $\mathbb{F}_q$.
\begin{IEEEeqnarray}{c}\notag
\widetilde{\mathbf{B}}_{i,\mathcal{K}}=
\left[
\begin{array}{cccc}
b_{i,0}(\alpha_{n_0}) & b_{i,1}(\alpha_{n_0}) & \cdots & b_{i,K+X-1}(\alpha_{n_0}) \\
b_{i,0}(\alpha_{n_1}) & b_{i,1}(\alpha_{n_1}) & \cdots & b_{i,K+X-1}(\alpha_{n_1}) \\
    \vdots & \vdots & \vdots & \vdots \\
b_{i,0}(\alpha_{n_{K+X-1}}) & b_{i,1}(\alpha_{n_{K+X-1}}) & \cdots & b_{i,K+X-1}(\alpha_{n_{K+X-1}}) \\
  \end{array}
\right].
\end{IEEEeqnarray}

This ensures that all the files can be reconstructed by connecting to any $K+X$ servers. 
\item[$\widetilde{\mathrm{F}}$1:] The matrices $\{\widetilde{\mathbf{B}}_{i,\mathcal{X}}\}_{i\in[0:\lambda)}$ and $\{\widetilde{\mathbf{V}}_{i,k,\mathcal{T}}^{(\mathcal{R})}\}_{i\in\mathcal{R},k\in[0:K)}$ are all non-singular over $\mathbb{F}_q$ for any subset $\mathcal{X}=\{n_0,n_1,\ldots,n_{X-1}\}\subseteq[0:N)$ of size $X$ and any subset $\mathcal{T}=\{n_0,n_1,\ldots,n_{T-1}\}\subseteq[0:N)$ of size $T$, where 
\begin{IEEEeqnarray}{rCl}
\widetilde{\mathbf{B}}_{i,\mathcal{X}}&=&
\left[
\begin{array}{cccc}
b_{i,K}(\alpha_{n_0}) & b_{i,K+1}(\alpha_{n_0}) & \cdots & b_{i,K+X-1}(\alpha_{n_0}) \\
b_{i,K}(\alpha_{n_1}) & b_{i,K+1}(\alpha_{n_1}) & \cdots & b_{i,K+X-1}(\alpha_{n_1}) \\
    \vdots & \vdots & \vdots & \vdots \\
b_{i,K}(\alpha_{n_{X-1}}) & b_{i,K+1}(\alpha_{n_{X-1}}) & \cdots & b_{i,K+X-1}(\alpha_{n_{X-1}}) \\
  \end{array}
\right]_{X\times X},\notag\\
\widetilde{\mathbf{V}}_{i,k,\mathcal{T}}^{(\mathcal{R})}&=&
\left[
\begin{array}{cccc}
v_{i,k,0}^{(\mathcal{R})}(\alpha_{n_0}) & v_{i,k,1}^{(\mathcal{R})}(\alpha_{n_0}) & \cdots & v_{i,k,T-1}^{(\mathcal{R})}(\alpha_{n_0}) \\
v_{i,k,0}^{(\mathcal{R})}(\alpha_{n_1}) & v_{i,k,1}^{(\mathcal{R})}(\alpha_{n_1}) & \cdots & v_{i,k,T-1}^{(\mathcal{R})}(\alpha_{n_1}) \\
\vdots & \vdots & \vdots & \vdots \\
v_{i,k,0}^{(\mathcal{R})}(\alpha_{n_{T-1}}) & v_{i,k,1}^{(\mathcal{R})}(\alpha_{n_{T-1}}) & \cdots & v_{i,k,T-1}^{(\mathcal{R})}(\alpha_{n_{T-1}}) \\
\end{array}
\right]_{T\times T}.\notag
\end{IEEEeqnarray}
This ensures that the storage system is secure against any $X$ colluding servers and the queries sent to the servers are private against any $T$ colluding servers.
\item[$\widetilde{\mathrm{F}}$2:] The user can recover the partial file $\widetilde{\mathbf{W}}^{(\theta)}(\mathcal{R},:)$ from any $N-(\lambda-r)$ out of the $N$ responses $\widetilde{\mathcal{A}}_{[0:N)}^{(\theta,\mathcal{R})}$. This ensures that the desired partial file can be decoded even in the presence of $\lambda-r$ stragglers.
\item[$\widetilde{\mathrm{F}}$3:] The partial file $\widetilde{\mathbf{W}}^{(\theta)}(\mathcal{R},:)$ can be recovered from any $N-(\lambda-r+d)$ out of the $N$ responses $\widetilde{\mathcal{A}}_{[0:N)}^{(\theta,\mathcal{R})}$ after obtaining any $d$ rows of    $\widetilde{\mathbf{W}}^{(\theta)}(\mathcal{R},:)$ for any $d\!\in\![1:r)$, i.e., for any subset $\Delta\!\subseteq\![0:N)$ of size $N-(\lambda-r+d)$ and any subset $\mathcal{D}\!\subseteq\!\mathcal{R}$ of size $d$,
\begin{IEEEeqnarray}{c}
H\big(\widetilde{\mathbf{W}}^{(\theta)}(\mathcal{R},:)|\widetilde{\mathcal{A}}_{\Delta}^{(\theta,\mathcal{R})},\widetilde{\mathbf{W}}^{(\theta)}(\mathcal{D},:)\big)=0.\notag
\end{IEEEeqnarray}
This ensures that the desired partial file can be decoded even in the presence of $\lambda-r+d$ stragglers, as long as any $d$ rows of the partial file are available. This is a crucial property that enables the adaptive PIR scheme, constructed using the coding framework, to successfully decode the desired file.
\end{enumerate}
\end{Definition} 

It is important to emphasize that the encoding parameters and encoding functions of any PIR coding framework are designed independently of the values of the files $\widetilde{\mathbf{W}}^{(0)},\widetilde{\mathbf{W}}^{(1)},\ldots,\widetilde{\mathbf{W}}^{(M-1)}$.
Therefore, any PIR coding framework satisfies the conditions $\widetilde{\mathrm{F}}$0--$\widetilde{\mathrm{F}}$3, regardless of the files $\widetilde{\mathbf{W}}^{(0)},\widetilde{\mathbf{W}}^{(1)},\ldots,\widetilde{\mathbf{W}}^{(M-1)}$.



Next, we illustrate the \emph{feasible PIR coding framework} through an example. This example will be extended into the adaptive scenario in the next section to demonstrate the connection between the \emph{feasible PIR coding framework} and the proposed adaptive PIR scheme.
\begin{Example}\label{Flexible:example}
Consider a PIR problem with system parameters $N=8,K=X=T=2,M=3$. Here, we have $\lambda=3$. Then, each of the $3$ files $\widetilde{\mathbf{W}}^{(0)},\widetilde{\mathbf{W}}^{(1)},\widetilde{\mathbf{W}}^{(2)}$ can be represented as a matrix of dimensions $3\times 2$, given by
\begin{IEEEeqnarray}{c}\label{example:files}
\widetilde{\mathbf{W}}^{(0)} =\left[
\begin{array}{@{\;}cc@{\;}}
\widetilde{w}^{(0)}_{0,0} & \widetilde{w}^{(0)}_{0,1}  \\
\widetilde{w}^{(0)}_{1,0} & \widetilde{w}^{(0)}_{1,1}  \\
\widetilde{w}^{(0)}_{2,0} & \widetilde{w}^{(0)}_{2,1}  \\
\end{array}
\right], \quad
\widetilde{\mathbf{W}}^{(1)} =\left[
\begin{array}{@{\;}cc@{\;}}
\widetilde{w}^{(1)}_{0,0} & \widetilde{w}^{(1)}_{0,1}  \\
\widetilde{w}^{(1)}_{1,0} & \widetilde{w}^{(1)}_{1,1}  \\
\widetilde{w}^{(1)}_{2,0} & \widetilde{w}^{(1)}_{2,1}  \\
\end{array}
\right], \quad
\widetilde{\mathbf{W}}^{(2)} =\left[
\begin{array}{@{\;}cc@{\;}}
\widetilde{w}^{(2)}_{0,0} & \widetilde{w}^{(2)}_{0,1}  \\
\widetilde{w}^{(2)}_{1,0} & \widetilde{w}^{(2)}_{1,1}  \\
\widetilde{w}^{(2)}_{2,0} & \widetilde{w}^{(2)}_{2,1}  \\
\end{array}
\right].
\end{IEEEeqnarray}

For any given \emph{feasible PIR coding framework}, let $\alpha_0,\alpha_1,\ldots,\alpha_{7}$ represent its encoding parameters, while $\{b_{i,0}(x),b_{i,1}(x),b_{i,2}(x),b_{i,3}(x)\}_{i\in[0:3)}$ and $\{v_{i,k,0}^{(\mathcal{R})}(x),v_{i,k,1}^{(\mathcal{R})}(x),v_{i,k,2}^{(\mathcal{R})}(x)\}_{i\in\mathcal{R},k\in[0:2)}$ denote its encoding functions for any subset $\mathcal{R}\!\subseteq\![0\!:\!3)$ of size $1\!\leq\! r\!\leq\! 3$. Next, we utilize this feasible PIR framework to provide a PIR construction for privately retrieving the partial file $\widetilde{\mathbf{W}}^{(\theta)}(\mathcal{R},:)$, for any $\theta\!\in\![0\!:\!3)$ and any nonempty subset $\mathcal{R}\!\subseteq\![0\!:\!3)$.

\subsubsection*{Secure Distributed Storage}
In the storage phase, we design secure storage functions to encode the data in each row of the three files using $\{b_{i,0}(x),\ldots,b_{i,3}(x)\}_{i\in[0:3)}$, i.e., for any $m\in[0:3)$,
\begin{IEEEeqnarray}{rCl}
\tilde{f}_0^{(m)}(x)&=&\widetilde{w}_{0,0}^{(m)}\cdot b_{0,0}(x)+\widetilde{w}_{0,1}^{(m)}\cdot b_{0,1}(x)+\widetilde{z}_{0,2}^{(m)}\cdot b_{0,2}(x)+\widetilde{z}_{0,3}^{(m)}\cdot b_{0,3}(x), \label{Example:framework:storage:1}\\
\tilde{f}_1^{(m)}(x)&=&\widetilde{w}_{1,0}^{(m)}\cdot b_{1,0}(x)+\widetilde{w}_{1,1}^{(m)}\cdot b_{1,1}(x)+\widetilde{z}_{1,2}^{(m)}\cdot b_{1,2}(x)+\widetilde{z}_{1,3}^{(m)}\cdot b_{1,3}(x), \label{Example:framework:storage:2}\\
\tilde{f}_2^{(m)}(x)&=&\widetilde{w}_{2,0}^{(m)}\cdot b_{2,0}(x)+\widetilde{w}_{2,1}^{(m)}\cdot b_{2,1}(x)+\widetilde{z}_{2,2}^{(m)}\cdot b_{2,2}(x)+\widetilde{z}_{2,3}^{(m)}\cdot b_{2,3}(x),\label{Example:framework:storage:3}
\end{IEEEeqnarray}
where $\{\widetilde{z}_{i,2}^{(m)},\widetilde{z}_{i,3}^{(m)}\}_{i\in[0:3)}$ are random noises over the finite field $\mathbb{F}_q$ and are used to provide secure storage guarantee.
Then the data stored at server $n,n\in[0:8)$ is given by
\begin{IEEEeqnarray}{c}
\mathcal{Y}_n=\left\{
\begin{array}{@{}ccc@{}}
\tilde{f}_{0}^{(0)}(a_n), & \tilde{f}_{0}^{(1)}(a_n), & \tilde{f}_{0}^{(2)}(a_n) \\
\tilde{f}_{1}^{(0)}(a_n), & \tilde{f}_{1}^{(1)}(a_n), & \tilde{f}_{1}^{(2)}(a_n) \\
\tilde{f}_{2}^{(0)}(a_n), & \tilde{f}_{2}^{(1)}(a_n), & \tilde{f}_{2}^{(2)}(a_n) \\
\end{array}
\right\}.\notag
\end{IEEEeqnarray}

\subsubsection*{Query Phase}
To privately retrieve the partial file $\widetilde{\mathbf{W}}^{(\theta)}(\mathcal{R},:)$, we choose the random noises $\widetilde{z}_{i,k,0}^{(m,\mathcal{R})},\widetilde{z}_{i,k,1}^{(m,\mathcal{R})}$ and design the query function $\widetilde{q}_{i,k}^{(m,\mathcal{R})}(x)$ using $\{v_{i,k,0}^{(\mathcal{R})}(x),v_{i,k,1}^{(\mathcal{R})}(x),v_{i,k,2}^{(\mathcal{R})}(x)\}$
for any $m\!\in\![0:3),i\!\in\!\mathcal{R}$, and $k\!\in\![0:2)$, given by
\begin{IEEEeqnarray}{c}\label{Example:framework:query:polynomial}
\widetilde{q}_{i,k}^{(m,\mathcal{R})}(x)=
\widetilde{z}_{i,k,0}^{(m,\mathcal{R})}\cdot v_{i,k,0}^{(\mathcal{R})}(x)+
\widetilde{z}_{i,k,1}^{(m,\mathcal{R})}\cdot v_{i,k,1}^{(\mathcal{R})}(x)
+\left\{
\begin{array}{@{}l@{\;\;}l}
v_{i,k,2}^{(\mathcal{R})}(x),
&\mathrm{if}\,\, m=\theta \\
0, & \mathrm{otherwise}
\end{array} \right..
\end{IEEEeqnarray}
Then the query sent to the server $n,n\in[0:8)$ is given by
\begin{IEEEeqnarray}{c}
\widetilde{\mathcal{Q}}_n^{(\theta,\mathcal{R})}=\left\{\widetilde{q}_{i,k}^{(m,\mathcal{R})}(\alpha_n):m\in[0:3),i\in\mathcal{R},k\in[0:2)\right\}.\notag
\end{IEEEeqnarray}

\subsubsection*{Answer Phase}
The server $n,n\in[0:8)$ will generate the response
\begin{IEEEeqnarray}{c}\label{Example:PIR:answer:1}
\widetilde{\mathcal{A}}_n^{(\theta,\mathcal{R})}=\left\{\widetilde{A}_{n,k}^{(\theta,\mathcal{R})}:k\in[0:2)\right\},
\end{IEEEeqnarray}
where $\widetilde{A}_{n,k}^{(\theta,\mathcal{R})}$ is given by 
\begin{IEEEeqnarray}{c}\label{Example:PIR:answer:2}
\widetilde{A}_{n,k}^{(\theta,\mathcal{R})}=\sum\limits_{m\in[0:3)}\sum\limits_{i\in\mathcal{R}}\widetilde{q}_{i,k}^{(m,\mathcal{R})}(\alpha_n)\cdot \tilde{f}_{i}^{(m)}(\alpha_n).
\end{IEEEeqnarray}

\subsubsection*{Decoding Phase}
Since the \emph{feasible PIR coding framework} satisfies the conditions $\widetilde{\mathrm{F}}$2 and $\widetilde{\mathrm{F}}$3, 
the user can decode the partial file $\widetilde{\mathbf{W}}^{(\theta)}(\mathcal{R},:)$ from any $5+r-d$ out of the $8$ responses $\widetilde{\mathcal{A}}_{0}^{(\theta,\mathcal{R})},\widetilde{\mathcal{A}}_{1}^{(\theta,\mathcal{R})},\ldots,\widetilde{\mathcal{A}}_{7}^{(\theta,\mathcal{R})}$ after obtaining any $d$ rows of $\widetilde{\mathbf{W}}^{(\theta)}(\mathcal{R},:)$, for any subset $\mathcal{R}\subseteq[0:3)$ of size $1\leq r\leq 3$ and any $d\in[0:r)$.
\end{Example}

\subsection{Implementation of Feasible PIR Coding Framework}\label{subsection:implementation:PIR}


Lagrange codes \cite{yu2019lagrange,zhu2022generalized,zhu2024generalized} and cross subspace alignment codes \cite{jia2019cross,jia2020x,jia2022xFSL,jia2020asymptotic} are two distinct types of codes that can effectively address various PIR problems, such as the multi-user blind PIR problem \cite{lu2021double,zhu2022multi} and the private function computation problem \cite{zhu2022symmetric,raviv2019private222,jia2024asymptotic}. Both types of codes can be used to separately implement the \emph{feasible PIR coding framework}. Here, to simplify the discussion, we only provide the construction of a \emph{feasible PIR coding framework} based on Lagrange codes, following the ideas from \cite[Section III]{zhu2022symmetric} and \cite[Section III]{zhu2022multi}.


Before proceeding, we identify a group of elements from the finite field $\mathbb{F}_q$ that will serve as the encoding parameters for the \emph{feasible PIR coding framework} implemented using Lagrange codes.


Let $\{\beta_{i,k}:i\in[0:\lambda),k\in[0:K+X)\}$ be the elements from the finite field $\mathbb{F}_q$ of size $q$. These elements can be denoted by a matrix $\boldsymbol{\beta}$ with dimensions $\lambda\times(K+X)$, given by
\begin{IEEEeqnarray}{rCl}\notag
\boldsymbol{\beta}\triangleq\left[
  \begin{array}{ccc;{2pt/2pt}ccc}
    \beta_{0,0} & \cdots & \beta_{0,K-1} & \beta_{0,K} & \cdots & \beta_{0,K+X-1}  \\
    \vdots & \vdots & \vdots & \vdots & \vdots & \vdots \\
    \beta_{\lambda-1,0} & \cdots &\beta_{\lambda-1,K-1} & \beta_{\lambda-1,K} & \cdots &\beta_{\lambda-1,K+X-1} \\
\end{array}
\right].
\end{IEEEeqnarray}

In order to ensure the feasibility of the PIR coding framework constructed using Lagrange codes, we require that the elements $\{\beta_{i,k}\!:\!i\!\in\![0\!:\!\lambda),k\!\in\![0\!:\!K\!+\!X)\}$, along with $\{\alpha_n\!:\!n\!\in\![0\!:\!N)\}$, satisfy the following four constraints.
\begin{enumerate}
  \item[P0.] All elements in each row of the matrix $\boldsymbol{\beta}$ are pairwise distinct, i.e., for any given $i\in[0:\lambda)$, $\beta_{i,k}\neq\beta_{i,k'}$ for all $k, k'\in[0:K+X)$ with $k\neq k'$;
  \item[P1.] Each column in the first $K$ columns of the matrix $\boldsymbol{\beta}$ is composed of distinct elements, i.e., for any given $k\in[0:K)$, $\beta_{i,k}\neq\beta_{i',k}$ for all $i,i'\in[0:\lambda)$ with $i\neq i'$;
  \item[P2.] The elements $\alpha_0,\ldots,\alpha_{N-1}$ are distinct of each other, i.e., $\alpha_n\neq\alpha_{n'}$ for all $n,n'\in[0:N)$ with $n\neq n'$;
  \item[P3.] The elements $\alpha_0,\ldots,\alpha_{N-1}$ are distinct from the elements in the first $K$ columns of the matrix $\boldsymbol{\beta}$, i.e., $\{\alpha_n:n\in[0:N)\}\cap\{\beta_{i,k}:i\in[0:\lambda),k\in[0:K)\}=\emptyset$.
\end{enumerate}

There exists a group of elements $\{\beta_{i,k},\alpha_n:i\!\in\![0:\lambda),k\!\in\![0:K+X),n\!\in\![0:N)\}$ that satisfies the constraints P0--P3 in a finite field $\mathbb{F}_q$, for any prime power $q$ with size $q\geq N+\max\{K, \lambda\}$. In particular, the optimal size of the finite field $\mathbb{F}_q$ that satisfies the constraints P0--P3 is $N+\max\{K,\lambda\}$. We summarize this result in the following lemma, and its proof is placed in Appendix-\ref{Proof:finite:field:size}.
\begin{Lemma}\label{theorem:tield size}
The optimal size of the finite field $\mathbb{F}_q$ that satisfies the constraints P0--P3 is $q=N+\max\{K,\lambda\}$.
\end{Lemma}


Next, we implement the \emph{feasible PIR coding framework} based on Lagrange codes,  utilizing the group of elements $\{\beta_{i,k},\alpha_n:i\in[0:\lambda),k\in[0:K+X),n\in[0:N)\}$ that satisfies the constraints P0--P3 as encoding parameters. 
It is sufficient to provide the construction of the encoding functions $\{b_{i,0}(x),\ldots,b_{i,K+X-1}(x)\}_{i\in[0:\lambda)}$ and $\{v_{i,k,0}^{(\mathcal{R})}(x),\ldots,v_{i,k,T}^{(\mathcal{R})}(x)\}_{i\in\mathcal{R},k\in[0:K)}$ for any nonempty subset $\mathcal{R}\subseteq[0:\lambda)$.
The key idea behind designing these encoding functions is to construct the secure storage functions and private query functions using Lagrange interpolating polynomials.

More formally, for any $m\in[0:M)$ and $i\in[0:\lambda)$, the storage function $\tilde{f}_i^{(m)}(x)$ is constructed as a Lagrange interpolating polynomial of degree $K+X-1$, given by
\begin{IEEEeqnarray}{c}
\tilde{f}_i^{(m)}(x)=\sum\limits_{k=0}^{K-1}\widetilde{w}_{i,k}^{(m)}\cdot\underbrace{\prod_{j\in[0:K+X)\backslash\{k\}}\frac{x-\beta_{i,j}}{\beta_{i,k}-\beta_{i,j}}}_{=b_{i,k}(x)}+\sum\limits_{k=K}^{K+X-1}\widetilde{z}_{i,k}^{(m)}\cdot\underbrace{\prod_{j\in[0:K+X)\backslash\{k\}}\frac{x-\beta_{i,j}}{\beta_{i,k}-\beta_{i,j}}}_{=b_{i,k}(x)}. \notag 
\end{IEEEeqnarray}
For any given subset $\mathcal{R}\subseteq[0:\lambda)$ of size $r$ with $1\leq r\leq\lambda$, the query function $\widetilde{q}_{i,k}^{(m,\mathcal{R})}(x)$ is constructed as a Lagrange interpolating polynomial of degree $r+T-1$ for any $m\in[0:M),i\in\mathcal{R},$ and $k\in[0:K)$, given by
\begin{IEEEeqnarray}{c}
\widetilde{q}_{i,k}^{(m,\mathcal{R})}\!(x)\!=\!\!\sum\limits_{t=0}^{T-1}\widetilde{z}_{i,k,t}^{(m,\mathcal{R})}\!\!\cdot\!\!\underbrace{\bigg(\!\prod\limits_{j\in[0:T)\backslash\{t\}}\!\!\frac{x\!-\!\alpha_{j}}{\alpha_{t}\!-\!\alpha_{j}}\!\bigg)\!\!\bigg(\!\prod\limits_{j\in\mathcal{R}}\!\!\frac{x\!-\!\beta_{j,k}}{\alpha_{t}\!-\!\beta_{j,k}}\!\bigg)}_{=v_{i,k,t}^{(\mathcal{R})}(x)}
\!+\!\left\{
\begin{array}{@{}l@{\;\,}l}
\!\!\underbrace{\bigg(\!\prod\limits_{j\in\mathcal{R}\backslash\{i\}}\!\!\frac{x\!-\!\beta_{j,k}}{\beta_{i,k}\!-\!\beta_{j,k}}\!\bigg)\!\!\bigg(\!\prod\limits_{j\in[0:T)}\!\!\frac{x\!-\!\alpha_{j}}{\beta_{i,k}\!-\!\alpha_{j}}\!\bigg)}_{=v_{i,k,T}^{(\mathcal{R})}(x)},
&\mathrm{if}\,\, m\!=\!\theta \\
0, & \mathrm{otherwise}
\end{array} \right.\!\!. \notag
\end{IEEEeqnarray}

The feasibility of this PIR coding framework can be readily established by following arguments similar to those in \cite[Sections III-C and IV]{zhu2022symmetric} and \cite[Section III-D and Appendix]{zhu2022multi}.
At a high level, the conditions $\widetilde{\mathrm{F}}$0 and $\widetilde{\mathrm{F}}$1 can be proved by the Lagrange interpolation theorem and the constraints P0--P3. In addition, the conditions $\widetilde{\mathrm{F}}$2 and $\widetilde{\mathrm{F}}$3 are satisfied because the server responses can be interpreted as evaluations of the product of the query and storage polynomials, such that the desired partial file can be recovered through interpolation of the resulting product polynomial.
A detailed example is provided in Appendix-\ref{example:lagrange:codes} to demonstrate the feasibility of the PIR coding framework based on Lagrange codes.

By Lemma \ref{theorem:tield size}, we obtain the following results regarding the \emph{feasible PIR coding framework}.
\begin{Lemma}\label{lemma:feasible:PIR}
The \emph{feasible PIR coding framework} can be implemented over the finite field $\mathbb{F}_q$ for any prime power $q$ with size $q\geq N+\max\{K, \lambda\}$.
\end{Lemma}

\begin{Remark}
The \emph{feasible PIR coding framework} can also be implemented based on cross subspace alignment (CSA) codes, which were initially introduced in \cite{jia2019cross} and have since been widely applied to various PIR problems \cite{jia2020x,jia2022xFSL,jia2020asymptotic,lu2021double}, as well as to coded distributed matrix multiplication problems \cite{jia2021cross,chen2021gcsa,zhu2021improved}. 
The CSA codes create a novel form of interference alignment inspired by the Cauchy-Vandermonde structure, where the desired data is aligned along the Cauchy terms, while all other components are aligned along the Vandermonde terms. In particular, the matrix formed by these Cauchy and/or Vandermonde terms is full rank, i.e., for any non-negative integers $k,t$ and any $2k+t$ pairwise distinct elements $\alpha_0,\ldots,\alpha_{k+t-1},\beta_0,\ldots,\beta_{k-1}$ from the finite field $\mathbb{F}_{q}$, the following Cauchy-Vandermonde matrix $\mathbf{C}$ is invertible over $\mathbb{F}_q$  \cite[Lemma 1]{jia2020x}.
\begin{IEEEeqnarray}{rCl}\label{C-V matrix}
\mathbf{C}=
\substack{
\left[
  \begin{array}{ccccccccccc}
    \frac{1}{\alpha_0-\beta_0} & \frac{1}{\alpha_0-\beta_1} & \ldots & \frac{1}{\alpha_0-\beta_{k-1}} & 1 & \alpha_0 & \ldots & \alpha_0^{t-1} \\
    \frac{1}{\alpha_1-\beta_0} & \frac{1}{\alpha_1-\beta_1} & \ldots & \frac{1}{\alpha_1-\beta_{k-1}} & 1 & \alpha_1 & \ldots & \alpha_1^{t-1} \\
    \vdots & \vdots & \vdots & \vdots & \vdots & \vdots & \vdots & \vdots \\
    \frac{1}{\alpha_{k+t-1}-\beta_0} & \frac{1}{\alpha_{k+t-1}-\beta_1} & \ldots & \frac{1}{\alpha_{k+t-1}-\beta_{k-1}} & 1 & \alpha_{k+t-1} & \ldots & \alpha_{k+t-1}^{t-1} \\
  \end{array}
\right]_{(k+t)\times(k+t)}.\\
~\underbrace{~~~~~~~~~~~~~~~~~~~~~~~~~~~~~~~~~~~~~~~~~~~~~~~~~~~~}_{\text{\footnotesize{Cauchy terms}}}~~\,
\underbrace{~~~~~~~~~~~~~~~~~~~~~~~~~~~~~~~}_{\text{\footnotesize{Vandermonde terms}}}~~~~~~~~~~~~\,
} 
\end{IEEEeqnarray}
More precisely, for any $m\in[0:M)$ and $i\in[0:\lambda)$, the CSA-encoded storage function $\tilde{f}_i^{(m)}(x)$ is constructed as
\begin{IEEEeqnarray}{c}
\tilde{f}_i^{(m)}(x)=\sum\limits_{k=0}^{K-1}\widetilde{w}_{i,k}^{(m)}\cdot\underbrace{\frac{1}{x-\beta_{i,k}}}_{=b_{i,k}(x)}+\sum\limits_{k=K}^{K+X-1}\widetilde{z}_{i,k}^{(m)}\cdot\underbrace{x^{k-K}}_{=b_{i,k}(x)}. \label{CSA:codes:storage}
\end{IEEEeqnarray}
Additionally, for any nonempty subset $\mathcal{R}\subseteq[0:\lambda)$ and any $m\in[0:M),i\in\mathcal{R},k\in[0:K)$, the CSA-encoded query polynomial $\widetilde{q}_{i,k}^{(m,\mathcal{R})}(x)$ is constructed as
\begin{IEEEeqnarray}{c}
\widetilde{q}_{i,k}^{(m,\mathcal{R})}(x)=\sum\limits_{t=0}^{T-1}\widetilde{z}_{i,k,t}^{(m,\mathcal{R})}\cdot \underbrace{\bigg(x^t\cdot\prod\limits_{k\in[0:K)}\left(x-\beta_{i,k}\right)\bigg)}_{=v_{i,k,t}^{(\mathcal{R})}(x)}
+\left\{
\begin{array}{@{}l@{\;\;}l}
\underbrace{\prod\limits_{j\in[0:K)\backslash\{k\}}\frac{x-\beta_{i,j}}{\beta_{i,k}-\beta_{i,j}}}_{=v_{i,k,T}^{(\mathcal{R})}(x)},
&\mathrm{if}\,\, m=\theta \\
0, & \mathrm{otherwise}
\end{array} \right..  \IEEEeqnarraynumspace\label{CSA:codes:query}
\end{IEEEeqnarray}
The PIR coding framework based on CSA codes satisfies the conditions $\widetilde{\mathrm{F}}$0--$\widetilde{\mathrm{F}}$3 due to the following two  factors: 
\begin{itemize}
    \item The storage and query functions are carefully designed such that the confidential information (i.e., the data symbols $\widetilde{w}_{i,k}^{(m)}$ and the desired file index $\theta$) is aligned along the Cauchy terms (i.e., the first term of \eqref{CSA:codes:storage} and the second term of \eqref{CSA:codes:query}), while the random noise used to ensure privacy is aligned along the Vandermonde terms (i.e., the second term of \eqref{CSA:codes:storage} and the first term of \eqref{CSA:codes:query}). This ensures that the conditions $\widetilde{\mathrm{F}}$0 and $\widetilde{\mathrm{F}}$1 are satisfied due to the invertibility of the Cauchy-Vandermonde matrix in \eqref{C-V matrix}.
    \item The server responses generated by the CSA-encoded storage and queries also exhibit the form of interference alignment inspired by the Cauchy-Vandermonde structure. Specifically, by \eqref{CSA:codes:storage} and \eqref{CSA:codes:query}, the sub-response $\widetilde{A}_{n,k}^{(\theta,\mathcal{R})}$ of server $n\in[0:N)$ in \eqref{Framework:answers:2} can be expressed in the following form:
\begin{IEEEeqnarray}{rCl}
\widetilde{A}_{n,k}^{(\theta,\mathcal{R})}&=&\sum\limits_{m\in[0:M)}\sum\limits_{i\in\mathcal{R}}\widetilde{q}_{i,k}^{(m,\mathcal{R})}(\alpha_n)\cdot \tilde{f}_{i}^{(m)}(\alpha_n) \notag\\
&=&\sum\limits_{i\in\mathcal{R}}\widetilde{q}_{i,k}^{(\theta,\mathcal{R})}(\alpha_n)\cdot \tilde{f}_{i}^{(\theta)}(\alpha_n)+\sum\limits_{m\in[0:M)\backslash\{\theta\}}\sum\limits_{i\in\mathcal{R}}\widetilde{q}_{i,k}^{(m,\mathcal{R})}(\alpha_n)\cdot \tilde{f}_{i}^{(m)}(\alpha_n) \label{framework:answer:456}\\
&=&\sum\limits_{i\in\mathcal{R}} \widetilde{w}_{i,k}^{(\theta)}\cdot\frac{1}{\alpha_n-\beta_{i,k}}+\sum\limits_{j=0}^{K+X+T-2}I_j\cdot \alpha_n^j, \label{framework:answer:789}
\end{IEEEeqnarray}
where $I_j$ represents various combinations of interference, which can be precisely determined by expanding the equation in \eqref{framework:answer:456}, though its exact form is not crucial.
From the sub-response $\widetilde{A}_{n,k}^{(\theta,\mathcal{R})}$ in \eqref{framework:answer:789} for any $n\in[0:N)$ and $k\in[0:K)$, we can find that the desired symbol $\widetilde{w}^{(\theta)}_{i,k}$ is aligned along the dimension corresponding to the Cauchy term $\frac{1}{\alpha_n-\beta_{i,k}}$ for any $i\in\mathcal{R}$, and the interference $I_j$ is aligned along the dimension corresponding to the Vandermonde term $\alpha_n^j$ for any $j\in[0:K+X+T-2]$. This ensures that the conditions $\widetilde{\mathrm{F}}$2 and $\widetilde{\mathrm{F}}$3 are satisfied by the invertibility of the Cauchy-Vandermonde matrix in \eqref{C-V matrix} again.
\end{itemize}


\end{Remark}

\section{General Adaptive PIR Scheme}\label{scheme:APIR}


In this section, we construct a general adaptive PIR scheme using the \emph{feasible PIR coding framework} as a building block. 
To provide the adaptive guarantee, we design the adaptive PIR scheme with $\lambda$ layers:
\begin{enumerate}
\item 
The $0$-th layer is used to privately retrieve the desired file in the presence of $S=0$ straggler. This ensures that the user can recover the desired file by decoding the responses associated with the $0$-th layer received from all the $N$ servers. 
\item 
Due to the straggler effect, some responses are missing for decoding the $0$-th layer, as the number of stragglers increases. For any given $h=1,2,\ldots,\lambda-1$, the $h$-th layer is used to compensate for the missing responses required to decode the first $h$ layers in the presence of $S=h$ stragglers. Note that the $h$-th layer is designed to compensate for the missing responses required to decode all previous $h$ layers, with the purpose of improving communication efficiency. This is because before receiving the responses associated with the $h$-th layer, the user would have already received the responses associated with  the $0$-th, $1$-th,$\ldots$, and $(h-1)$-th layers. These responses can also be exploited to compensate for the missing responses required to decode the $0$-th layer.
\end{enumerate}
This enables the adaptive PIR scheme to tolerate the presence of $S$ stragglers simultaneously for all $0\leq S\leq \lambda-1$. 
Specifically, in the proposed adaptive PIR scheme, we design the $\lambda$ layers by introducing a query array composed of $\lambda$ subarrays, each corresponding exactly to a layer of the scheme. The construction of the adaptive PIR scheme is then completed by iteratively invoking the \emph{feasible PIR coding framework} via the query array.

Assume that the \emph{feasible PIR coding framework} operates over a finite field $\mathbb{F}_q$ of size $q$. Let $\alpha_0,\alpha_1,\ldots,\alpha_{N-1}$ represent the encoding parameters of the \emph{feasible PIR coding framework}, while $\{b_{i,0}(x),\!\ldots\!,b_{i,K+X-1}(x)\}_{i\in[0:\lambda)}$ and $\{v_{i,k,0}^{(\mathcal{R})}(x),\!\ldots\!,v_{i,k,T}^{(\mathcal{R})}(x)\}_{i\in\mathcal{R},k\in[0:K)}$ represent its encoding functions for any nonempty subset $\mathcal{R}\subseteq[0:\lambda)$. 
Next, we present the formal construction of the adaptive PIR scheme based on the \emph{feasible PIR coding framework}. 

\subsection{Secure Distributed Storage}\label{section:storage:system}
In this subsection, we describe the data encoding procedures using the parameters $\alpha_0,\alpha_1,\ldots,\alpha_{N-1}$ and the functions $\{b_{i,0}(x),\ldots,b_{i,K+X-1}(x)\}_{i\in[0:\lambda)}$ from the \emph{feasible PIR coding framework}, which ensures that all the $M$ files $\mathbf{W}^{(0)},\mathbf{W}^{(1)},\ldots,\mathbf{W}^{(M-1)}$ are stored secretly and reliably in the distributed system. 

Since the number of stragglers is unknown in advance, to tolerate the presence of $S$ stragglers simultaneously for all $0\leq S\leq\lambda-1$, the file size $L$ needs to be sufficiently large to enhance the flexibility in designing the adaptive PIR scheme.
Here, we set $L=K\lambda\cdot\mathrm{lcm}(1,2,\ldots,\lambda)$, where $\lambda$ is defined in \eqref{define:lambda} and $\mathrm{lcm}(1,2,\ldots,\lambda)$ is the least common multiple of $1,2,\ldots,\lambda$. That is, each file $\mathbf{W}^{(m)}$ consists of $K\lambda\cdot\mathrm{lcm}(1,2,\ldots,\lambda)$ symbols, which can be arranged in a $P\times K$ matrix, given by
\begin{IEEEeqnarray}{rCl}\label{file:symbols}
\mathbf{W}^{(m)} =\left[
  \begin{array}{@{\;}cccc@{\;}}
w^{(m)}_{0,0} &  w^{(m)}_{0,1} & \cdots & w^{(m)}_{0,K-1}  \\
w^{(m)}_{1,0} &  w^{(m)}_{1,1} & \cdots & w^{(m)}_{1,K-1}  \\
 \vdots &  \vdots & \vdots & \vdots \\
w^{(m)}_{P-1,0} &  w^{(m)}_{P-1,1} & \cdots &w^{(m)}_{P-1,K-1}\\
\end{array}
\right],\IEEEeqnarraynumspace \forall\,m\in[0:M),
\end{IEEEeqnarray}
where 
\begin{IEEEeqnarray}{c}\label{number:value}
    P\triangleq\lambda\cdot\mathrm{lcm}(1,2,\ldots,\lambda).
\end{IEEEeqnarray}

The submatrix $\mathbf{W}^{(m)}(\{i+j\cdot\lambda:i\in[0:\lambda)\},:)$ of dimensions $\lambda\times K$ is stored in the distributed system using the same encoding method as the matrix $\widetilde{\mathbf{W}}^{(m)}$ in \eqref{Franmework:file:symbols}, for all $j\in[0:\frac{P}{\lambda})$ and $m\in[0:M)$. More specifically, similar to \eqref{PIRFramework:storage:polynomial}, we select the $X$ random noises $z_{i,K}^{(m)},\ldots,z_{i,K+X-1}^{(m)}$ and then construct a function $f_i^{(m)}(x)$ to encode the $K$ elements $w_{i,0}^{(m)},\ldots,w_{i,K-1}^{(m)}$ in the $i$-th row of the file $\mathbf{W}^{(m)}$ for any $m\in[0:M)$ and $i\in[0:P)$, where 
\begin{IEEEeqnarray}{c}\label{Adaptive:PIR:storage:function}
f_i^{(m)}(x)=\sum\limits_{k=0}^{K-1}w_{i,k}^{(m)}\cdot b_{(i)_{\lambda},k}(x)+\sum\limits_{k=K}^{K+X-1}z_{i,k}^{(m)}\cdot b_{(i)_{\lambda},k}(x),\quad\forall\,m\in[0:m),i\in[0:P).
\end{IEEEeqnarray}
The encoded data stored at the server $n$ is given by evaluating the functions $\{f_i^{(m)}(x):m\in[0:M),i\in[0:P)\}$ at $x=\alpha_n$, i.e.,
\begin{IEEEeqnarray}{c}\label{storage:data}
\mathcal{Y}_n=\left\{f_i^{(m)}(\alpha_n):m\in[0:M),i\in[0:P)\right\},\quad\forall\,n\in[0:N).
\end{IEEEeqnarray}

Since the encoding functions $\{b_{i,0}(x),\ldots,b_{i,K+X-1}(x)\}_{i\in[0:\lambda)}$ of the \emph{feasible PIR coding framework} meet the condition $\widetilde{\mathrm{F}}$0, the coded storage system satisfies the reliability constraint in \eqref{reconstruction}. Moreover, the storage overhead at each server is limited to $MP=ML/K$, satisfying the storage overhead constraint in \eqref{storage:head}. We also prove in Section \ref{APIR:secrecy:Privacy} that the storage system satisfies the secrecy constraint in \eqref{security}.

We now aim to design an adaptive PIR scheme to privately retrieve the desired file from the secure distributed storage system. 
As is well known, the key to the adaptive PIR scheme lies in designing private queries that can adaptively tolerate the presence of stragglers. 
However, this inevitably presents a significant challenge, as the number of stragglers is unknown in advance.
We address this challenge by introducing the notion of a query array, which indicates how to design private queries using the \emph{feasible PIR coding framework}.
We start with a simple example to illustrate the key idea behind the proposed adaptive PIR scheme.
In this example, we consider the same system parameters as in Example \ref{Flexible:example} to demonstrate how the feasible PIR framework in Example \ref{Flexible:example} can be combined with the query array to construct an adaptive PIR scheme.

\subsection{Illustrated Example for Adaptive PIR Scheme}

Consider the PIR problem with system parameters $N=8,K=X=T=2,M=3$, i.e., the user wishes to privately retrieve the file $\mathbf{W}^{(\theta)}$ from a secure distributed system that stores the $3$ files $\mathbf{W}^{(0)},\mathbf{W}^{(1)},\mathbf{W}^{(2)}$ across $8$ servers in a $2$-secure $2$-coded manner, while keeping the index $\theta$ private from any $2$ colluding servers. Here, we have $\lambda=3$ and $P=18$.
Therefore, 
each of files can be denoted by a matrix of dimensions $18\times 2$, given by
\begin{IEEEeqnarray}{rCl}\label{Example2:file:symbols}
\mathbf{W}^{(0)} =\left[
  \begin{array}{@{\;}cc@{\;}}
w^{(0)}_{0,0} &  w^{(0)}_{0,1}  \\
w^{(0)}_{1,0} &  w^{(0)}_{1,1}  \\
 \vdots &  \vdots \\
w^{(0)}_{17,0} &  w^{(0)}_{17,1} \\
\end{array}
\right], \quad
\mathbf{W}^{(1)} =\left[
  \begin{array}{@{\;}cc@{\;}}
w^{(1)}_{0,0} &  w^{(1)}_{0,1}  \\
w^{(1)}_{1,0} &  w^{(1)}_{1,1}  \\
 \vdots &  \vdots \\
w^{(1)}_{17,0} &  w^{(1)}_{17,1} \\
\end{array}
\right], \quad
\mathbf{W}^{(2)} =\left[
  \begin{array}{@{\;}cc@{\;}}
w^{(2)}_{0,0} &  w^{(2)}_{0,1}  \\
w^{(2)}_{1,0} &  w^{(2)}_{1,1}  \\
 \vdots &  \vdots \\
w^{(2)}_{17,0} &  w^{(2)}_{17,1} \\
\end{array}
\right].
\end{IEEEeqnarray}

Similar to Example \ref{Flexible:example}, 
let $\alpha_0,\alpha_1,\ldots,\alpha_{7}$ denote the encoding parameters of the \emph{feasible PIR coding framework}, and $\{b_{i,0}(x),b_{i,1}(x),b_{i,2}(x),b_{i,3}(x)\}_{i\in[0:3)}$ and $\{v_{i,k,0}^{(\mathcal{R})}(x),v_{i,k,1}^{(\mathcal{R})}(x),v_{i,k,2}^{(\mathcal{R})}(x)\}_{i\in\mathcal{R},k\in[0:2)}$ denote its encoding functions for any subset $\mathcal{R}\subseteq[0:3)$ of size $1\leq r\leq 3$.
We will utilize the \emph{feasible PIR coding framework} to construct an adaptive PIR scheme that can simultaneously tolerate the presence of $S=0,1,2$ stragglers.

\subsubsection*{Secure Distributed Storage} Following a similar encoding method to \eqref{Example:framework:storage:1}--\eqref{Example:framework:storage:3},
we choose the random noises $\{{z}_{i,2}^{(m)},{z}_{i,3}^{(m)}\}_{i\in[0:18)}$ and encode the file $\mathbf{W}^{(m)}$ using the functions $\{b_{i,0}(x),b_{i,1}(x),b_{i,2}(x),b_{i,3}(x)\}_{i\in[0:3)}$ of the \emph{feasible PIR coding framework} 
for any $m\in[0:3)$, given by
\begin{IEEEeqnarray}{rCl}
f_0^{(m)}(x)&=&{w}_{0,0}^{(m)}\cdot b_{0,0}(x)+{w}_{0,1}^{(m)}\cdot b_{0,1}(x)+{z}_{0,2}^{(m)}\cdot b_{0,2}(x)+{z}_{0,3}^{(m)}\cdot b_{0,3}(x), \label{Aexample:storage:1}\\
f_1^{(m)}(x)&=&{w}_{1,0}^{(m)}\cdot b_{1,0}(x)+{w}_{1,1}^{(m)}\cdot b_{1,1}(x)+{z}_{1,2}^{(m)}\cdot b_{1,2}(x)+{z}_{1,3}^{(m)}\cdot b_{1,3}(x), \\
f_2^{(m)}(x)&=&{w}_{2,0}^{(m)}\cdot b_{2,0}(x)+{w}_{2,1}^{(m)}\cdot b_{2,1}(x)+{z}_{2,2}^{(m)}\cdot b_{2,2}(x)+{z}_{2,3}^{(m)}\cdot b_{2,3}(x), \\ 
f_3^{(m)}(x)&=&{w}_{3,0}^{(m)}\cdot b_{0,0}(x)+{w}_{3,1}^{(m)}\cdot b_{0,1}(x)+{z}_{3,2}^{(m)}\cdot b_{0,2}(x)+{z}_{3,3}^{(m)}\cdot b_{0,3}(x), \\
f_4^{(m)}(x)&=&{w}_{4,0}^{(m)}\cdot b_{1,0}(x)+{w}_{4,1}^{(m)}\cdot b_{1,1}(x)+{z}_{4,2}^{(m)}\cdot b_{1,2}(x)+{z}_{4,3}^{(m)}\cdot b_{1,3}(x), \\
f_5^{(m)}(x)&=&{w}_{5,0}^{(m)}\cdot b_{2,0}(x)+{w}_{5,1}^{(m)}\cdot b_{2,1}(x)+{z}_{5,2}^{(m)}\cdot b_{2,2}(x)+{z}_{5,3}^{(m)}\cdot b_{2,3}(x), \\ 
f_6^{(m)}(x)&=&{w}_{6,0}^{(m)}\cdot b_{0,0}(x)+{w}_{6,1}^{(m)}\cdot b_{0,1}(x)+{z}_{6,2}^{(m)}\cdot b_{0,2}(x)+{z}_{6,3}^{(m)}\cdot b_{0,3}(x), \\
f_7^{(m)}(x)&=&{w}_{7,0}^{(m)}\cdot b_{1,0}(x)+{w}_{7,1}^{(m)}\cdot b_{1,1}(x)+{z}_{7,2}^{(m)}\cdot b_{1,2}(x)+{z}_{7,3}^{(m)}\cdot b_{1,3}(x), \\
f_8^{(m)}(x)&=&{w}_{8,0}^{(m)}\cdot b_{2,0}(x)+{w}_{8,1}^{(m)}\cdot b_{2,1}(x)+{z}_{8,2}^{(m)}\cdot b_{2,2}(x)+{z}_{8,3}^{(m)}\cdot b_{2,3}(x), \\  
f_9^{(m)}(x)&=&{w}_{9,0}^{(m)}\cdot b_{0,0}(x)+{w}_{9,1}^{(m)}\cdot b_{0,1}(x)+{z}_{9,2}^{(m)}\cdot b_{0,2}(x)+{z}_{9,3}^{(m)}\cdot b_{0,3}(x), \\
f_{10}^{(m)}(x)&=&{w}_{10,0}^{(m)}\cdot b_{1,0}(x)+{w}_{10,1}^{(m)}\cdot b_{1,1}(x)+{z}_{10,2}^{(m)}\cdot b_{1,2}(x)+{z}_{10,3}^{(m)}\cdot b_{1,3}(x), \\
f_{11}^{(m)}(x)&=&{w}_{11,0}^{(m)}\cdot b_{2,0}(x)+{w}_{11,1}^{(m)}\cdot b_{2,1}(x)+{z}_{11,2}^{(m)}\cdot b_{2,2}(x)+{z}_{11,3}^{(m)}\cdot b_{2,3}(x), \\  
f_{12}^{(m)}(x)&=&{w}_{12,0}^{(m)}\cdot b_{0,0}(x)+{w}_{12,1}^{(m)}\cdot b_{0,1}(x)+{z}_{12,2}^{(m)}\cdot b_{0,2}(x)+{z}_{12,3}^{(m)}\cdot b_{0,3}(x), \\
f_{13}^{(m)}(x)&=&{w}_{13,0}^{(m)}\cdot b_{1,0}(x)+{w}_{13,1}^{(m)}\cdot b_{1,1}(x)+{z}_{13,2}^{(m)}\cdot b_{1,2}(x)+{z}_{13,3}^{(m)}\cdot b_{1,3}(x), \\
f_{14}^{(m)}(x)&=&{w}_{14,0}^{(m)}\cdot b_{2,0}(x)+{w}_{14,1}^{(m)}\cdot b_{2,1}(x)+{z}_{14,2}^{(m)}\cdot b_{2,2}(x)+{z}_{14,3}^{(m)}\cdot b_{2,3}(x), \\  
f_{15}^{(m)}(x)&=&{w}_{15,0}^{(m)}\cdot b_{0,0}(x)+{w}_{15,1}^{(m)}\cdot b_{0,1}(x)+{z}_{15,2}^{(m)}\cdot b_{0,2}(x)+{z}_{15,3}^{(m)}\cdot b_{0,3}(x), \\
f_{16}^{(m)}(x)&=&{w}_{16,0}^{(m)}\cdot b_{1,0}(x)+{w}_{16,1}^{(m)}\cdot b_{1,1}(x)+{z}_{16,2}^{(m)}\cdot b_{1,2}(x)+{z}_{16,3}^{(m)}\cdot b_{1,3}(x), \\
f_{17}^{(m)}(x)&=&{w}_{17,0}^{(m)}\cdot b_{2,0}(x)+{w}_{17,1}^{(m)}\cdot b_{2,1}(x)+{z}_{17,2}^{(m)}\cdot b_{2,2}(x)+{z}_{17,3}^{(m)}\cdot b_{2,3}(x). \label{Aexample:storage:2}
\end{IEEEeqnarray}
Then, the data stored at the server $n,n\in[0:8)$ is given by evaluating these encoding polynomials at $x=\alpha_n$, i.e.,
\begin{IEEEeqnarray}{c}\label{Example:storage}
\mathcal{Y}_n=\left\{
  \begin{array}{@{}ccc@{}}
f_{0}^{(0)}(a_n), & f_{0}^{(1)}(a_n), & f_{0}^{(2)}(a_n) \\
f_{1}^{(0)}(a_n), & f_{1}^{(1)}(a_n), & f_{1}^{(2)}(a_n) \\
\vdots  & \vdots & \vdots \\
f_{17}^{(0)}(a_n), & f_{17}^{(1)}(a_n), & f_{17}^{(2)}(a_n) \\
\end{array}
\right\}.
\end{IEEEeqnarray}

\subsubsection*{Query Phase} 
To design private queries that provide the adaptive guarantee simultaneously for all $S=0,1,2$ stragglers,
we construct a query array of dimensions $3\times 18$, composed of three subarrays $\mathbf{U}^{0},\mathbf{U}^{1}$, and $\mathbf{U}^2$, given by
\begin{IEEEeqnarray}{c}\label{Example2:query:array}
\mathbf{U}=\left[
  \begin{array}{c;{2pt/2pt}c;{2pt/2pt}c}
  \underbrace{
  \begin{array}{cccccc}
   0 & 3 & 6 & 9 &  12 & 15 \\
   1 & 4 & 7 & 10 & 13 & 16 \\
   2 & 5 & 8 & 11 & 14 & 17 \\
  \end{array}
  }_{=\mathbf{U}^{0}} & 
    \underbrace{
  \begin{array}{ccc}
  * & 0 & 9 \\
  4 & * & 13 \\
  8 & 17 & * \\
   \end{array}
  }_{=\mathbf{U}^{1}}
   &
   \underbrace{
  \begin{array}{ccccccccc}
  * & * & 3 & * & * & 12 & * & * & 0 \\
  7 & * & * & 16 & * & * & 13 & * & * \\
  * & 2 & * & * & 11 & * & * & 8 & * \\
  \end{array}
  }_{=\mathbf{U}^{2}}
 \end{array}
\right].    
\end{IEEEeqnarray}

For convenience, let $\Gamma^h$ denote the number of columns in the subarray $\mathbf{U}^{h}$ for any $h\in[0:3)$, then we have $\Gamma^0=6,\Gamma^1=3$, and $\Gamma^2=9$. 
Clearly, the integer elements in each column of the query array form a subset of $[0:18)$, and under the modulo $\lambda=3$ operation, they form a subset of $[0:3)$, i.e., ${\mathcal{R}}^{h}_j\subseteq[0:18)$ and $\widetilde{\mathcal{R}}^{h}_j\subseteq[0:3)$ for any $h\in[0:3)$ and $j\in[0:\Gamma^h)$, where $\mathcal{R}^{h}_j$ denotes the set consisting of the integer elements in the $j$-th column of the subarray $\mathbf{U}^{h}$, and $\widetilde{\mathcal{R}}^{h}_j$ denotes the set consisting of the integer elements in $\mathcal{R}^{h}_j$ modulo $\lambda$.
Moreover, for any $h\in[0:3)$ and $j\in[0:\Gamma^h)$, we observe from \eqref{Example:framework:storage:1}--\eqref{Example:framework:storage:3} and \eqref{Aexample:storage:1}--\eqref{Aexample:storage:2} that the partial file $\mathbf{W}^{(m)}(\mathcal{R}_j^h,:)$ is encoded in exactly the same way as $\widetilde{\mathbf{W}}^{(m)}(\widetilde{\mathcal{R}}^{h}_j,:)$ in \eqref{example:files} and the only difference between them lies in the values of the files $\mathbf{W}^{(m)}(\mathcal{R}_j^h,:)$ and $\widetilde{\mathbf{W}}^{(m)}(\widetilde{\mathcal{R}}^{h}_j,:)$, for any $m\in[0:3)$. Recall that the \emph{feasible PIR coding framework} is constructed regardless of the values of the data files $\widetilde{\mathbf{W}}^{(0)}(\widetilde{\mathcal{R}}^{h}_j,:),\widetilde{\mathbf{W}}^{(1)}(\widetilde{\mathcal{R}}^{h}_j,:),\widetilde{\mathbf{W}}^{(2)}(\widetilde{\mathcal{R}}^{h}_j,:)$.
This ensures that the user can privately retrieve the desired partial file $\mathbf{W}^{(\theta)}(\mathcal{R}_j^h,:)$ by following the query, response, and decoding steps similar to those in Example \ref{Flexible:example}.


More precisely, to privately retrieve the partial file $\mathbf{W}^{(\theta)}(\mathcal{R}_j^h,:)$ for any $h\in[0:3)$ and $j\in[0:\Gamma^h)$, similar to \eqref{Example:framework:query:polynomial}, the user chooses the random noises $z_{(i)_{\lambda},k,0}^{(m,\widetilde{\mathcal{R}}^{h}_j)},z_{(i)_{\lambda},k,1}^{(m,\widetilde{\mathcal{R}}^{h}_j)}$ and designs the query function $q_{i,k}^{(m,\mathcal{R}^{h}_j)}(x)$ using the encoding functions $v_{(i)_{\lambda},k,0}^{(\widetilde{\mathcal{R}}^{h}_j)}(x),v_{(i)_{\lambda},k,1}^{(\widetilde{\mathcal{R}}^{h}_j)}(x),v_{(i)_{\lambda},k,2}^{(\widetilde{\mathcal{R}}^{h}_j)}(x)$
for any $m\in[0:3),i\in\mathcal{R}^{h}_j$, and $k\in[0:2)$, given by
\begin{IEEEeqnarray}{c}
q_{i,k}^{(m,\mathcal{R}^{h}_j)}(x)=
z_{(i)_{\lambda},k,0}^{(m,\widetilde{\mathcal{R}}^{h}_j)}\cdot v_{(i)_{\lambda},k,0}^{(\widetilde{\mathcal{R}}^{h}_j)}(x)+
z_{(i)_{\lambda},k,1}^{(m,\widetilde{\mathcal{R}}^{h}_j)}\cdot v_{(i)_{\lambda},k,1}^{(\widetilde{\mathcal{R}}^{h}_j)}(x)
+\left\{
\begin{array}{@{}l@{\;\;}l}
v_{(i)_{\lambda},k,2}^{(\widetilde{\mathcal{R}}^{h}_j)}(x),
&\mathrm{if}\,\, m=\theta \\
0, & \mathrm{otherwise}
\end{array} \right.. \notag
\end{IEEEeqnarray}
Then, the query sent to the server $n,n\in[0:8)$ for retrieving the partial file $\mathbf{W}^{(\theta)}(\mathcal{R}_j^h,:)$ is given by
\begin{IEEEeqnarray}{c}
\mathcal{Q}_n^{(\theta,{\mathcal{R}}_j^h)}=\left\{q_{i,k}^{(m,\mathcal{R}^{h}_j)}(\alpha_n):m\in[0:3),i\in\mathcal{R}^{h}_j,k\in[0:2)\right\}. \notag
\end{IEEEeqnarray}

In order to adaptively retrieve the desired file, 
the user combines the queries across all the columns of the query array in \eqref{Example2:query:array} and sends the following query $\mathcal{Q}^{(\theta)}_n$ to the server $n,n\in[0:8)$: 
\begin{IEEEeqnarray}{c}\label{Example2:query}
\mathcal{Q}_n^{(\theta)}=\left\{
  \begin{array}{@{}lll@{}}
\mathcal{Q}_n^{(\theta,\{0,1,2\})}, & \mathcal{Q}_n^{(\theta,\{3,4,5\})}, & \mathcal{Q}_n^{(\theta,\{6,7,8\})} \\
\mathcal{Q}_n^{(\theta,\{9,10,11\})}, & \mathcal{Q}_n^{(\theta,\{12,13,14\})}, & \mathcal{Q}_n^{(\theta,\{15,16,17\})} \\
\mathcal{Q}_n^{(\theta,\{4,8\})}, & \mathcal{Q}_n^{(\theta,\{0,17\})}, & \mathcal{Q}_n^{(\theta,\{9,13\})} \\
\mathcal{Q}_n^{(\theta,\{7\})}, & \mathcal{Q}_n^{(\theta,\{2\})}, & \mathcal{Q}_n^{(\theta,\{3\})} \\
\mathcal{Q}_n^{(\theta,\{16\})}, & \mathcal{Q}_n^{(\theta,\{11\})}, & \mathcal{Q}_n^{(\theta,\{12\})} \\
\mathcal{Q}_n^{(\theta,\{13\})}, & \mathcal{Q}_n^{(\theta,\{8\})}, & \mathcal{Q}_n^{(\theta,\{0\})}
\end{array}
\right\}.
\end{IEEEeqnarray}

\subsubsection*{Answer Phase} 
For any given query $\mathcal{Q}_n^{(\theta,{\mathcal{R}}_j^h)}$ with $h\in[0:3)$ and $j\in[0:\Gamma^h)$, similar to \eqref{Example:PIR:answer:1}-\eqref{Example:PIR:answer:2}, the server $n,n\in[0:8)$ can generate the following response using the local storage in \eqref{Example:storage}:
\begin{IEEEeqnarray}{c}
\mathcal{A}_n^{(\theta,\mathcal{R}_j^h)}=\left\{ A_{n,k}^{(\theta,\mathcal{R}_j^h)}:k\in[0:2)\right\},\notag
\end{IEEEeqnarray}
where $A_{n,k}^{(\theta,\mathcal{R}_j^h)}$ is given by 
\begin{IEEEeqnarray}{c}
A_{n,k}^{(\theta,\mathcal{R}_j^h)}=\sum\limits_{m\in[0:3)}\sum\limits_{i\in\mathcal{R}_{j}^{h}}q_{i,k}^{(m,{\mathcal{R}}^{h}_j)}(\alpha_n)\cdot {f}_{i}^{(m)}(\alpha_n),\quad\forall\,k\in[0:2).\notag
\end{IEEEeqnarray}

Upon receiving the query in \eqref{Example2:query},
the server $n,n\in[0:8)$ will sequentially return the following responses to the user in the given order until it recovers the desired file:
\begin{IEEEeqnarray}{c}
  \begin{array}{@{}lll@{}}
\mathcal{A}_n^{(\theta,\{0,1,2\})}, & \mathcal{A}_n^{(\theta,\{3,4,5\})}, & \mathcal{A}_n^{(\theta,\{6,7,8\})} \\
\mathcal{A}_n^{(\theta,\{9,10,11\})}, & \mathcal{A}_n^{(\theta,\{12,13,14\})}, & \mathcal{A}_n^{(\theta,\{15,16,17\})} \\
\mathcal{A}_n^{(\theta,\{4,8\})}, & \mathcal{A}_n^{(\theta,\{0,17\})}, & \mathcal{A}_n^{(\theta,\{9,13\})} \\
\mathcal{A}_n^{(\theta,\{7\})}, & \mathcal{A}_n^{(\theta,\{2\})}, & \mathcal{A}_n^{(\theta,\{3\})} \\
\mathcal{A}_n^{(\theta,\{16\})}, & \mathcal{A}_n^{(\theta,\{11\})}, & \mathcal{A}_n^{(\theta,\{12\})} \\
\mathcal{A}_n^{(\theta,\{13\})}, & \mathcal{A}_n^{(\theta,\{8\})}, & \mathcal{A}_n^{(\theta,\{0\})}.
\end{array}\notag
\end{IEEEeqnarray}

\subsubsection*{Decoding Phase} 
Since $|{\mathcal{R}}^{h}_j|=|\widetilde{\mathcal{R}}^{h}_j|=3-h$ for any given $h\in[0:3)$ and $j\in[0:\Gamma^h)$,
similar to the decoding phase in Example \ref{Flexible:example}, we can obtain the decoding property of the partial files.
Specifically, the user can decode the partial file $\mathbf{W}^{(\theta)}(\mathcal{R}_j^h,:)$ from any $8-h-d$ out of the $8$ responses ${\mathcal{A}}_{[0:8)}^{(\theta,\mathcal{R}_j^h)}$ after obtaining any $d$ rows of $\mathbf{W}^{(\theta)}(\mathcal{R}_j^h,:)$ for any $d\in[0:3-h)$. This property will be employed to adaptively decode the desired file, as shown in the following decoding steps.

\textbf{In the initial scenario,} assume that there are no stragglers. 
Note from the query array in \eqref{Example2:query:array} that the integer elements in the $0$-th subarray $\mathbf{U}^{0}$ contain all elements in the range $0$ to $17$. This means that the user can recover the file $\mathbf{W}^{(\theta)}$ once decoding the desired partial files associated with the integer elements in each column of the subarray $\mathbf{U}^{0}$ by \eqref{Example2:file:symbols}.
Therefore, in order to retrieve the desired file, the user will wait for the first $6$ responses from each server $n,n\in[0:8)$, given by 
\begin{IEEEeqnarray}{c}
\left\{
\begin{array}{@{}lll@{}}
\mathcal{A}_n^{(\theta,\{0,1,2\})}, & \mathcal{A}_n^{(\theta,\{3,4,5\})}, & \mathcal{A}_n^{(\theta,\{6,7,8\})} \\
\mathcal{A}_n^{(\theta,\{9,10,11\})}, & \mathcal{A}_n^{(\theta,\{12,13,14\})}, & \mathcal{A}_n^{(\theta,\{15,16,17\})}
\end{array}
\right\}_{n\in[0:8)}.\notag
\end{IEEEeqnarray}
By the decoding property of the partial files, the partial file $\mathbf{W}^{(\theta)}(\{0,1,2\},:)$ can be decoded from the received responses $\{\mathcal{A}_n^{(\theta,\{0,1,2\})}\}_{n\in[0:8)}$. Similarly, the user can decode  
$\mathbf{W}^{(\theta)}(\{3,\!4,\!5\},:),\!\!\mathbf{W}^{(\theta)}(\{6,\!7,\!8\},:),\!\!\mathbf{W}^{(\theta)}(\{9,\!10,\!11\},:),\!\!$ $\mathbf{W}^{(\theta)}(\{12,\!13,\!14\},:),\!\mathbf{W}^{(\theta)}(\{15,\!16,\!17\},:)$ and then recover the entire file $\mathbf{W}^{(\theta)}$.

\textbf{In the presence of $S=1$ straggler}, without loss of generality, assume that server $0$ is the straggler.
By \eqref{Example2:query:array} again, we find that for each column in the subarray $\mathbf{U}^{0}$, there exists an integer element that appears in the subsequent subarray $\mathbf{U}^{1}$. 
This ensures that when there is $S=1$ straggler, the missing responses required for decoding the partial files associated with the subarray $\mathbf{U}^{0}$ can be compensated by utilizing the partial files associated with the subarray $\mathbf{U}^{1}$. Therefore, in this straggler case, the user will continue to wait for the other three responses $\mathcal{A}_n^{(\theta,\{4,8\})},\mathcal{A}_n^{(\theta,\{0,17\})},\mathcal{A}_n^{(\theta,\{9,13\})}$ corresponding to the queries designed using the subarray $\mathbf{U}^{1}$  from the remaining non-straggler server $n,n\in[0:8)\backslash\{0\}$. In total, the user obtains the following responses:
\begin{IEEEeqnarray}{c}
\left\{
\begin{array}{@{}lll@{}}
\mathcal{A}_n^{(\theta,\{0,1,2\})}, & \mathcal{A}_n^{(\theta,\{3,4,5\})}, & \mathcal{A}_n^{(\theta,\{6,7,8\})} \\
\mathcal{A}_n^{(\theta,\{9,10,11\})}, & \mathcal{A}_n^{(\theta,\{12,13,14\})}, & \mathcal{A}_n^{(\theta,\{15,16,17\})} \\
\mathcal{A}_n^{(\theta,\{4,8\})}, & \mathcal{A}_n^{(\theta,\{0,17\})}, & \mathcal{A}_n^{(\theta,\{9,13\})}
\end{array}
\right\}_{n\in[0:8)\backslash\{0\}}.\notag
\end{IEEEeqnarray}
Following the decoding property of the partial files again, the user can first decode the partial files $\mathbf{W}^{(\theta)}(\{4,8\},:),$ $\mathbf{W}^{(\theta)}(\{0,17\},:),\mathbf{W}^{(\theta)}(\{9,13\},:)$ from the responses received, and then decode the partial file $\mathbf{W}^{(\theta)}(\{0,1,2\},:)$ from the received responses $\{\mathcal{A}_n^{(\theta,\{0,1,2\})}\}_{n\in[0:8)\backslash\{0\}}$ along with the previously decoded data $\mathbf{W}^{(\theta)}(\{0\},:)$. Similarly, the partial files $\mathbf{W}^{(\theta)}(\{3,\!4,\!5\},:),\!\!\mathbf{W}^{(\theta)}(\{6,\!7,\!8\},:),\!\!\mathbf{W}^{(\theta)}(\{9,\!10,\!11\},:),\!\!\mathbf{W}^{(\theta)}(\{12,\!13,\!14\},:),\!\!$ $\mathbf{W}^{(\theta)}(\{15,16,17\},:)$ can be decoded. Therefore, the user recovers the desired file $\mathbf{W}^{(\theta)}$.

\textbf{In the presence of $S=2$ stragglers}, without loss of generality, assume that server $1$ is another straggler. Similar to the above case of $S=1$, we further find from \eqref{Example2:query:array} that for each column in the first two subarrays $\mathbf{U}^{0},\mathbf{U}^{1}$, there exists another integer element that appears in the subsequent subarray $\mathbf{U}^{2}$. This ensures that when there are $S=2$ stragglers, the user can decode the desired file by continuing to wait for the responses corresponding to the queries designed using the subarray $\mathbf{U}^{2}$.
More specifically, the user will continue to wait for the other $9$ responses $\mathcal{A}_n^{(\theta,\{7\})},\mathcal{A}_n^{(\theta,\{2\})}, \mathcal{A}_n^{(\theta,\{3\})},\mathcal{A}_n^{(\theta,\{16\})}, \mathcal{A}_n^{(\theta,\{11\})}, \mathcal{A}_n^{(\theta,\{12\})}, \mathcal{A}_n^{(\theta,\{13\})}, \mathcal{A}_n^{(\theta,\{8\})},\mathcal{A}_n^{(\theta,\{0\})}$ from the remaining non-straggler
server $n,n\in[0:8)\backslash\{0,1\}$, i.e., the user obtains the following responses:
\begin{IEEEeqnarray}{c}
\left\{
  \begin{array}{@{}lll@{}}
\mathcal{A}_n^{(\theta,\{0,1,2\})}, & \mathcal{A}_n^{(\theta,\{3,4,5\})}, & \mathcal{A}_n^{(\theta,\{6,7,8\})} \\
\mathcal{A}_n^{(\theta,\{9,10,11\})}, & \mathcal{A}_n^{(\theta,\{12,13,14\})}, & \mathcal{A}_n^{(\theta,\{15,16,17\})} \\
\mathcal{A}_n^{(\theta,\{4,8\})}, & \mathcal{A}_n^{(\theta,\{0,17\})}, & \mathcal{A}_n^{(\theta,\{9,13\})} \\
\mathcal{A}_n^{(\theta,\{7\})}, & \mathcal{A}_n^{(\theta,\{2\})}, & \mathcal{A}_n^{(\theta,\{3\})} \\
\mathcal{A}_n^{(\theta,\{16\})}, & \mathcal{A}_n^{(\theta,\{11\})}, & \mathcal{A}_n^{(\theta,\{12\})} \\
\mathcal{A}_n^{(\theta,\{13\})}, & \mathcal{A}_n^{(\theta,\{8\})}, & \mathcal{A}_n^{(\theta,\{0\})}
\end{array}
\right\}_{n\in[0:8)\backslash\{0,1\}}.\notag
\end{IEEEeqnarray}
The decoding is also similar to the case of $S=1$. 
The user first decodes the partial files $\mathbf{W}^{(\theta)}(\{7\},:),\!\mathbf{W}^{(\theta)}(\{2\},:),$ $\mathbf{W}^{(\theta)}(\{3\},:),\mathbf{W}^{(\theta)}(\{16\},:),\mathbf{W}^{(\theta)}(\{11\},:),\mathbf{W}^{(\theta)}(\{12\},:),\mathbf{W}^{(\theta)}(\{13\},:),\mathbf{W}^{(\theta)}(\{8\},:),\mathbf{W}^{(\theta)}(\{0\},:)$ from the received responses, and then decodes
$\mathbf{W}^{(\theta)}(\{4,8\},:),\mathbf{W}^{(\theta)}(\{0,17\},:),\mathbf{W}^{(\theta)}(\{9,13\},:)$ from the received responses along with the previously decoded partial files. Furthermore, having obtained these partial files, the user can decode $\mathbf{W}^{(\theta)}(\{0,\!1,\!2\},:),\mathbf{W}^{(\theta)}(\{3,\!4,\!5\},:),\!\mathbf{W}^{(\theta)}(\{6,\!7,\!8\},:),\mathbf{W}^{(\theta)}(\{9,\!10,\!11\},:),\mathbf{W}^{(\theta)}(\{12,\!13,\!14\},:),\mathbf{W}^{(\theta)}(\{15,\!16,\!17\},:)$ from the received responses and thus complete the decoding of the desired file $\mathbf{W}^{(\theta)}$.


In summary, after sending the queries in \eqref{Example2:query} to the servers, the user will keep waiting for the servers' responses. If the user receives $6$ responses from each of the $8$ servers, $9$ responses from each of any $7$ servers, or $18$ responses from each of any $6$ servers, then the desired file can be successfully decoded. This ensures that the user can adaptively tolerate the presence of $S=0,1,2$ stragglers, even if their identities and numbers are unknown in advance and may change over time.

Table \ref{tab:ASPLC:example} presents the retrieval rate of the adaptive PIR example and compares it to the traditional PIR example with a known number of stragglers \cite{jia2020x,jia2022xFSL}. 
Compared to the traditional PIR example with a known number of stragglers, the proposed adaptive PIR example can simultaneously tolerate the presence of $S=0,1,2$ stragglers while achieving the maximum retrieval rate based on the actual number of stragglers.

\begin{table*}[htbp]
\centering
\caption{Comparison of the retrieval rate between the adaptive PIR example and the traditional PIR example with a known number of stragglers \cite{jia2020x,jia2022xFSL}.}\label{tab:ASPLC:example}
\begin{tabular}{|c|c|c|c|c|}
\hline
Actual Number of Stragglers & PIR  with $0$ Straggler & PIR  with $1$ Straggler & PIR  with $2$ Stragglers  & Our Adaptive PIR \\ \hline
$S=0$ & $\frac{3}{8}$ & $\frac{2}{7}$ & $\frac{1}{6}$ & $\frac{3}{8}$ \\ \hline
$S=1$ & \cmarkk & $\frac{2}{7}$ & $\frac{1}{6}$ & $\frac{2}{7}$ \\ \hline
$S=2$ & \cmarkk & \cmarkk & $\frac{1}{6}$ & $\frac{1}{6}$ \\ \hline
\end{tabular}
\end{table*}

\subsection{General Construction of Adaptive PIR Scheme}
In this subsection, we present the general construction of the adaptive PIR scheme that can tolerate the presence of $S$ stragglers simultaneously for all $0\leq S\leq \lambda-1$. 

\subsubsection*{Query Phase} 
In the query design phase, to provide the adaptive guarantee, we construct a query array with dimensions $\lambda\times P$. This query array is composed of $\lambda$ subarrays, denoted as $\mathbf{U}^{0},\mathbf{U}^{1},\ldots,\mathbf{U}^{\lambda-1}$, with elements drawn from the range $0$ to $P-1$ along with a special symbol ``*''.\footnote{The special symbol ``*'' has no practical significance; it merely serves to fill the array and distinguish itself from the integer elements.}
Given the query array, the private queries are designed to retrieve the desired partial files associated with the integer elements in each column of the query array, following an approach similar to the query design in the \emph{feasible PIR coding framework}. 
The private queries designed using the $h$-th subarray $\mathbf{U}^{h}$ are used to provide robustness against $S=h$ stragglers for any given $h\in[0:\lambda)$. These queries across all the query subarrays are then combined to provide the adaptive guarantee.

More formally, the query array $\mathbf{U}$ with dimensions $\lambda\times P$, composed of the $\lambda$ subarrays  $\mathbf{U}^{0},\mathbf{U}^{1},\ldots,\mathbf{U}^{\lambda-1}$, is defined as
\begin{IEEEeqnarray}{rCl}\label{query:array:subarray}
\mathbf{U}=\left[
  \begin{array}{@{\;}c;{2pt/2pt}c;{2pt/2pt}c;{2pt/2pt}c@{\;}}
    \mathbf{U}^{0} & \mathbf{U}^{1} & \cdots & \mathbf{U}^{\lambda-1}  \\
\end{array}
\right],
\end{IEEEeqnarray}
where the subarray $\mathbf{U}^{h}$ has dimensions $\lambda\times\Gamma^h$ for all $h\in[0:\lambda)$ and $\Gamma^{h}$ represents the number of columns in the subarray $\mathbf{U}^{h}$, given by
\begin{IEEEeqnarray}{c}\label{number:column}
    \Gamma^h\triangleq\left\{
\begin{array}{@{}ll}
\frac{P}{\lambda},   &\text{if}~h=0 \\
\frac{P}{(\lambda-h)(\lambda-h+1)},&\text{if}~h\in[1:\lambda)
\end{array}\right., \quad\forall\,h\in[0:\lambda).
\end{IEEEeqnarray}
Clearly, the dimensions of these subarrays are feasible due to $\lambda|P$ and $(\lambda-h)(\lambda-h+1)|P$ by \eqref{number:value} along with the condition $\sum_{h=0}^{\lambda-1}\Gamma^h=P$.

Our goal is to design the query array $\mathbf{U}$ satisfying the following four conditions. 
These conditions play a crucial role in enabling the private queries designed based on this query array to adaptively tolerate the presence of stragglers.
\begin{enumerate}
\item[C0.] The elements of the array $\mathbf{U}$ are either integers in the range $0$ to $P-1$ or the special symbol ``*''.
\item[C1.] The number of integer elements in each column of the subarray $\mathbf{U}^{h}$ is $\lambda-h$ for any $h\in[0:\lambda)$, and particularly the integer elements in each column of $\mathbf{U}^{h}$ are pairwise distinct under the modulo $\lambda$ operation. More precisely, for any given $h\in[0:\lambda)$ and $j\in[0:\Gamma^h)$, let $\mathcal{R}^{h}_j$ denote the set consisting of the integer elements in the $j$-th column of the subarray $\mathbf{U}^{h}$, and $\widetilde{\mathcal{R}}^{h}_j$ denote the set consisting of those integer elements taken modulo $\lambda$, i.e., 
$\widetilde{\mathcal{R}}^{h}_j=\{(a)_{\lambda}:a\in\mathcal{R}^{h}_j\}$. Then, we have $|\mathcal{R}_j^{h}|=|\widetilde{\mathcal{R}}^{h}_j|=\lambda-h$.
\item[C2.] The integer elements of the subarray $\mathbf{U}^{0}$ include all integers in the range $0$ to $P-1$, i.e., $\bigcup_{j\in[0:\Gamma^{0})}\mathcal{R}_j^{0}=\{0,1,\ldots,P-1\}$.
\item[C3.] In the $j$-th column of the subarray $\mathbf{U}^{h}$ for any $h\in[0:\lambda-1)$ and $j\in[0:\Gamma^h)$,
there exist $\lambda-h-1$ integer elements that appear sequentially in the subsequent $\lambda-h-1$ subarrays $\mathbf{U}^{h+1},\mathbf{U}^{h+2},\ldots,\mathbf{U}^{\lambda-1}$, i.e., there exists a subset $\mathcal{D}^{h}_j=\{a_0,a_1,\ldots,a_{\lambda-h-2}\}\subseteq\mathcal{R}^h_j$ of size $\lambda-h-1$ such that $a_{r-h-1}\in\bigcup_{j\in[0:\Gamma^{r})}\mathcal{R}_j^{r}$ for all $r\in[h+1:\lambda)$.
\end{enumerate}

Next, we provide a construction of the query array $\mathbf{U}$ satisfying the conditions C0--C3. 



Firstly, the subarray $\mathbf{U}^{0}$ of dimensions $\lambda\times\Gamma^0$ is constructed by assigning the $P$ values in $[0:P)$ to all $\lambda\cdot\Gamma^0=P$ elements of $\mathbf{U}^{0}$, given by
\begin{IEEEeqnarray}{c}\label{array:0}
    u_{i,j}^{0}=i+j\cdot\lambda, \quad\forall\,i\in[0:\lambda ),j\in[0:\Gamma^0).
\end{IEEEeqnarray}

Then, the remaining subarrays $\mathbf{U}^{1},\mathbf{U}^{2},\ldots,\mathbf{U}^{\lambda-1}$ are constructed in an iterative approach, meaning that $\mathbf{U}^{h}$ is constructed after completing the constructions of $\mathbf{U}^{0},\mathbf{U}^{1},\ldots,\mathbf{U}^{h-1}$ for all $h\in[1:\lambda)$. 
Given the subarrays $\mathbf{U}^{0},\mathbf{U}^{1},\ldots,\mathbf{U}^{h-1}$, we describe the construction of the subarray $\mathbf{U}^{h}$ with dimensions $\lambda\times\Gamma^h$ in the following two steps for all $h\in[1:\lambda)$:
\begin{itemize}
\item In the $j$-th column of the subarray $\mathbf{U}^{h}$ for any given $j\in[0:\Gamma^h)$, the $h$ elements with row indices $(j)_{\lambda},(j-1)_{\lambda},\ldots,(j-h+1)_{\lambda}$ are all set to ``*'', i.e.,
\begin{IEEEeqnarray}{c}\label{U0:cons}
    u_{\left(j-r\right)_{\lambda},j}^{h}=*,\quad\forall\,r\in[0:h),j\in[0:\Gamma^h).
\end{IEEEeqnarray}
Equivalently, in the $i$-th row of the subarray $\mathbf{U}^{h}$ for any given $i\in[0:\lambda)$, an element is set to ``*'' if its column index belongs to the set $\{(i+r)_{\lambda}+s\cdot\lambda:r\in[0:h),s\in[0:\frac{\Gamma^h}{\lambda})\}$, i.e., the construction in \eqref{U0:cons} can be equivalently expressed as
\begin{IEEEeqnarray}{c}\label{array:h:1}
    u_{i,(i+r)_{\lambda}+s\cdot\lambda}^{h}=*,\quad\forall\,i\in[0:\lambda),r\in[0:h),s\in[0:\frac{\Gamma^h}{\lambda}),
\end{IEEEeqnarray}
where $\lambda|\Gamma^h$ according to \eqref{number:value} and \eqref{number:column}.

Therefore, for the subarray $\mathbf{U}^{h}$, this step completes the construction of its elements with row index $i\in[0:\lambda )$ and column index $j\in\{(i+r)_{\lambda}+s\cdot\lambda:r\in[0:h),s\in[0:\frac{\Gamma^h}{\lambda})\}$. 
\item In the second step, we proceed to complete the construction of the remaining elements with row index $i\in[0:\lambda)$ and column index 
$j\in\{(i+r)_{\lambda}+s\cdot\lambda:r\in[h:\lambda),s\in[0:\frac{\Gamma^h}{\lambda})\}$
for the subarray $\mathbf{U}^{h}$.

We know from \eqref{query:array:subarray}-\eqref{number:column} that after providing the previous subarrays $\mathbf{U}^{0},\mathbf{U}^{1},\ldots,\mathbf{U}^{h-1}$, 
all the elements in the first $\Gamma^0+\Gamma^{1}+\ldots+\Gamma^{h-1}=(\lambda-h)\cdot\Gamma^h$ columns of the array $\mathbf{U}$ (i.e., $\{u_{i,j}:i\in[0:\lambda),j\in[0:(\lambda-h)\cdot\Gamma^h)\}$) are known. 
Therefore, in the $i$-th row of the subarray $\mathbf{U}^{h}$ for any given $i\in[0:\lambda)$, the elements with column indices in $\{(i+r)_{\lambda}+s\cdot\lambda:r\in[h:\lambda),s\in[0:\frac{\Gamma^h}{\lambda})\}$ can be assigned values from the known elements $\{u_{i,(i+h-1)_{\lambda}+t\cdot\lambda}:t\in[0:\frac{(\lambda-h)\cdot\Gamma^h}{\lambda})\}$ in the current row $i$ of the previous subarrays $\mathbf{U}^{0},\mathbf{U}^{1},\ldots,\mathbf{U}^{h-1}$. More precisely, we set
\begin{IEEEeqnarray}{c}\label{array:h:2}
    u_{i,(i+r)_{\lambda}+s\cdot\lambda}^{h}=u_{i,(i+h-1)_{\lambda}+(r-h+s(\lambda-h))\lambda}, \quad\forall\,i\in[0:\lambda),r\in[h:\lambda),s\in[0:\frac{\Gamma^h}{\lambda}).\IEEEeqnarraynumspace
\end{IEEEeqnarray}
\end{itemize}

By iteratively substituting $h=1,2,\ldots,\lambda-1$, we complete the construction of the array $\mathbf{U}$.
In summary, the overall procedures described above are outlined in Algorithm \ref{Query:array}.

\begin{Lemma}\label{lemma:query:array}
The constructed query array $\mathbf{U}$ of dimensions $\lambda\times P$ satisfies the four conditions C0--C3.
\end{Lemma}
\begin{proof}
The proof of this lemma is provided in Appendix-\ref{proof:C0-C3}.
\end{proof}

\begin{algorithm}[httb]
\caption{Construction of Query Array}
\label{Query:array}
\begin{algorithmic}[1] 
\REQUIRE System parameters $N,K,X,T$
\ENSURE Query array $\mathbf{U}$ of dimensions $\lambda\times P$ with $\lambda=N-(K+X+T-1)$ and $P=\lambda\cdot\mathrm{lcm}(1,2,\ldots,\lambda)$.
\STATE Let $\mathbf{U}^{0}$ denote a subarray of dimensions $\lambda\times\frac{P}{\lambda}$ and $\mathbf{U}^{h}$ denote another subarray of dimensions $\lambda\times\frac{P}{(\lambda-h)(\lambda-h+1)}$ for all $h\in[1:\lambda)$. Then set  $\mathbf{U}=\big[
  \begin{array}{@{}c;{2pt/2pt}c;{2pt/2pt}c;{2pt/2pt}c@{}}
    \mathbf{U}^{0} & \mathbf{U}^{1} & \cdots & \mathbf{U}^{\lambda-1}  \\
\end{array}
\big]$.
\FORALL{$i\in[0:\lambda ),j\in[0:\frac{P}{\lambda})$}
\STATE $u_{i,j}^{0}=i+j\cdot\lambda$,
\ENDFOR
\FORALL{$h=1,2,\ldots,\lambda-1$}
\FORALL{$i\in[0:\lambda),r\in[0:\lambda),s\in[0:\frac{P}{\lambda(\lambda-h)(\lambda-h+1)})$}
\IF{$r\in[0:h)$}
\STATE $u_{i,(i+r)_{\lambda}+s\cdot\lambda}^{h}=*$,
\ELSE
\STATE $u_{i,(i+r)_{\lambda}+s\cdot\lambda}^{h}=u_{i,(i+h-1)_{\lambda}+(r-h+s(\lambda-h))\lambda}$,
\ENDIF
\ENDFOR
\ENDFOR
\RETURN Array $\mathbf{U}=\big[
  \begin{array}{@{}c;{2pt/2pt}c;{2pt/2pt}c;{2pt/2pt}c@{}}
    \mathbf{U}^{0} & \mathbf{U}^{1} & \cdots & \mathbf{U}^{\lambda-1}  \\
\end{array}
\big]$.
\end{algorithmic}
\end{algorithm}

\begin{Example}
Assume that the parameter $\lambda=3$, which results in $P=18$. The query array $\mathbf{U}$ of dimensions $3\times 18$ is given in \eqref{Example2:query:array}.
In addition, when $\lambda=4$, the query array ${\mathbf{U}}$ of dimensions $4\times 48$ is given by
\begin{IEEEeqnarray}{c}\label{Example:APIR:query:array}
{\mathbf{U}}=\left[
  \begin{array}{c;{2pt/2pt}c;{2pt/2pt}c;{2pt/2pt}c}
    {\mathbf{U}}^{0} & {\mathbf{U}}^{1} & {\mathbf{U}}^{2} & {\mathbf{U}}^{3}  \\
\end{array}
\right]_{4\times 48},
\end{IEEEeqnarray}
where 
\begin{IEEEeqnarray*}{c}
{\mathbf{U}}^{0}=\left[
  \begin{array}{cccccccccccc}
   0 & 4 & 8  & 12 & 16 & 20 & 24 & 28 & 32 & 36 & 40 & 44 \\
   1 & 5 & 9  & 13 & 17 & 21 & 25 & 29 & 33 & 37 & 41 & 45 \\
   2 & 6 & 10 & 14 & 18 & 22 & 26 & 30 & 34 & 38 & 42 & 46 \\
   3 & 7 & 11 & 15 & 19 & 23 & 27 & 31 & 35 & 39 & 43 & 47
\end{array}
\right]_{4\times12},\\
{\mathbf{U}}^{1}=\left[
  \begin{array}{cccc}
 * & 0 & 16 & 32 \\
 5 & * & 21 & 37 \\
 10 & 26 & * & 42 \\
 15 & 31 & 47 & *
\end{array}
\right]_{4\times4}, \quad\quad
{\mathbf{U}}^{2}=\left[
  \begin{array}{cccccccc}
  * & * & 4 & 20 & * & * & 36 & 0  \\
  9 & * & * & 25 & 41 & * & * & 21 \\
  14 & 30 & * & * & 46 & 42 & * & * \\
  * & 3 & 19 & * & * & 35 & 15 & *\\
\end{array}
\right]_{4\times8}, \\
{\mathbf{U}}^{3}=\left[
  \begin{array}{c@{\;\;\,}c@{\;\;\,}c@{\;\;\,}c@{\;\;\,}c@{\;\;\,}c@{\;\;\,}c@{\;\;\,}c@{\;\;\,}c@{\;\;\,}c@{\;\;\,}c@{\;\;\,}c@{\;\;\,}c@{\;\;\,}c@{\;\;\,}c@{\;\;\,}c@{\;\;\,}c@{\;\;\,}c@{\;\;\,}c@{\;\;\,}c@{\;\;\,}c@{\;\;\,}c@{\;\;\,}c@{\;\;\,}c@{\;\;\,}c@{\;\;\,}c@{\;\;\,}c@{\;\;\,}c@{\;\;\,}c@{\;\;\,}c@{\;\;\,}c@{\;\;\,}c@{\;\;\,}c@{\;\;\,}c@{\;\;\,}c@{\;\;\,}c@{\;\;\,}c@{\;\;\,}c@{\;\;\,}c@{\;\;\,}c@{\;\;\,}c@{\;\;\,}c@{\;\;\,}c@{\;\;\,}c@{\;\;\,}c@{\;\;\,}c@{\;\;\,}c@{\;\;\,}c}
   * & * & * & 8 & * & * & * & 24 & * & * & * & 40 & * & * & * & 16 & * & * & * & 4 & * & * & * & 36 \\
   13 & * & * & * & 29 & * & * & * & 45 & * & * & * & 37 & * & * & * & 25 & * & * & * & 21 & * & * & *\\
   * & 2 & * & * & * & 18 & * & * & * & 34 & * & * & * & 10 & * & * & * & 14 & * & * & * & 46 & * & * \\
   * & * & 7 & * & * & * & 23 & * & * & * & 39 & * & * & * & 31 & * & * & * & 3 & * & * & * & 35 & * \\
\end{array}
\right]_{4\times 24}.
\end{IEEEeqnarray*}

It is straightforward to verify that both the query arrays in \eqref{Example2:query:array} and \eqref{Example:APIR:query:array} satisfy the conditions C0--C3.
\end{Example}

After completing the construction of the query array satisfying the conditions C0--C3, we make use of the query array and the \emph{feasible PIR coding framework} to design the adaptive private queries. 

From C0 and C1, we know that $\mathcal{R}_j^h\subseteq[0:P)$ and $\widetilde{\mathcal{R}}_{j}^{h}\subseteq[0:\lambda)$ with $|\mathcal{R}_j^h|=|\widetilde{\mathcal{R}}_{j}^{h}|=\lambda-h$ for any given $h\in[0:\lambda)$ and $j\in[0:\Gamma^h)$.
Therefore, it is reasonable to design private queries to retrieve the partial file $\mathbf{W}^{(\theta)}(\mathcal{R}_j^h,:)$ from the secure storage system in Section \ref{section:storage:system}, for any $h\in[0:\lambda)$ and $j\in[0:\Gamma^h)$. We observe from \eqref{file:symbols}--\eqref{storage:data} and \eqref{Franmework:file:symbols}--\eqref{PIRFramework:storage:evaluation} that the partial file $\mathbf{W}^{(m)}(\mathcal{R}_j^h,:)$ is stored at the distributed system using the same encoding method as the partial file $\widetilde{\mathbf{W}}^{(m)}(\widetilde{\mathcal{R}}_j^h,:)$ for all $m\in[0:M)$. In particular, the \emph{feasible PIR coding framework} is designed regardless of the values of the data files $\widetilde{\mathbf{W}}^{(0)}(\widetilde{\mathcal{R}}_j^h,:),\widetilde{\mathbf{W}}^{(1)}(\widetilde{\mathcal{R}}_j^h,:),\ldots,\widetilde{\mathbf{W}}^{(M-1)}(\widetilde{\mathcal{R}}_j^h,:)$. Thus, following a method similar to that of \eqref{Framework:query:polynomial}-\eqref{Framework:query:design} in the \emph{feasible PIR coding framework}, the encoding functions $\{v_{i,k,0}^{(\widetilde{\mathcal{R}}_j^h)}(x),\ldots,v_{i,k,T}^{(\widetilde{\mathcal{R}}_j^h)}(x)\}_{i\in{\widetilde{\mathcal{R}}}_j^h,k\in[0:K)}$ can be employed to design private queries for retrieving the desired partial file $\mathbf{W}^{(\theta)}(\mathcal{R}_j^h,:)$ for any given $h\in[0:\lambda)$ and $j\in[0:\Gamma^h)$. 

More specifically, for any $m\in[0:M),i\in\mathcal{R}_j^h$, and $k\in[0:K)$, the user locally generates $T$ random noises $z_{(i)_{\lambda},k,0}^{(m,\widetilde{\mathcal{R}}_j^h)},\ldots,z_{(i)_{\lambda},k,T-1}^{(m,\widetilde{\mathcal{R}}_j^h)}$ and then designs the private query function $q_{i,k}^{(m,\mathcal{R}_j^h)}(x)$ as
\begin{IEEEeqnarray}{c}\label{AdaptivePIR:query:function}
q_{i,k}^{(m,\mathcal{R}_j^h)}(x)=\sum\limits_{t=0}^{T-1}{z}_{(i)_{\lambda},k,t}^{(m,\widetilde{\mathcal{R}}_j^h)}\cdot v_{(i)_{\lambda},k,t}^{(\widetilde{\mathcal{R}}_j^h)}(x)
+\left\{
\begin{array}{@{}l@{\;\;}l}
v_{(i)_{\lambda},k,T}^{(\widetilde{\mathcal{R}}_j^h)}(x),
&\mathrm{if}\,\, m=\theta \\
0, & \mathrm{otherwise}
\end{array} \right..
\end{IEEEeqnarray}
The private query $\mathcal{Q}_n^{(\theta,\mathcal{R}_j^h)}$ sent to the server $n$ for retrieving the partial file $\mathbf{W}^{(\theta)}(\mathcal{R}_j^h,:)$ is given by evaluating these query functions at $x=\alpha_n$, i.e.,
\begin{IEEEeqnarray}{c}
\mathcal{Q}_n^{(\theta,\mathcal{R}_j^h)}=\left\{q_{i,k}^{(m,\mathcal{R}_j^h)}(\alpha_n):m\in[0:M),i\in\mathcal{R}_j^h,k\in[0:K)\right\},\quad\forall\,n\in[0:N).
\end{IEEEeqnarray}

To provide the adaptive guarantee, the user will combine the query $\mathcal{Q}_{n}^{(\theta,\mathcal{R}_j^h)}$ across all $h\in[0:\lambda)$ and $j\in[0:\Gamma^h)$, and then sends a total of $F=P=\sum_{h=0}^{\lambda-1}\Gamma^h$ queries to the server $n$, given by
\begin{IEEEeqnarray}{c}\label{APIR:query}
    \mathcal{Q}_{n}^{(\theta)}=\left\{\mathcal{Q}_{n}^{(\theta,\mathcal{R}_j^h)}:h\in[0:\lambda),
  j\in[0:\Gamma^h) \right\},\quad\forall\,n\in[0:N).
\end{IEEEeqnarray}

\subsubsection*{Answer Phase} Due to $\mathcal{R}_j^{h}\subseteq[0:P)$ for any $h\in[0:\lambda)$ and $j\in[0:\Gamma^h)$, the encoded data stored at the server $n$ contains $\{f_{i}^{(m)}(\alpha_n):m\in[0:M),i\in\mathcal{R}_j^{h}\}$ by \eqref{storage:data}. Therefore, upon receiving the query $\mathcal{Q}^{(\theta,\mathcal{R}_j^h)}_{n}$, the server $n$ can generate the response $\mathcal{A}_n^{(\theta,\mathcal{R}_j^h)}$ consisting of $K$ sub-responses in a manner similar to \eqref{Framework:answers:1}-\eqref{Framework:answers:2}, i.e.,
\begin{IEEEeqnarray}{c}\label{APIR:answers}
\mathcal{A}_n^{(\theta,\mathcal{R}_j^h)}=\left\{A_{n,k}^{(\theta,\mathcal{R}_j^h)}:k\in[0:K)\right\}, \quad\forall\, h\in[0:\lambda),  j\in[0:\Gamma^h),
\end{IEEEeqnarray}
where the sub-response $A_{n,k}^{(\theta,\mathcal{R}_j^h)}$ is a linear combination of the encoded data $\{f_{i}^{(m)}(\alpha_n)\}_{m\in[0:M),i\in\mathcal{R}_j^{h}}$ using the received query data $\{q_{i,k}^{(m,\mathcal{R}_j^h)}(\alpha_n)\}_{m\in[0:M),i\in\mathcal{R}_j^h}$ as coefficients, given by
\begin{IEEEeqnarray}{c}\label{answer:data}
A_{n,k}^{(\theta,\mathcal{R}_j^h)}=\sum\limits_{m\in[0:M)}\sum\limits_{i\in\mathcal{R}_j^h}q_{i,k}^{(m,\mathcal{R}_j^h)}(\alpha_n)\cdot f_{i}^{(m)}(\alpha_n),\quad\forall\,k\in[0:K).
\end{IEEEeqnarray}

For the $F=P$ queries in $\mathcal{Q}^{(\theta)}_{n}$ \eqref{APIR:query}, the server $n$ can generate the following $F=P$ responses:
\begin{IEEEeqnarray}{c}\label{answer:APIR}
    \mathcal{A}_{n}^{(\theta)}=\left\{\mathcal{A}_n^{(\theta,\mathcal{R}_j^h)}:h\in[0:\lambda), j\in[0:\Gamma^h) \right\},\quad\forall\,n\in[0:N).
\end{IEEEeqnarray}
In particular, the $P$ responses are calculated and sent to the user in the following order until it can decode the desired file $\mathbf{W}^{(\theta)}$.
\begin{IEEEeqnarray}{c}\label{answer:order}
\left(\mathcal{A}_n^{(\theta,\mathcal{R}_0^0)},\ldots,\mathcal{A}_n^{(\theta,\mathcal{R}_{\Gamma^0-1}^0)},\mathcal{A}_n^{(\theta,\mathcal{R}_0^1)},\ldots,\mathcal{A}_n^{(\theta,\mathcal{R}_{\Gamma^1-1}^1)},\ldots,\mathcal{A}_n^{(\theta,\mathcal{R}_0^{\lambda-1})},\ldots,\mathcal{A}_n^{(\theta,\mathcal{R}_{\Gamma^{\lambda\!-\!1}\!-\!1}^{\lambda-1})}\right).
\end{IEEEeqnarray}
Due to the sequential processing nature of the servers, once the user receives the response $\mathcal{A}_n^{(\theta,\mathcal{R}_j^h)}$ from the server $n$, then it must have received the previous responses $\{\mathcal{A}_n^{(\theta,\mathcal{R}_{j'}^{h'})}:h'\in[0:h),j'\in[0:\Gamma^{h'})\}\bigcup\{\mathcal{A}_n^{(\theta,\mathcal{R}_{j'}^h)}:j'\in[0:j)\}$, for any $h\in[0:\lambda)$ and $j\in[0:\Gamma^h)$.

\subsubsection*{Decoding Phase}
Let $S_{\max}=\lambda-1$, i.e., the proposed adaptive PIR scheme can tolerate the presence of $S$ stragglers simultaneously for all $0\leq S\leq\lambda-1$.
In the presence of $S$ stragglers for all $S\in[0:\lambda)$, we set
\begin{IEEEeqnarray}{c}\notag
F_{S}=\sum\limits_{h=0}^{S}\Gamma^h, 
\end{IEEEeqnarray}
meaning that the user needs to wait for the first $\sum_{h=0}^{S}\Gamma^h$ responses from each of the fastest $N-S$ servers in order to decode the desired file. 
More precisely, in the presence of $S$ stragglers for all $S\in[0:\lambda)$, by \eqref{answer:order} the responses received by the user from the fastest $N-S$ servers are given by
\begin{IEEEeqnarray}{c} \label{APIR:response:sub}
\left\{\mathcal{A}_{n}^{(\theta,\mathcal{R}_{j}^{h})}:n\in[0:N)\backslash\mathcal{S},h\in[0:S],j\in[0:\Gamma^h)\right\},\quad\forall\, \mathcal{S}\subseteq [0:N), |\mathcal{S}|=S,
\end{IEEEeqnarray}
where $\mathcal{S}$ represents the indices of straggler servers. 


To demonstrate that the proposed adaptive PIR scheme can tolerate the presence of $S$ stragglers simultaneously for all $0\leq S\leq\lambda-1$, it suffices to show that the user can decode the partial files $\{{\mathbf{W}}^{(\theta)}(\mathcal{R}^0_j,:)\}_{j\in[0:\Gamma^0)}$ from the received responses in \eqref{APIR:response:sub} for any given $S\in[0:\lambda)$. This is because the desired file $\mathbf{W}^{(\theta)}$ can be recovered from $\{{\mathbf{W}}^{(\theta)}(\mathcal{R}^0_j,:)\}_{j\in[0:\Gamma^0)}$ by C2 and \eqref{file:symbols}. 

Recall from C1 that $\widetilde{\mathcal{R}}_{j}^{h}\subseteq[0:\lambda)$ with $|\widetilde{\mathcal{R}}_{j}^{h}|=\lambda-h$ for any $h\in[0:\lambda)$ and $j\in[0:\Gamma^h)$. Therefore, according to Definition \ref{Definition:feasible:PIR:framwork}, the \emph{feasible PIR coding framework} can be used to privately download one out of the $M$ partial files $\{\widetilde{\mathbf{W}}^{(m)}(\widetilde{\mathcal{R}}_{j}^{h},:)\}_{m\in[0:M)}$ from a secure distributed system, by employing the encoding parameters $\alpha_0,\alpha_1,\ldots,\alpha_{N-1}$ along with the encoding functions $\{b_{i,0}(x),\ldots,b_{i,K+X-1}(x)\}_{i\in\widetilde{\mathcal{R}}_{j}^{h}}$ and $\{v_{i,k,0}^{(\widetilde{\mathcal{R}}_{j}^{h})}(x),\ldots,v_{i,k,T}^{(\widetilde{\mathcal{R}}_{j}^{h})}(x)\}_{i\in\widetilde{\mathcal{R}}_{j}^{h},k\in[0:K)}$.
For any $h\in[0:\lambda)$ and $j\in[0:\Gamma^h)$, we design the adaptive PIR scheme to privately retrieve one out of the $M$ partial files $\{\mathbf{W}^{(m)}(\mathcal{R}_j^h,:)\}_{m\in[0:M)}$ from a secure distributed system, using the same encoding parameters and encoding functions as those in the \emph{feasible PIR coding framework}. 
In terms of retrieving the partial file, the only difference between the adaptive PIR scheme and the PIR coding framework lies in the original data files used. 
Since the parameters and encoding functions of the \emph{feasible PIR coding framework} are designed independently of the values of the original files, the adaptive PIR scheme satisfies the decoding properties similar to those of the constraints $\widetilde{\mathrm{F}}$2--$\widetilde{\mathrm{F}}$3, i.e., for any $h\in[0:\lambda)$ and $j\in[0:\Gamma^h)$,
\begin{enumerate}
\item[F2:] The user can recover the partial file ${\mathbf{W}}^{(\theta)}(\mathcal{R}_j^h,:)$ from any $N-h$ out of the $N$ responses $\mathcal{A}_{[0:N)}^{(\theta,\mathcal{R}_j^h)}$. 
\item[F3:] The partial file ${\mathbf{W}}^{(\theta)}(\mathcal{R}^h_j,:)$ can be decoded from any $N-(h+d)$ out of the $N$ responses $\mathcal{A}_{[0:N)}^{(\theta,\mathcal{R}_j^h)}$ after obtaining any $d$ rows of ${\mathbf{W}}^{(\theta)}(\mathcal{R}_j^h,:)$ for any $d\in[1:\lambda-h)$. 
\end{enumerate}


Next, we use an iterative approach to show that the partial files $\{{\mathbf{W}}^{(\theta)}(\mathcal{R}^0_j,:)\}_{j\in[0:\Gamma^0)}$ can be decoded from the received responses in \eqref{APIR:response:sub}. 

\begin{itemize}
\item In the initial step, the partial file ${\mathbf{W}}^{(\theta)}(\mathcal{R}_j^S,:)$ can be decoded from the $N-S$ received responses $\mathcal{A}_{[0:N)\backslash\mathcal{S}}^{(\theta,\mathcal{R}_{j}^{S})}$ by \eqref{APIR:response:sub} and F2 for all $j\in[0:\Gamma^S)$.
\item In the iterative step, we prove that the user can decode the partial files $\{\mathbf{W}^{(\theta)}(\mathcal{R}_j^h,:)\}_{j\in[0:\Gamma^h)}$ after completing the decoding of $\{{\mathbf{W}}^{(\theta)}(\mathcal{R}_j^r,:):r\in[h+1:S],j\in[0:\Gamma^r)\}$ for any $h\in[0:S)$.

We know from C3 that  for any $h\!\in\![0\!:\!S)$ and $j\!\in\![0\!:\!\Gamma^h)$, there is a subset $\mathcal{D}^{h}_j\!=\!\{a_0,a_1,\ldots,a_{S-h-1}\}\!\subseteq\!\mathcal{R}^h_j$ of size $S-h$ such that $a_{r-h-1}\!\in\!\bigcup_{j\in\Gamma^{r}}\mathcal{R}_j^{r}$ for all $r\!\in\![h+1\!:\!S]$, i.e., the row vector ${\mathbf{W}}^{(\theta)}(\{a_{r-h-1}\},:)$ of the partial file ${\mathbf{W}}^{(\theta)}(\mathcal{R}_j^h,:)$ is known after completing the decoding of $\{{\mathbf{W}}^{(\theta)}(\mathcal{R}_j^r,:)\}_{j\in\Gamma^r}$ for any $r\in[h+1:S]$. Thus, 
the $S-h$ row vectors ${\mathbf{W}}^{(\theta)}(\mathcal{D}_j^h,:)$  of ${\mathbf{W}}^{(\theta)}(\mathcal{R}_j^h,:)$ can be obtained from $\{{\mathbf{W}}^{(\theta)}(\mathcal{R}_j^r,:)\}_{r\in[h+1:S],j\in[0:\Gamma^r)}$.
Then, by applying $d=S-h$ to F3, the partial file ${\mathbf{W}}^{(\theta)}(\mathcal{R}^h_j,:)$ is decoded from the $N-S$ received responses $\mathcal{A}_{[0:N)\backslash\mathcal{S}}^{(\theta,\mathcal{R}_j^h)}$ along with 
the $S-h$ row vectors ${\mathbf{W}}^{(\theta)}(\mathcal{D}_j^h,:)$  of ${\mathbf{W}}^{(\theta)}(\mathcal{R}_j^h,:)$ for all $j\in[0:\Gamma^h)$. This completes the proof of the iterative step.
\end{itemize}
By iteratively setting $h=S-1,S-2,\ldots,0$, the user will recover the partial files $\{{\mathbf{W}}^{(\theta)}(\mathcal{R}^0_j,:)\}_{j\in[0:\Gamma^0)}$ and then obtains the desired file $\mathbf{W}^{(\theta)}$.
Consequently, in the presence of $S$ stragglers for all $S\in[0:\lambda)$, the user can correctly decode the desired file $\mathbf{W}^{(\theta)}$ from the responses in \eqref{APIR:response:sub}.

From \eqref{APIR:answers} and \eqref{answer:data}, each response consists of $K$ sub-responses, with each sub-response being an encoded symbol defined over the finite field $\mathbb{F}_q$. Therefore, in the presence of $S$ stragglers for all $S\in[0:\lambda)$, the communication cost of downloading the responses in \eqref{APIR:response:sub} from the fastest $N-S$ servers is given by
\begin{IEEEeqnarray}{c}\notag
 D_S=(N-S)K\cdot\sum\limits_{h=0}^{S}\Gamma^{h}=\frac{(N-S)KP}{N-S-(K+X+T-1)},
\end{IEEEeqnarray}
where the last equation is due to \eqref{define:lambda} and \eqref{number:column}.
Since each file contains $L=PK$ symbols by \eqref{file:symbols}, the retrieval rate \eqref{def:PPC rate} is given by
\begin{IEEEeqnarray}{c}\notag
  R_S=\frac{PK}{D_S}
  =1-\frac{K+X+T-1}{N-S}.
\end{IEEEeqnarray}

It is straightforward to observe that the secure storage system, private queries, server responses, and decoding processes of the proposed adaptive PIR scheme are all designed using the encoding parameters $\{\alpha_n\}_{n\in[0:N)}$ along with the encoding functions $\{b_{i,0}(x),\ldots,b_{i,K+X-1}(x)\}_{i\in[0:\lambda)}$ and $\{v_{i,k,0}^{(\widetilde{\mathcal{R}}_{j}^{h})}(x),\ldots,v_{i,k,T}^{(\widetilde{\mathcal{R}}_{j}^{h})}(x)\}_{i\in\widetilde{\mathcal{R}}_{j}^{h},k\in[0:K),h\in[0:\lambda),j\in[0:\Gamma^h)}$
from the \emph{feasible PIR coding framework}. 
Therefore, the finite field for the \emph{feasible PIR coding framework} is enough to ensure the achievability of the adaptive PIR scheme.

We summarize the general results of the proposed adaptive PIR scheme in the following theorem.
\begin{Theorem}
For any given system parameters $N,K,M,X,T$ with $N>K+X+T-1$, 
any \emph{feasible PIR coding framework} over a finite field $\mathbb{F}_q$ of size $q$ can be utilized to construct an adaptive PIR scheme that can tolerate the presence of $S$ stragglers simultaneously for all $0\leq S\leq N-(K+X+T)$. In particular, the constructed adaptive PIR scheme achieves a retrieval rate of $1-\frac{K+X+T-1}{N-S}$ simultaneously for all numbers of stragglers $0\leq S\leq N-(K+X+T)$ over the same finite filed.
\end{Theorem}

By Lemmas \ref{theorem:tield size} and \ref{lemma:feasible:PIR}, the following results are immediately obtained.

\begin{Theorem}
For any given system parameters $N,K,M,X,T$ with $N>K+X+T-1$, there exists an adaptive PIR scheme that achieves a retrieval rate of $1-\frac{K+X+T-1}{N-S}$ in the presence of $S$ stragglers simultaneously for all $0\leq S\leq N-(K+X+T)$. 
Additionally, the adaptive PIR scheme can operate over the finite filed $\mathbb{F}_q$ for any prime power $q$ of size $q\geq N+\max\{K, N-(K+X+T-1)\}$ and achieves the optimal finite field size $q= N+\max\{K, N-(K+X+T-1)\}$ under the constraints P0--P3.
\end{Theorem}

\begin{Remark}
Although the \emph{feasible PIR coding framework} implemented in Section \ref{subsection:implementation:PIR} achieves the optimal finite field size, it is constrained by the conditions P0--P3. Therefore, it is preferable to relax these constraints and design an alternative framework with a smaller finite field size. This will allow the associated adaptive PIR scheme to operate over a smaller finite field.
Furthermore, different storage codes exhibit varying performance metrics, such as the repair bandwidth when repairing failed servers \cite{dimakis2010network}. Designing \emph{feasible PIR coding framework}s based on diverse storage coding systems can be advantageous \cite{kumar2019achieving,zhu2019new,li2020towards}, as it allows adaptive PIR schemes to function effectively across a wide range of storage systems.
\end{Remark}

\begin{Remark}
We might additionally take into account the existence of certain Byzantine servers with a size of up to $B$, who intentionally return arbitrarily incorrect responses to the user. In the \emph{feasible PIR coding framework} implemented using Lagrange codes, all the server responses for retrieving any partial file constitute Reed-Solomon (RS) codewords \cite{zhu2022symmetric,zhu2022multi}.
By leveraging the decoding properties of RS codes, the adaptive PIR scheme built on the \emph{feasible PIR coding framework} implemented using Lagrange codes can provide robustness against any $B$ Byzantine servers by waiting for responses from an additional $2B$ servers. 
In particular, the additional waiting for responses from $2B$ servers does not change the required finite field size.
Therefore, in the scenario of up to $B$ Byzantine servers, the adaptive PIR scheme achieves a retrieval rate of $1-\frac{K+X+T+2B-1}{N-S}$ simultaneously for all numbers of stragglers $0\leq S\leq N-(K+X+T+2B)$ over any finite filed $\mathbb{F}_q$ of size $q\geq N+\max\{K, N-(K+X+T-1)\}$.
\end{Remark}

\subsection{Secrecy and Privacy}\label{APIR:secrecy:Privacy}
In this subsection, we demonstrate the secrecy of the distributed storage system and the privacy of the proposed adaptive PIR scheme. 


\begin{Lemma}[Generalized Secret Sharing \cite{shamir1979share,jackson1996combinatorial,bitar2017staircase}]\label{secrecy:lemma}
Given any positive integers $n,k,t, m$ with $n\geq k+t$, let $\{s_0^{\ell},\ldots,s_{k-1}^{\ell}\}_{\ell\in[0:m)}$ denote $km$ random secrets on $\mathbb{F}_q$, and $\{z_0^{\ell},\ldots,z_{t-1}^{\ell}\}_{\ell\in[0:m)}$ denote $tm$ random noises chosen independently and uniformly from $\mathbb{F}_q$. Additionally, these noises are selected independently of the $km$ secrets. Let $\alpha_0,\ldots,\alpha_{n-1}$ be $n$ pairwise distinct elements from $\mathbb{F}_q$.
Define a function as
\begin{IEEEeqnarray}{c}
f^{\ell}(x)=s^{\ell}_{0}\cdot u_0^{\ell}(x)+\ldots+s_{k-1}^{\ell}\cdot u_{k-1}^{\ell}(x)+{z}_{0}^{\ell}\cdot v_0^{\ell}(x)+\ldots+{z}^{\ell}_{t-1}\cdot v^{\ell}_{t-1}(x),\quad\forall\,\ell\in[0:m),\notag
\end{IEEEeqnarray}
where $u_0^{\ell}(x),\ldots,u^{\ell}_{k-1}(x),v^{\ell}_0(x),\ldots,v^{\ell}_{t-1}(x)\in\mathbb{F}_q[x]$ are arbitrary deterministic functions of $x$.
If the matrix
\begin{IEEEeqnarray}{c}\notag
\mathbf{V}^{\ell}=
\left[
  \begin{array}{cccc}
    v^{\ell}_0(\alpha_{i_0}) & v^{\ell}_1(\alpha_{i_0}) & \cdots & v^{\ell}_{t-1}(\alpha_{i_0}) \\
    v^{\ell}_0(\alpha_{i_1}) & v^{\ell}_1(\alpha_{i_1}) & \cdots & v^{\ell}_{t-1}(\alpha_{i_1}) \\
    \vdots & \vdots & \vdots & \vdots \\
    v^{\ell}_0(\alpha_{i_{t-1}}) & v^{\ell}_1(\alpha_{i_{t-1}}) & \cdots & v^{\ell}_{t-1}(\alpha_{i_{t-1}}) \\
  \end{array}
\right]_{t\times t}
\end{IEEEeqnarray}
is non-singular over $\mathbb{F}_q$ for all $\ell\in[0:m)$, then the $tm$ values $\{f^{\ell}(\alpha_{i_0}),\ldots,f^{\ell}(\alpha_{i_{t-1}})\}_{\ell\in[0:m)}$  reveal nothing about the $km$ secrets $\{{s}^{\ell}_{0},\ldots,s^{\ell}_{k-1}\}_{\ell\in[0:m)}$, for any subset $\{i_0,\ldots,i_{t-1}\}\subseteq[0:n)$ of size $t$, i.e.,
\begin{IEEEeqnarray}{c}\notag
I(\{{s}^{\ell}_{0},\ldots,{s}^{\ell}_{k-1}\}_{\ell\in[0:m)};\{f^{\ell}(\alpha_{i_0}),\ldots,f^{\ell}(\alpha_{i_{t-1}})\}_{\ell\in[0:m)})=0.
\end{IEEEeqnarray}
\end{Lemma}


The proofs for the secrecy and privacy requirements are similar. In the construction of the adaptive PIR scheme, the storage system \eqref{storage:data} and query phase \eqref{APIR:query} are designed using the encoding parameters $\{\alpha_n\}_{n\in[0:N)}$ along with the encoding functions $\{b_{(i)_{\lambda},k}(x)\}_{i\in[0:P),k\in[0:K+X)}$ and $\{v_{(i)_{\lambda},k,t}^{(\widetilde{\mathcal{R}}_j^h)}(x)\}_{h\in[0:\lambda),j\in[0:\Gamma^h),i\in\mathcal{R}_j^h,k\in[0:K)}$ from the feasible adaptive PIR coding framework.
Since the \emph{feasible PIR coding framework} satisfies the condition $\widetilde{\mathrm{F}}$1, both the storage and query designs of the adaptive PIR scheme fulfill the prerequisites of Lemma \ref{secrecy:lemma}, such that a specific number of colluding servers learn nothing about the original confidential data. 
Next, we provide formal proofs.   

\subsubsection*{Secrecy}
For any given $i\in[0:P)$ and any subset $\mathcal{X}=\{n_0,n_1,\ldots,n_{X-1}\}\subseteq[0:N)$ of size $X$, the following matrix $\mathbf{B}_{i,\mathcal{X}}$ of dimensions $X\times X$ is non-singular over $\mathbb{F}_q$ by $\widetilde{\mathrm{F}}$1.
\begin{IEEEeqnarray}{c}\notag
\mathbf{B}_{i,\mathcal{X}}=
\left[
\begin{array}{@{}cccc@{}}
  b_{(i)_{\lambda},K}(\alpha_{n_0}) & b_{(i)_{\lambda},K+1}(\alpha_{n_0}) & \cdots & b_{(i)_{\lambda},K+X-1}(\alpha_{n_0}) \\
  b_{(i)_{\lambda},K}(\alpha_{n_1}) & b_{(i)_{\lambda},K+1}(\alpha_{n_1}) & \cdots & b_{(i)_{\lambda},K+X-1}(\alpha_{n_1}) \\
  \vdots & \vdots & \vdots & \vdots \\
  b_{(i)_{\lambda},K}(\alpha_{n_{X-1}}) & b_{(i)_{\lambda},K+1}(\alpha_{n_{X-1}}) & \cdots & b_{(i)_{\lambda},K+X-1}(\alpha_{n_{X-1}}) \\
\end{array}
\right].
\end{IEEEeqnarray}
Then, by \eqref{Adaptive:PIR:storage:function} and Lemma \ref{secrecy:lemma}, 
\begin{IEEEeqnarray}{c}
I(\{w_{i,0}^{(m)},\ldots,w_{i,K-1}^{(m)}\}_{m\in[0:M),i\in[0:P)};\{f_i^{(m)}(\alpha_{n_0}),\ldots,f_i^{(m)}(\alpha_{n_{X-1}})\}_{m\in[0:M),i\in[0:P)})=0. \notag
\end{IEEEeqnarray}
Furthermore, due to \eqref{file:symbols} and \eqref{storage:data}, we have
\begin{IEEEeqnarray}{rCl}
&&I(\mathbf{W}^{(0)},\ldots,\mathbf{W}^{(M-1)};\mathcal{Y}_{n_0},\ldots,\mathcal{Y}_{n_{X-1}}) \notag\\
&=&I(\{w_{i,0}^{(m)},\ldots,w_{i,K-1}^{(m)}\}_{m\in[0:M),i\in[0:P)};\{f_i^{(m)}(\alpha_{n_0}),\ldots,f_i^{(m)}(\alpha_{n_{X-1}})\}_{m\in[0:M),i\in[0:P)})\notag\\
&=&0. \notag
\end{IEEEeqnarray}
This shows that the storage system satisfies the secrecy requirement in \eqref{security}. 

\subsubsection*{Privacy}
For any given $h\in[0:\lambda),j\in[0:\Gamma^h),i\in\mathcal{R}_j^h,k\in[0:K)$, and any subset $\mathcal{T}=\{n_0,n_1,\ldots,n_{T-1}\}\subseteq[0:N)$ of size $T$, the following matrix $\mathbf{V}_{(i)_{\lambda},k,\mathcal{T}}^{(\widetilde{\mathcal{R}}^{h}_j)}$ of dimensions $T\times T$ is non-singular over $\mathbb{F}_q$ by $\widetilde{\mathrm{F}}$1.
\begin{IEEEeqnarray}{c}\notag
\mathbf{V}_{(i)_{\lambda},k,\mathcal{T}}^{(\widetilde{\mathcal{R}}^{h}_j)}=
\left[
\begin{array}{cccc}
v_{(i)_{\lambda},k,0}^{(\widetilde{\mathcal{R}}^{h}_j)}(\alpha_{n_0}) & v_{(i)_{\lambda},k,1}^{(\widetilde{\mathcal{R}}^{h}_j)}(\alpha_{n_0}) & \cdots & v_{(i)_{\lambda},k,T-1}^{(\widetilde{\mathcal{R}}^{h}_j)}(\alpha_{n_0}) \\
v_{(i)_{\lambda},k,0}^{(\widetilde{\mathcal{R}}^{h}_j)}(\alpha_{n_1}) & v_{(i)_{\lambda},k,1}^{(\widetilde{\mathcal{R}}^{h}_j)}(\alpha_{n_1}) & \cdots & v_{(i)_{\lambda},k,T-1}^{(\widetilde{\mathcal{R}}^{h}_j)}(\alpha_{n_1}) \\
\vdots & \vdots & \vdots & \vdots \\
v_{(i)_{\lambda},k,0}^{(\widetilde{\mathcal{R}}^{h}_j)}(\alpha_{n_{T-1}}) & v_{(i)_{\lambda},k,1}^{(\widetilde{\mathcal{R}}^{h}_j)}(\alpha_{n_{T-1}}) & \cdots & v_{(i)_{\lambda},k,T-1}^{(\widetilde{\mathcal{R}}^{h}_j)}(\alpha_{n_{T-1}}) \\
\end{array}
\right].
\end{IEEEeqnarray}
Then, by \eqref{AdaptivePIR:query:function}--\eqref{APIR:query} and Lemma \ref{secrecy:lemma},
\begin{IEEEeqnarray}{rCl}
&&I(\theta;\mathcal{Q}_{n_0}^{(\theta)},\ldots.\mathcal{Q}_{n_{T-1}}^{(\theta)}) \notag \\
&=&I\big(\theta;\big\{q^{(m,\mathcal{R}_{j}^{h})}_{i,k}(\alpha_{n_0}),\ldots,q^{(m,\mathcal{R}_{j}^{h})}_{i,k}(\alpha_{n_{T-1}})\big\}_{m\in[0:M),h\in[0:\lambda),j\in[0:\Gamma^h),i\in\mathcal{R}_j^h,k\in[0:K)}\big) \quad \notag \\
&=&0.\label{privacy:pppp}
\end{IEEEeqnarray}

In the proposed adaptive PIR scheme, all possible information available to the colluding servers $\mathcal{T}$ contains the stored data $\mathcal{Y}_{\mathcal{T}}$ \eqref{storage:data}, the received query $\mathcal{Q}_{\mathcal{T}}^{(\theta)}$  \eqref{APIR:query}, and the generated responses $\mathcal{A}_{\mathcal{T}}^{(\theta)}$ \eqref{answer:APIR}. Particularly, we have
\begin{IEEEeqnarray*}{rCl}
0&\leq&I(\theta;\mathcal{Q}_{\mathcal{T}}^{(\theta)},\mathcal{A}_{\mathcal{T}}^{(\theta)},\mathcal{Y}_{\mathcal{T}})\\
&=&I(\theta;\mathcal{Q}_{\mathcal{T}}^{(\theta)})+I(\theta;\mathcal{Y}_{\mathcal{T}}|\mathcal{Q}_{\mathcal{T}}^{(\theta)})+
I(\theta;\mathcal{A}_{\mathcal{T}}^{(\theta)}|\mathcal{Q}_{\mathcal{T}}^{(\theta)},\mathcal{Y}_{\mathcal{T}})\\
&\overset{(a)}{=}&I(\theta;\mathcal{Q}_{\mathcal{T}}^{(\theta)})\\
&\overset{(b)}{=}&0,
\end{IEEEeqnarray*}
where $(a)$ follows because the encoded data $\mathcal{Y}_{\mathcal{T}}$ is generated independently of the query $\mathcal{Q}_{\mathcal{T}}^{(\theta)}$ and the index $\theta$ by \eqref{Adaptive:PIR:storage:function}-\eqref{storage:data} and the response $\mathcal{A}_{\mathcal{T}}^{(\theta)}$ is a deterministic function of $\mathcal{Q}_{\mathcal{T}}^{(\theta)}$ and $\mathcal{Y}_{\mathcal{T}}$ by \eqref{APIR:answers}--\eqref{answer:APIR}, such that $I(\theta;\mathcal{Y}_{\mathcal{T}}|\mathcal{Q}_{\mathcal{T}}^{(\theta)})=0$ and $I(\theta;\mathcal{A}_{\mathcal{T}}^{(\theta)}|\mathcal{Q}_{\mathcal{T}}^{(\theta)},\mathcal{Y}_{\mathcal{T}})=0$; $(b)$ is due to \eqref{privacy:pppp}.

This means that all available information to any $T$ colluding servers reveals nothing about the index $\theta$ and thus the privacy requirement in \eqref{Infor:priva cons} is satisfied.

\section{Concluding Remarks}\label{conclusion}
In this paper, we investigate the $T$-colluding adaptive PIR problem with $X$-secure $K$-coded distributed storage. By introducing the concept of the \emph{feasible PIR coding framework}, we proposed a general coding method for designing adaptive PIR schemes that can be capable of tolerating the presence of a varying number of stragglers.
We showed that any \emph{feasible PIR coding framework} over a finite field $\mathbb{F}_q$ of size $q$ can be utilized to construct an adaptive PIR scheme that achieves a retrieval rate of $1-\frac{K+X+T-1}{N-S}$, conjectured to be asymptotically optimal by Jia and Jafar, simultaneously for all numbers of stragglers $0\leq S\leq N-(K+X+T)$ over the same finite field. This emphasizes the importance of constructing a feasible PIR framework with a small finite field size to design efficient adaptive PIR schemes.
We provided an implementation of the \emph{feasible PIR coding framework} and demonstrated that the adaptive PIR scheme can operate over any finite field $\mathbb{F}_q$ with size $q\geq N+\max\{K, N-(K+X+T-1)\}$.

An interesting direction for future research is to determine the optimal finite field size for the \emph{feasible PIR coding framework}. 
Since the storage system must satisfy the MDS property, it is natural to trace the problem of determining the optimal finite field size back to the MDS conjecture \cite[Conjecture 11.16]{roth2006introduction}, which postulates that the optimal finite field size required for linear MDS codes is close to the code length $N$. Compared to the finite field size suggested by the MDS conjecture, the \emph{feasible PIR coding framework} we implement requires an increase in the finite field size by $\max\{K, N-(K+X+T-1)\}$.
This increase can be attributed to two main reasons. The first one is that we use Lagrange interpolating polynomials to construct a secure $(N,K+X)$-coded distributed storage system, which requires $K+X$ distinct interpolation points and $N$ distinct evaluation points. In particular, to guarantee storage security, the $K$ interpolation points corresponding to the confidential data must be distinct from the $N$ evaluation points, resulting in the need for a finite field of size $N+K$. The second one is that, to enhance the retrieval efficiency, 
each file is encoded using $N-(K+X+T-1)$ distinct MDS codes with different coding parameters, which leads to a requirement for a finite field of size $N+(N-(K+X+T-1))$. The key to reducing the finite field size toward the lower bound suggested by the MDS conjecture lies in designing a novel type of private query tailored for secure storage systems in which all files are encoded using $(N,K+X)$ MDS codes over a finite field of size approximately $N$.

Additionally, to provide the adaptive guarantee, our scheme requires the file size to be $K(N-K-X-T+1)\cdot\mathrm{lcm}(1,2,\ldots,N-K-X-T+1)$. An immediate challenge for future work is to reduce the file size to a linear level for the adaptive PIR problem. 
Furthermore, the concept of designing adaptive PIR schemes using the \emph{feasible PIR coding framework} shows significant potential for addressing other variants of the adaptive PIR problem, including those involving symmetric privacy, graph-based storage, and federated submodel learning. These avenues deserve further exploration. 
In particular, determining the (asymptotically) optimal retrieval rate of the $T$-colluding MDS-coded PIR problem, with or without $X$-secure storage, remains one of the most significant open problems in the PIR literature. Given its theoretical importance and practical implications, this problem is highly promising and warrants continued investigation.

\begin{appendix}\label{proof:achevable}

In this appendix, we provide the proofs of Lemmas \ref{theorem:tield size} and \ref{lemma:query:array}, as well as an illustrative example of the \emph{feasible PIR coding framework} based on Lagrange codes.

\subsection{Proof of Lemma \ref{theorem:tield size}}\label{Proof:finite:field:size}

We first demonstrate how to select a group of elements $\{\beta_{i,k},\alpha_n:i\!\in\![0:\lambda),k\!\in\![0:K+X),n\!\in\![0:N)\}$ that satisfies the constraints P0--P3 from a finite field $\mathbb{F}_q$ of size $q\geq N+\max\{K, \lambda\}$. 



Let $(\alpha_{0},\alpha_{1},\ldots,\alpha_{N-1})$ be $N$ pairwise distinct elements from $\mathbb{F}_q$, then set
\begin{IEEEeqnarray}{c}\notag
    (\beta_{i,K},\beta_{i,K+1},\ldots,\beta_{i,K+X-1})=(\alpha_{0},\alpha_{1},\ldots,\alpha_{X-1}),\quad\forall\,i\in[0:\lambda).
\end{IEEEeqnarray}
Depending on whether $K$ is greater than $\lambda$ or not, the remaining elements $\{\beta_{i,k}:i\in[0:\lambda),k\in[0:K)\}$ are assigned values through the following two steps.
\begin{itemize}
  \item If $K>\lambda$, let $(\beta_{0,0},\beta_{0,1},\ldots,\beta_{0,K-1})$ be $K$ distinct elements from $\mathbb{F}_q\backslash\{\alpha_{n}:n\in[0:N)\}$ such that 
  \begin{IEEEeqnarray}{c}\notag
  \{\beta_{0,k}:k\in[0:K)\}\cap\{\alpha_{n}:n\in[0:N)\}=\emptyset.
  \end{IEEEeqnarray}
  Then set
  \begin{IEEEeqnarray}{c}\notag
      (\beta_{i,0},\beta_{i,1},\ldots,\beta_{i,K-1})=(\beta_{0,(0+i)_{K}},\beta_{0,(1+i)_{K}},\ldots,\beta_{0,(K-1+i)_{K}}),\quad\forall\,i\in[1:\lambda).
  \end{IEEEeqnarray}
  \item If $K\leq\lambda$, let $(\beta_{0,0},\beta_{1,0},\ldots,\beta_{\lambda-1,0})$ be $\lambda$ distinct elements from $\mathbb{F}_q\backslash\{\alpha_{n}:n\in[0:N)\}$ such that 
  \begin{IEEEeqnarray}{c}\notag
  \{\beta_{i,0}:i\in[0:\lambda)\}\cap\{\alpha_{n}:n\in[0:N)\}=\emptyset.
  \end{IEEEeqnarray}
  Then set
  \begin{IEEEeqnarray}{c}\notag
      (\beta_{0,k},\beta_{1,k},\ldots,\beta_{\lambda-1,k})=(\beta_{(0+k)_{\lambda},0},\beta_{(1+k)_{\lambda},0},\ldots,\beta_{(\lambda-1+k)_{\lambda},0}),\quad\forall\,k\in[1:K).
  \end{IEEEeqnarray}
\end{itemize}

It is easy to prove that the group of selected elements $\{\beta_{i,k},\alpha_n:i\in[0:\lambda),k\in[0:K+X),n\in[0:N)\}$ satisfies the constraints P0--P3 with $|\{\beta_{i,k},\alpha_n:i\in[0:\lambda),k\in[0:K+X),n\in[0:N)\}|=N+\max\{K,\lambda\}$.

Next, to complete the proof of this lemma, it is sufficient to show that the size of the finite field satisfying the constraints P0--P3 must match the lower bound $q\geq N+\max\{K, \lambda\}$.

Since both $\alpha_0,\ldots,\alpha_{N-1}$ and $\beta_{0,0},\ldots,\beta_{0,K-1}$ are pairwise distinct with $\{\alpha_0,\ldots,\alpha_{N-1}\}\cap\{\beta_{0,0},\ldots,\beta_{0,K-1}\}=\emptyset$ by P0 and P2-P3, the group of elements $\{\beta_{i,k},\alpha_n:i\!\in\![0\!:\!\lambda),k\!\in\![0\!:\!K+X),n\!\in\![0\!:\!N)\}$ must satisfy $|\{\beta_{i,k},\alpha_n:i\in[0:\lambda),k\in[0:K+X),n\in[0:N)\}|\geq N+K$. Additionally, since $\beta_{0,0},\beta_{1,0},\ldots,\beta_{\lambda-1,0}$ are pairwise distinct elements with $\{\alpha_0,\ldots,\alpha_{N-1}\}\cap\{\beta_{0,0},\ldots,\beta_{\lambda-1,0}\}=\emptyset$ by P1 and P3, this group of elements must also satisfy $|\{\beta_{i,k},\alpha_n:i\in[0:\lambda),k\in[0:K+X),n\in[0:N)\}|\geq N+\lambda$. Consequently, we must have $|\{\beta_{i,k},\alpha_n:i\in[0:\lambda),k\in[0:K+X),n\in[0:N)\}|\geq N+\max\{K,\lambda\}$ under the constraints P0--P3.

\subsection{Illustrated Example for Feasible PIR Coding Framework Based on Lagrange Codes}\label{example:lagrange:codes}

In this subsection, we present a detailed example to illustrate the feasibility of the proposed PIR coding framework based on Lagrange codes. For consistency, we revisit the PIR problem described in Example \ref{Flexible:example} and implement the corresponding PIR coding framework using Lagrange codes. 

Before that, we provide a lemma that will be useful for proving the feasibility of the PIR coding framework.
\begin{Lemma}[Lemma 2 in \cite{yu2019lagrange}]\label{g-cauchy matrix}
Let $\alpha_0,\ldots,\alpha_{t-1}$ and $\beta_0,\ldots,\beta_{t-1}$ be the elements from $\mathbb{F}_q$ such that $\alpha_i\neq\alpha_j,\beta_i\neq\beta_j$ for any $i,j\in[0:t)$ with $i\neq j$, and $u_0,\ldots,u_{t-1},c_0,\ldots,c_{t-1}$ be $2t$ nonzero elements from $\mathbb{F}_q$. Denote by $v_i(x)$ a polynomial of degree $t-1$, given by
\begin{IEEEeqnarray}{c}\notag
v_i(x)=\prod\limits_{j\in[0:t)\backslash\{i\}}\frac{x-\beta_j}{\beta_i-\beta_j},\quad\forall\,i\in[0:t).
\end{IEEEeqnarray}
Then the following matrix $\mathbf{V}$ is invertible over $\mathbb{F}_q$.
\begin{IEEEeqnarray}{c}\notag
\mathbf{V}=\mathrm{diag}\left(u_0,u_1,\ldots,u_{t-1}\right)\cdot
\left[
  \begin{array}{cccc}
    v_{0}(\alpha_0) & v_{1}(\alpha_0)& \ldots & v_{t-1}(\alpha_0)  \\
    v_{0}(\alpha_1) & v_{1}(\alpha_1)& \ldots & v_{t-1}(\alpha_1)  \\
    \vdots & \vdots & \vdots & \vdots \\
    v_{0}(\alpha_{t-1}) & v_{1}(\alpha_{t-1})& \ldots & v_{t-1}(\alpha_{t-1})  \\
\end{array}
\right]\cdot\mathrm{diag}\left(c_0,c_1,\ldots,c_{t-1}\right).
\end{IEEEeqnarray}
\end{Lemma}

Consider the same PIR problem as in Example \ref{Flexible:example}, where the system parameters satisfy $N=8,K=X=T=2,M=3$, and $\lambda=3$. Similar to \eqref{example:files}, the $3$ files $\widetilde{\mathbf{W}}^{(0)},\widetilde{\mathbf{W}}^{(1)},\widetilde{\mathbf{W}}^{(2)}$ are given by
\begin{IEEEeqnarray}{c}\label{example:feasible:PIR}
\widetilde{\mathbf{W}}^{(0)} =\left[
\begin{array}{@{\;}cc@{\;}}
\widetilde{w}^{(0)}_{0,0} & \widetilde{w}^{(0)}_{0,1}  \\
\widetilde{w}^{(0)}_{1,0} & \widetilde{w}^{(0)}_{1,1}  \\
\widetilde{w}^{(0)}_{2,0} & \widetilde{w}^{(0)}_{2,1}  \\
\end{array}
\right], \quad
\widetilde{\mathbf{W}}^{(1)} =\left[
\begin{array}{@{\;}cc@{\;}}
\widetilde{w}^{(1)}_{0,0} & \widetilde{w}^{(1)}_{0,1}  \\
\widetilde{w}^{(1)}_{1,0} & \widetilde{w}^{(1)}_{1,1}  \\
\widetilde{w}^{(1)}_{2,0} & \widetilde{w}^{(1)}_{2,1}  \\
\end{array}
\right], \quad
\widetilde{\mathbf{W}}^{(2)} =\left[
\begin{array}{@{\;}cc@{\;}}
\widetilde{w}^{(2)}_{0,0} & \widetilde{w}^{(2)}_{0,1}  \\
\widetilde{w}^{(2)}_{1,0} & \widetilde{w}^{(2)}_{1,1}  \\
\widetilde{w}^{(2)}_{2,0} & \widetilde{w}^{(2)}_{2,1}  \\
\end{array}
\right].
\end{IEEEeqnarray}

Let $\sigma_0,\sigma_1,\sigma_2,\sigma_3,\sigma_4,\sigma_5,\sigma_6,\sigma_7,\sigma_8,\sigma_9,\sigma_{10}$ represent $11$ pairwise distinct elements from the finite field $\mathbb{F}_q$, then the group of coding parameters $\{\alpha_n,\beta_{i,k}:n\in[0:8),i\in[0:3),k\in[0:4)\}$ is set as
\begin{IEEEeqnarray}{llll}
\alpha_0=\sigma_0, \quad\quad  & \alpha_1=\sigma_1,\quad\quad & \alpha_2=\sigma_2,\quad\quad  & \alpha_3=\sigma_3, \label{example:feasible:parameter:0} \\
\alpha_4=\sigma_4, & \alpha_5=\sigma_5, & \alpha_6=\sigma_6, & \alpha_7=\sigma_7, \label{example:feasible:parameter:2} \\
\beta_{0,0}=\sigma_8, & \beta_{0,1}=\sigma_9, & \beta_{0,2}=\sigma_{0}, & \beta_{0,3}=\sigma_1, \label{example:feasible:parameter:3} \\
\beta_{1,0}=\sigma_9, & \beta_{1,1}=\sigma_{10}, & \beta_{1,2}=\sigma_0, & \beta_{1,3}=\sigma_1, \label{example:feasible:parameter:4}  \\
\beta_{2,0}=\sigma_{10}, & \beta_{2,1}=\sigma_8, & \beta_{2,2}=\sigma_0, & \beta_{2,3}=\sigma_1. \label{example:feasible:parameter:1}
\end{IEEEeqnarray}

Since the elements $\beta_{i,0},\beta_{i,1},\beta_{i,2},\beta_{i,3}$ are distinct from each other, they can be used as interpolation points to construct the storage function $\tilde{f}_i^{(m)}(x)$ as a Lagrange polynomial of degree $3$, for any given $m\in[0:3)$ and $i\in[0:3)$, given by
\begin{IEEEeqnarray}{rCl}
\tilde{f}_i^{(m)}(x)&=&
\widetilde{w}_{i,0}^{(m)}\cdot\underbrace{\frac{(x-\beta_{i,1})(x-\beta_{i,2})(x-\beta_{i,3})}{(\beta_{i,0}-\beta_{i,1})(\beta_{i,0}-\beta_{i,2})(\beta_{i,0}-\beta_{i,3})}}_{=b_{i,0}(x)}
+\widetilde{w}_{i,1}^{(m)}\cdot\underbrace{\frac{(x-\beta_{i,0})(x-\beta_{i,2})(x-\beta_{i,3})}{(\beta_{i,1}-\beta_{i,0})(\beta_{i,1}-\beta_{i,2})(\beta_{i,1}-\beta_{i,3})}}_{=b_{i,1}(x)} \notag \\
&&+\widetilde{z}_{i,2}^{(m)}\cdot\underbrace{\frac{(x-\beta_{i,0})(x-\beta_{i,1})(x-\beta_{i,3})}{(\beta_{i,2}-\beta_{i,0})(\beta_{i,2}-\beta_{i,1})(\beta_{i,2}-\beta_{i,3})}}_{=b_{i,2}(x)}
+\widetilde{z}_{i,3}^{(m)}\cdot\underbrace{\frac{(x-\beta_{i,0})(x-\beta_{i,1})(x-\beta_{i,2})}{(\beta_{i,3}-\beta_{i,0})(\beta_{i,3}-\beta_{i,1})(\beta_{i,3}-\beta_{i,2})}}_{=b_{i,3}(x)},\IEEEeqnarraynumspace \label{example:feasible:storage:poly}
\end{IEEEeqnarray}
which satisfies 
\begin{IEEEeqnarray}{c}
\tilde{f}_i^{(m)}(\beta_{i,0})=\widetilde{w}_{i,0}^{(m)}, \quad
\tilde{f}_i^{(m)}(\beta_{i,1})=\widetilde{w}_{i,1}^{(m)}, \quad
\tilde{f}_i^{(m)}(\beta_{i,2})=\widetilde{z}_{i,2}^{(m)}, \quad
\tilde{f}_i^{(m)}(\beta_{i,3})=\widetilde{z}_{i,3}^{(m)}. \label{example:feasible:storage:poly:2}
\end{IEEEeqnarray}
Similarly, since the elements $\{\alpha_0,\alpha_1,\beta_{j,k}:j\in\mathcal{R}\}$ are pairwise distinct, the query function $\widetilde{q}_{i,k}^{(m,\mathcal{R})}(x)$ can be constructed as a Lagrange interpolating polynomial of degree $r+1$, for any subset $\mathcal{R}\subseteq[0:3)$ of size $1\leq r\leq 3$ and any $i\in\mathcal{R},k\in[0:2),m\in[0:3)$, given by
\begin{IEEEeqnarray}{rCl}
\widetilde{q}_{i,k}^{(m,\mathcal{R})}(x)&=&\widetilde{z}_{i,k,0}^{(m,\mathcal{R})}\cdot\underbrace{\frac{x-\alpha_{1}}{\alpha_{0}-\alpha_{1}}\cdot\bigg(\prod\limits_{j\in\mathcal{R}}\frac{x-\beta_{j,k}}{\alpha_{0}-\beta_{j,k}}\bigg)}_{=v_{i,k,0}^{(\mathcal{R})}(x)}
+\widetilde{z}_{i,k,1}^{(m,\mathcal{R})}\cdot\underbrace{\frac{x-\alpha_{0}}{\alpha_{1}-\alpha_{0}}\cdot\bigg(\prod\limits_{j\in\mathcal{R}}\frac{x-\beta_{j,k}}{\alpha_{1}-\beta_{j,k}}\bigg)}_{=v_{i,k,1}^{(\mathcal{R})}(x)} \notag\\
&&+\left\{
\begin{array}{@{}ll}
\underbrace{\bigg(\prod\limits_{j\in\mathcal{R}\backslash\{i\}}\frac{x-\beta_{j,k}}{\beta_{i,k}-\beta_{j,k}}\bigg)\cdot\frac{(x-\alpha_{0})(x-\alpha_1)}{(\beta_{i,k}-\alpha_{0})(\beta_{i,k}-\alpha_1)}}_{=v_{i,k,2}^{(\mathcal{R})}(x)},
&\mathrm{if}\; m=\theta \\
0, & \mathrm{otherwise}
\end{array} \right., \label{example:feasible:query:poly}
\end{IEEEeqnarray}
which satisfies
\begin{IEEEeqnarray}{rCll}
\widetilde{q}_{i,k}^{(m,\mathcal{R})}(\beta_{j,k})&=&\left\{
\begin{array}{@{}ll}
1, &\mathrm{if}\,\, m=\theta, j=i\\
0, & \mathrm{otherwise}
\end{array}
\right.,& ~\forall\, j\in\mathcal{R},\label{example:feasible:query:poly:2}\\
\widetilde{q}_{i,k}^{(m,\mathcal{R})}(\alpha_{t})&=&\widetilde{z}_{i,k,t}^{(m,\mathcal{R})},&~\forall\, t\in[0:2). \notag
\end{IEEEeqnarray}

We have now completed the construction of the PIR coding framework based on Lagrange codes.
Next, we verify the feasibility of this PIR coding framework.

\subsubsection*{Verification of Conditions $\widetilde{\mathrm{F}}0$ and $\widetilde{\mathrm{F}}1$}
According to \eqref{example:feasible:parameter:0}--\eqref{example:feasible:parameter:1} and Lemma \ref{g-cauchy matrix}, it is not difficult to verify that the constructed PIR coding framework satisfies the conditions $\widetilde{\mathrm{F}}$0 and $\widetilde{\mathrm{F}}$1.

\subsubsection*{Verification of Conditions $\widetilde{\mathrm{F}}2$ and $\widetilde{\mathrm{F}}3$}
Regarding the conditions $\widetilde{\mathrm{F}}$2 and $\widetilde{\mathrm{F}}$3, it suffices to demonstrate that for any subset $\Delta\subseteq[0:8)$ of size $5+r-d$ and any subset $\mathcal{D}\subseteq\mathcal{R}$ of size $d$ with $d\in[0:r)$, the partial file $\widetilde{\mathbf{W}}^{(\theta)}(\mathcal{R},:)$ can be decoded from the server responses $\widetilde{\mathcal{A}}_{\Delta}^{(\theta,\mathcal{R})}$ along with the $d$ row vectors $\widetilde{\mathbf{W}}^{(\theta)}(\mathcal{D},:)$ of $\widetilde{\mathbf{W}}^{(\theta)}(\mathcal{R},:)$, as given by
\begin{IEEEeqnarray}{c}\label{Framework:subresponse}
\left\{ \widetilde{\mathcal{A}}_{\Delta}^{(\theta,\mathcal{R})},\widetilde{\mathbf{W}}^{(\theta)}(\mathcal{D},:) \right\}
=\left\{ \widetilde{A}_{n,k}^{(\theta,\mathcal{R})},\widetilde{w}^{(\theta)}_{i,k}:n\in\Delta,k\in[0:2),i\in\mathcal{D} \right\}, 
\end{IEEEeqnarray}
where the equation follows from \eqref{Example:PIR:answer:1} and \eqref{example:feasible:PIR}, and the sub-response $\widetilde{A}_{n,k}^{(\theta,\mathcal{R})}$ is given in \eqref{Example:PIR:answer:2} as follows:
\begin{IEEEeqnarray}{c}
\widetilde{A}_{n,k}^{(\theta,\mathcal{R})}=\sum\limits_{m\in[0:3)}\sum\limits_{i\in\mathcal{R}}\widetilde{q}_{i,k}^{(m,\mathcal{R})}(\alpha_n)\cdot \tilde{f}_{i}^{(m)}(\alpha_n). \notag
\end{IEEEeqnarray}

Let $\widetilde{A}_{k}^{(\theta,\mathcal{R})}(x)$ represent an answer polynomial, given by
\begin{IEEEeqnarray}{c}
\widetilde{A}_{k}^{(\theta,\mathcal{R})}(x)=\sum\limits_{m\in[0:3)}\sum\limits_{i\in\mathcal{R}}\widetilde{q}_{i,k}^{(m,\mathcal{R})}(x)\cdot \tilde{f}_{i}^{(m)}(x), \notag
\end{IEEEeqnarray}
which has a degree of $4+r$ by \eqref{example:feasible:storage:poly} and \eqref{example:feasible:query:poly}. Obviously, the sub-response $\widetilde{A}_{n,k}^{(\theta,\mathcal{R})}$ is an evaluation of $\widetilde{A}_{k}^{(\theta,\mathcal{R})}(x)$ at $x=\alpha_n$. In particular, by \eqref{example:feasible:storage:poly:2} and \eqref{example:feasible:query:poly:2}, evaluating $\widetilde{A}_{k}^{(\theta,\mathcal{R})}(x)$ at $x=\beta_{i,k}$ for any $i\in\mathcal{R}$ and $k\in[0:2)$ yields
\begin{IEEEeqnarray}{c}
\widetilde{A}_{k}^{(\theta,\mathcal{R})}(\beta_{i,k})=\sum\limits_{m\in[0:3)}\sum\limits_{j\in\mathcal{R}}\widetilde{q}_{j,k}^{(m,\mathcal{R})}(\beta_{i,k})\cdot \tilde{f}_{j}^{(m)}(\beta_{i,k})
=\tilde{f}_{i}^{(\theta)}(\beta_{i,k}) 
=\widetilde{w}_{i,k}^{(\theta)}.\label{Framework:lagrange:decoding}
\end{IEEEeqnarray}
Therefore, the values in \eqref{Framework:subresponse} can be equivalently written as  
\begin{IEEEeqnarray}{c}\notag
\left\{ \widetilde{A}_{k}^{(\theta,\mathcal{R})}(\alpha_n),\widetilde{A}_{k}^{(\theta,\mathcal{R})}(\beta_{i,k}):n\in\Delta,k\in[0:2),i\in\mathcal{D} \right\}. 
\end{IEEEeqnarray}
Since the $5+r$ elements $\{\alpha_n,\beta_{i,k}:n\in\Delta,i\in\mathcal{D}\}$ are pairwise distinct  by \eqref{example:feasible:parameter:0}--\eqref{example:feasible:parameter:1} for any $k\in[0:2)$,
we can interpolate the polynomial $\widetilde{A}_{k}^{(\theta,\mathcal{R})}(x)$ of degree $4+r$ from the values $\{\widetilde{A}_{k}^{(\theta,\mathcal{R})}(\alpha_n),\widetilde{A}_{k}^{(\theta,\mathcal{R})}(\beta_{i,k})\}_{n\in\Delta,i\in\mathcal{D}}$. 
Then, evaluating $\widetilde{A}_{k}^{(\theta,\mathcal{R})}(x)$ at $x=\beta_{i,k}$ for all $i\in\mathcal{R}$ and $k\in[0:2)$ allows us to decode the partial file $\widetilde{\mathbf{W}}^{(\theta)}(\mathcal{R},:)$ by \eqref{Framework:lagrange:decoding} and \eqref{example:feasible:PIR}.

Consequently, in this example, we demonstrate that the PIR coding framework based on Lagrange codes is feasible over a finite filed of size $q=11$.



\subsection{Proof of Lemma \ref{lemma:query:array}}\label{proof:C0-C3}
From \eqref{array:0}, all elements of the subarray $\mathbf{U}^{0}$ are integers in the range $0$ to $P-1$. Moreover, by \eqref{array:h:1} and \eqref{array:h:2}, all elements of the subarray $\mathbf{U}^{h}$ are assigned values from both the special symbol ``*'' and the elements 
of the previous subarrays $\mathbf{U}^{0},\mathbf{U}^{1},\ldots,\mathbf{U}^{h-1}$ for any $h\in[1:\lambda)$. So, it is easy to iteratively prove that C0 is satisfied.

Next, we proceed to prove C1.
By \eqref{array:0} again, the $\lambda$ elements in each column of the subarray $\mathbf{U}^{0}$ are all integers, and particularly $(u_{i,j}^{0})_{\lambda}=i$ for any $i\!\in\![0\!:\!\lambda)$ and $j\!\in\![0\!:\!\Gamma^0)$. This completes the proof of C1 in the case of $h=0$.

By \eqref{query:array:subarray}-\eqref{number:column}, the elements of the array $\mathbf{U}$ and the subarrays $\mathbf{U}^{0},\mathbf{U}^{1},\ldots,\mathbf{U}^{\lambda-1}$ satisfy the following relationship:
\begin{IEEEeqnarray}{c}\label{array:subarray}
u_{i,j}=\left\{\begin{array}{@{}ll}
u_{i,j}^{0},   &\text{if}~j\in[0:\frac{P}{\lambda}) \\
u_{i,j-\frac{P}{\lambda-h+1}}^{h},&\text{if}~j\in[\frac{P}{\lambda-h+1}:\frac{P}{\lambda-h})~\text{for some $h\in[1:\lambda)$}
\end{array}\right.,\quad\forall\,i\in[0:\lambda),j\in[0:P).
\end{IEEEeqnarray}
Then, in the $i$-th row of the subarray $\mathbf{U}^{h}$ for any $h\in[1:\lambda)$ and $i\in[0:\lambda)$, 
the elements with column indices in $\{(i+r)_{\lambda}+s\cdot\lambda:r\in[h:\lambda),s\in[0:\frac{\Gamma^h}{\lambda})\}$ satisfy
\begin{IEEEeqnarray}{rCl}
&&\Big\{u^{h}_{i,(i+r)_{\lambda}+s\cdot\lambda}:r\in[h:\lambda),s\in[0:\frac{\Gamma^h}{\lambda})\Big\} \notag\\
&\overset{(a)}{=}&\Big\{u_{i,(i+h-1)_{\lambda}+t\cdot\lambda}:t\in[0:\frac{P}{\lambda(\lambda-h+1)})\Big\} \notag\\
&=&\Big\{u_{i,(i+h-1)_{\lambda}+t\cdot\lambda}:t\in[0:\frac{P}{\lambda\lambda})\Big\}
\bigcup\Big\{u_{i,(i+h-1)_{\lambda}+t\cdot\lambda}:t\in[\frac{P}{\lambda(\lambda-h'+1)}:\frac{P}{\lambda(\lambda-h')}),h'\in[1:h)\Big\} \notag\\
&\overset{(b)}{=}&\Big\{u_{i,(i+h-1)_{\lambda}+t\cdot\lambda}^{0}:t\!\in\![0\!:\!\frac{P}{\lambda\lambda})\Big\}
\bigcup\Big\{u^{h'}_{i,(i+h-1)_{\lambda}+t\cdot\lambda-\frac{P}{\lambda-h'+1}}:t\!\in\![\frac{P}{\lambda(\lambda\!-\!h'\!+\!1)}:\frac{P}{\lambda(\lambda\!-\!h')}),h'\in[1:h)\Big\} \notag\\
&=&\Big\{u_{i,(i+h-1)_{\lambda}+s\cdot\lambda}^{0}:s\in[0:\frac{P}{\lambda\lambda})\Big\}\bigcup
\Big\{u_{i,(i+h-1)_{\lambda}+s\cdot\lambda}^{h'}:s\in [0:\frac{P}{\lambda(\lambda-h')(\lambda-h'+1)}),h'\in[1:h)\Big\}\IEEEeqnarraynumspace \notag\\
&\overset{(c)}{=}&
\Big\{u_{i,(i+h-1)_{\lambda}+s\cdot\lambda}^{h'}:s\in [0:\frac{\Gamma^{h'}}{\lambda}),h'\in[0:h)\Big\},\IEEEeqnarraynumspace \label{Proof:C1}
\end{IEEEeqnarray}
where $(a)$ is due to  \eqref{array:h:2} and \eqref{number:column}, $(b)$ follows from \eqref{array:subarray}, and $(c)$ follows by \eqref{number:column} again.


Since the elements $\{u^{0}_{i,(i+r)_{\lambda}+s\cdot\lambda}:r\in[0:\lambda),s\in[0:\frac{\Gamma^0}{\lambda})\}$ of the subarray $\mathbf{U}^{0}$ are all integers, the elements $\{u^{1}_{i,(i+r)_{\lambda}+s\cdot\lambda}:r\in[1:\lambda),s\in[0:\frac{\Gamma^1}{\lambda})\}$ of the subarray $\mathbf{U}^{1}$ are also all integers by setting $h=1$ in \eqref{Proof:C1}, for any $i\in[0:\lambda)$. Furthermore, by iteratively setting $h=2,3,\ldots,\lambda-1$ in \eqref{Proof:C1}, it is not difficult to prove that the elements $\{u^{h}_{i,(i+r)_{\lambda}+s\cdot\lambda}:r\in[h:\lambda),s\in[0:\frac{\Gamma^h}{\lambda})\}$ of the subarray $\mathbf{U}^{h}$ are all integers.
Generally, for any $h\in[1:\lambda)$, the elements $\{u^{h}_{i,(i+r)_{\lambda}+s\cdot\lambda}:i\in[0:\lambda),r\in[h:\lambda),s\in[0:\frac{\Gamma^h}{\lambda})\}$ of the subarray $\mathbf{U}^{h}$ are all integers.


Then, for any given $j\in[0:\Gamma^h)$ and $t\in[h:\lambda)$, we find that $u^{h}_{(j-t)_{\lambda},j}\in\{u^{h}_{i,(i+r)_{\lambda}+s\cdot\lambda}:i\in[0:\lambda),r\in[h:\lambda),s\in[0:\frac{\Gamma^h}{\lambda})\}
$ due to the relationship $(i+r)_{\lambda}+s\cdot\lambda=j$ in the case of $i=(j-t)_{\lambda}$, $r=t$, and $s=\lfloor\frac{j}{\lambda}\rfloor$.
Therefore, in the $j$-th column of the subarray $\mathbf{U}^{h}$, the $\lambda-h$ elements $\{u^{h}_{(j-t)_{\lambda},j}:t\in[h:\lambda)\}$  are all integers,
while the remaining $h$ elements $\{u^{h}_{(j-t)_{\lambda},j}:t\in[0:h)\}$ are all the special symbols ``*'' by \eqref{U0:cons}. 
Accordingly, for any $h\in[0:\lambda)$ and $j\in[0:\Gamma^h)$, we have 
\begin{IEEEeqnarray}{c}\label{Array:integer:element}
\mathcal{R}^h_j=\{u^{h}_{(j-t)_{\lambda},j}:t\in[h:\lambda)\}.
\end{IEEEeqnarray}



Recall from \eqref{array:0} that all elements in the $i$-th row of the subarray $\mathbf{U}^{0}$ are equal to $i$ modulo $\lambda$ for any $i\in[0:\lambda)$, i.e., $(u_{i,j}^{0})_{\lambda}=i$ for all $j\in[0:\Gamma^0)$.
Since the integer elements $\{u^{h}_{i,(i+r)_{\lambda}+s\cdot\lambda}:r\in[h:\lambda),s\in[0:\frac{\Gamma^h}{\lambda})\}$ in the $i$-th row of the subarray $\mathbf{U}^{h}$ are assigned values from the integer elements in the \emph{same} row $i$ of the previous subarrays $\mathbf{U}^{0},\mathbf{U}^{1},\ldots,\mathbf{U}^{h-1}$ by \eqref{array:h:2} and \eqref{Proof:C1}, a simple iterative proof shows that the integer elements $\{u^{h}_{i,(i+r)_{\lambda}+s\cdot\lambda}:r\in[h:\lambda),s\in[0:\frac{\Gamma^h}{\lambda})\}$ are all equal to $i$ modulo $\lambda$, for any $i\in[0:\lambda)$ and $h\in[1:\lambda)$. 
Therefore, for any given $j\in[0:\Gamma^h)$ and $t\in[h:\lambda)$, the integer element $u^{h}_{(j-t)_{\lambda},j}$ satisfies $(u^{h}_{(j-t)_{\lambda},j})_{\lambda}=(j-t)_{\lambda}$ due to $u^{h}_{(j-t)_{\lambda},j}\in\{u^{h}_{i,(i+r)_{\lambda}+s\cdot\lambda}:i=(j-t)_{\lambda},r\in[h:\lambda),s\in[0:\frac{\Gamma^h}{\lambda})\}$. 
Then, by \eqref{Array:integer:element}, we have $\widetilde{\mathcal{R}}^{h}_j=\{(j-t)_{\lambda}:t\in[h:\lambda)\}$, which ensures $|\mathcal{R}_j^{h}|=|\widetilde{\mathcal{R}}^{h}_j|=\lambda-h$.
This completes the proof of C1 in the case of $h\in[1:\lambda)$.

It is straightforward to prove C2 by \eqref{array:0}.
We now begin the proof of C3. 
For any given $h\in[1:\lambda)$,
\begin{IEEEeqnarray}{rCl}
\bigcup\limits_{j\in[0:\Gamma^h)}\mathcal{R}^h_j&\overset{(a)}{=}&\Big\{u^{h}_{(j-t)_{\lambda},j}:t\in[h:\lambda),j\in[0:\Gamma^h)\Big\} \notag\\
&\overset{(b)}{=}&\Big\{u^{h}_{i,(i+r)_{\lambda}+s\cdot\lambda}:i\in[0:\lambda),r\in[h:\lambda),s\in[0:\frac{\Gamma^{h}}{\lambda})\Big\} \notag\\
&\overset{(c)}{=}&\Big\{u_{i,(i+h-1)_{\lambda}+s\cdot\lambda}^{h'}:i\in[0:\lambda),s\in[0:\frac{\Gamma^{h'}}{\lambda}),h'\in[0:h)\Big\},\notag
\end{IEEEeqnarray}
where $(a)$ follows from \eqref{Array:integer:element}, $(b)$ is due to the fact that the set
$\{u^{h}_{(j-t)_{\lambda},j}:t\in[0:h),j\in[0:\Gamma^h)\}$
is equivalent to
$\{u^{h}_{i,(i+r)_{\lambda}+s\cdot\lambda}:i\in[0:\lambda),r\in[0:h),s\in[0:\frac{\Gamma^{h}}{\lambda})\}$ by \eqref{U0:cons}-\eqref{array:h:1},
and $(c)$ follows by \eqref{Proof:C1}.
In other words, for any given $h\in[0:\lambda-1)$ and $r\in[h+1:\lambda)$,
\begin{IEEEeqnarray}{c}
\bigcup\limits_{j\in[0:\Gamma^r)}\mathcal{R}^r_j=\Big\{u_{i,(i+r-1)_{\lambda}+s\cdot\lambda}^{h'}:i\in[0:\lambda),s\in[0:\frac{\Gamma^{h'}}{\lambda}),h'\in[0:r)\Big\}.\notag
\end{IEEEeqnarray}
Then, for any $j\in[0:\Gamma^{h})$ and $r\in[h+1:\lambda)$, the integer element $u^{h}_{(j-r+1)_{\lambda},j}$ in the $j$-th column of the subarray $\mathbf{U}^{h}$ satisfies
\begin{IEEEeqnarray}{c}
u^{h}_{(j-r+1)_{\lambda},j}\in\left\{u_{i,(i+r-1)_{\lambda}+s\cdot\lambda}^{h'}:i=(j-r+1)_{\lambda},s=\big\lfloor \frac{j}{\lambda}\big\rfloor,h'=h\right\}\subseteq\bigcup\limits_{j\in[0:\Gamma^r)}\mathcal{R}^r_j.\notag
\end{IEEEeqnarray}
This completes the proof of C3.
Consequently, the constructed query array satisfies the conditions C0--C3.

\end{appendix}


\bibliographystyle{ieeetr}
\bibliography{reference.bib}

\end{document}